\documentclass[reqno]{tpms-lmyversion}

\usepackage[hyphens]{url}
\usepackage{enumitem}
\usepackage{amssymb}
\usepackage{soul,color}
\usepackage{graphicx}
\usepackage{epstopdf}
\usepackage{bbm}
\usepackage{dsfont}
\usepackage{amsfonts}
\usepackage{subfig}

\newtheorem{thm}{Theorem}[section]
\newtheorem{lem}{Lemma}[section]
\newtheorem{prop}{Proposition}[section]
\newtheorem{cor}{Corollary}[section]

\theoremstyle{definition}
\newtheorem{defn}{Definition}[section]

\theoremstyle{remark}
\newtheorem{rem}{Remark}[section]

\newcounter{MyStepNo}
\DeclareRobustCommand{\mystepTwo}[1]{%
   \refstepcounter{MyStepNo}%
   \theMyStepNo\label{#1}}

\newcounter{MyStepNoTwo}

\numberwithin{equation}{section}

\begin{document}

\title{Mini-Flash Crashes, Model Risk, and Optimal Execution}



\author{Erhan Bayraktar}
\address{Department of Mathematics, University of Michigan, Ann Arbor, Michigan 48109}
\email{erhan@umich.edu}
\thanks{Erhan Bayraktar is supported by the National Science Foundation under grant DMS-1613170 and the
Susan M. Smith Professorship. Alexander Munk is supported by a Rackham Predoctoral Fellowship. 
We gratefully acknowledge E. Jerome Benveniste, Charles-Albert LeHalle, Sebastian Jaimungal as well as the participants in the 
de Finetti Risk Seminar (jointly organized by Bocconi University and the University of Milan), 
LUISS's Mathematical Economics and Finance Seminar,
University of Michigan's Financial/Actuarial Mathematics Seminar,
and the Fields Institute's Quantitative Finance Seminar, for their valuable suggestions.}


\author{Alexander Munk}
\address{Department of Mathematics, University of Michigan, Ann Arbor, Michigan 48109}
\email{amunk@umich.edu}





%

\begin{abstract}

Oft-cited causes of mini-flash crashes include
human errors, 
endogenous feedback loops,
the nature of modern liquidity provision.
We develop a mathematical model which 
captures aspects of these
explanations. Empirical features of recent mini-flash 
crashes are present 
in our framework. 
For example, 
there are periods when 
no such events will occur. 
If they do, even just before their onset, 
market participants may not know with 
certainty 
that a disruption will unfold. 
Our mini-flash crashes can materialize
in both low and high trading volume
environments and may be accompanied by 
a partial synchronization in
order submission.

Instead of adopting a classically-inspired
equilibrium approach, we borrow 
ideas from the 
optimal execution literature.
Each of our agents begins with 
beliefs about how 
his own trades impact prices
and how prices would move
in his absence. 
They, along with 
other market participants, 
then  
submit orders which
are executed at a common venue. 
Naturally, this leads us to
explicitly distinguish between 
how prices actually 
evolve and our agents' opinions. 
In particular, every agent's 
beliefs will be expressly incorrect.

\end{abstract}

\maketitle

\section{Overview}\label{intro sect}

 Amidst the violent market disruption on May 6, 2010, the infamous Flash Crash,
\begin{quote}
``Over 20,000 trades across more than 300 securities were executed at prices more than 60\% 
away from their values just moments before. 
Moreover, many of these {\it trades were executed at prices of a penny or
less, or as high as \$100,000,} before prices of those securities returned to their `pre-crash' levels'' (\cite{sec+cftc+report}).
\end{quote}
Today, this particular event remains so memorable due to its
remarkable scale.

In fact, lesser versions of the Flash Crash, or {\it mini-flash crashes}, happen quite often. 
Anecdotal evidence suggests that
there may be over a dozen every day (\cite{cnn+num+crashes}). 
A rigorous empirical analysis uncovered ``18,520 crashes and spikes
with durations less than 1,500 ms'' in stock prices from 2006 through 2011 (\cite{nat+paper}). 
The exhaustive documentation on Nanex LLC's ``NxResearch'' site offers further corroboration 
as well (\cite{nanex+main+research}).

Why do such phenomena occur? Several answers have been proposed. 
Roughly, most point to 
human errors, 
endogenous feedback loops,
the nature of modern liquidity provision. 
These ideas can be viewed as different ways to
rationalize how an 
 extreme (local or global) dislocation in supply and demand
can arise in modern markets. 
Our contribution
is the development of a model 
which captures this.
Our model also appears to exhibit features of
 historical mini-flash crashes. 
For instance, there are 
periods in which extreme
 price moves will not manifest. 
 If they do, accompanying trade 
 volumes can be high or low. 
 Some market participants 
 may partially synchronize their 
 trading during a mini-flash crash. 
 Our agents may not know that
 a mini-flash crash is about to begin
 even just before its onset.

 Our results seem to be aligned 
 with intuitive expectations as well.
 For example, our mini-flash crashes
can begin if some of our agents are too
 uncertain about their initial beliefs, 
 inaccurate in their understanding of price dynamics, 
 slow to update their models and objectives, 
 or willing to take on risk.

We construct our model beginning with 
a finite population of agents trading in a single risky asset, 
each of whom must decide how to act 
 based upon his own preferences, beliefs, 
 and observations. 
 Our specifications are drawn from
 ideas in the price impact and optimal execution literature 
 and are given in Sections \ref{inv model subsec} - \ref{obj fxn subsec}.

 We imagine that our agents' orders are submitted to a single venue, 
 where they are executed together 
 with trades from other (unmodeled) market participants. 
This naturally compels us to
make an explicit distinction between 
how the risky asset price actually evolves
and our agents' beliefs about its future evolution (see Section \ref{exc price sect FIN}).

Since we view our agents as 
simultaneously
solving their own optimal execution 
problems, we avoid certain strong 
assumptions that would have been 
implicitly needed, if we had used 
a classical equilibrium-based approach. We use a similar framework to \cite{Cardaliaguet2018, 2018arXiv180304094C} in that our players' interaction with the rest of the world (in addition to each other) is given by a price impact model.
An additional consequence is that 
we precisely describe the errors in our 
agents' beliefs.
Potentially, each agent could be wrong
both about how
his trades affect prices and 
how prices would move in his absence.

To the best of our knowledge, 
this general setup appears to be a new
paradigm for modeling 
heterogeneous agent systems
in the contexts of optimal execution
and mini-flash crashes.

We are ready to begin presenting 
our work in detail. 
We highlight key background material and our paper's contributions 
in relation to it in Section \ref{back contrib sect}.
Our agents and their beliefs
are described in Section \ref{mod details sec}.
 We describe the correct dynamics 
 of the risky asset's price in Section \ref{exc price sect FIN}. 
 General analysis of the dynamical system when our agents 
 act as prescribed by Section \ref{mod details sec}
 but prices actually move as in Section \ref{exc price sect FIN}
 are given in Section \ref{true price dyn section}. 
 Using the material in Section \ref{true price dyn section},
we obtain explicit characterization of mini-flash crashes when \emph{uncertain agents} are semi-symmetric in
 \ref{exmp sect}. Numerical examples illustrating our main results are given in Section~\ref{sec:numerical}.
 Our longer proofs are contained in Appendices 
 \ref{sect mod details sec proofs app} - \ref{exmp sect PRF APP}.

\section{Background \& Contributions}\label{back contrib sect}

In this section, we  
clarify our contributions and explain
how they fit into the current literature. 
We already mentioned that existing theories 
on the causes of
mini-flash crashes could be 
viewed as falling into one of five categories (see Section \ref{intro sect}). 
Here are further details.

\begin{enumerate}[label=\roman*)]

\item \label{human errors crash}  
Human errors (and, relatedly, improper risk management) are among the 
most commonly cited causes of mini-flash crashes (\cite{flash+pound+algos}, \cite{bloom+fat+fing}, \cite{sec+merrill}). 
The SEC claims that the majority of mini-flash crashes originate from such sources, in fact (\cite{bloom+fat+fing}). 
When we read about 
{\it fat finger trades}, {\it rogue algorithms}, or {\it glitches}
in the media, typically human errors are indirectly responsible.
For example, due to a bug in the 
systems at the Tokyo Stock Exchange
and a typo in a trade submitted by Mizuho Securities, 
the share price of the recruitment agency J-Com
fell in minutes from $\yen$672,000
to $\yen$572,000
on December 8, 2005 (\cite{miz+jcom+crash}).

\item \label{same direct crash}
Mini-flash crashes may be caused by 
the rapid, endogenous formation of positive feedback loops
(\cite{nanex+flash+crash+sum}, \cite{nat+paper}, \cite{PhysRevE+pred+flash}, 
\cite{flash+pound+algos}, \cite{huang+wang+2008}, \cite{rock+clock}). 
As Johnson et al. put it, 
\begin{quote}
``Crowds of agents frequently converge on the same
strategy and hence simultaneously flood the market with the same
type of order, thereby generating the frequent extreme price-change
events'' (\cite{nat+paper}). 
\end{quote}
A separate empirical study on the Flash Crash of May 6, 2010, specifically, 
determined that at its peak, ``95\% of the trading was due to endogenous triggering effects''
(\cite{PhysRevE+pred+flash}).

\item \label{modern liq provision}

The nature of liquidity provision in modern markets is thought to 
cause some mini-flash crashes
(\cite{kirilenko2016flash}, \cite{easley+flash+crash}, \cite{dennis+erroneous}, 
\cite{hft+mini+flash+crash}, \cite{news+vol}, \cite{gayduk1}, \cite{easley+flow+defn}).
Today, the majority of liquidity is provided by participants that are
free from formal market-making obligations (\cite{dennis+erroneous}).
In particular, they can instantly vanish, effectively taking one 
or both sides of the order book at some venue with them. 
A mini-flash crash can arise either directly as bid-ask spreads
blow out or indirectly when a market order (of any size)
tears through a nearly empty collection of limit orders. 
Such a phenomenon has been called {\it fleeting liquidity} 
and may have contributed to the occurrence of 
38\% of mini-flash 
crashes from 2006 - 2011 (\cite{hft+mini+flash+crash}).

This proposed explanation is deeply intertwined 
with a crucial empirical observation:
Mini-flash crashes occur in both 
high and low trading volume regimes. For instance, the
trading volume during the 30s mini-flash crash 
of ``WisdomTree LargeCap'' Growth Fund 
on November 27, 2012
was nearly eight times the average daily trading volume 
for this security (\cite{sec+merrill}). 
The empirical study by Florescu et al. 
offers extensive evidence that mini-flash crashes 
often occur during low trading volume periods as well.

\end{enumerate}

Aspects of
(\ref{human errors crash}, 
(\ref{same direct crash}, 
and (\ref{modern liq provision}
are reflected in our work. 
For example, the human error theory arises
in each of the following ways:

\begin{enumerate}[label=\alph*)]

\item \label{first human error}
Every agent
believes that 
a mini-flash crash is a null event. 
On the contrary, there are cases in which 
one will occur almost surely
(see Theorems \ref{SS main blow up lem FIN} and 
\ref{no Inv expl lem SS FIN}).

\item \label{first human error B}
Every agent thinks that his trades
affect prices through specific 
temporary and permanent price impact
coefficients (see Section \ref{inv model subsec}). 
His estimates for these parameters might be 
wrong (see Section \ref{exc price sect FIN}).

\item \label{first human error C}
Every agent's trades may also indirectly impact prices 
by inducing others to make different 
decisions than they would otherwise 
(see Section \ref{obj fxn subsec} and Section \ref{exc price sect FIN}). 
This potential effect
is not modeled by our agents (see Section \ref{inv model subsec}). 
More generally, even if we have a single agent in our setup
trading with other unspecified market participants, 
the parameters in his fundamental value 
model might be inaccurate (see Section \ref{inv model subsec}
and Section \ref{exc price sect FIN}).

\item  \label{first human error D}  
No agent revises the general class of his beliefs, 
admissible strategies, 
or objectives on our time horizon
(see Sections \ref{inv model subsec} - Section \ref{obj fxn subsec}).
In some cases, a mini-flash crash will not occur if this period
is fairly short but will if it is too long (see Lemmas \ref{det A no root no blow up}
and \ref{te study lem SS FIN}).

\item \label{first human error E}
Every agent is averse to 
his position's apparent volatility risks 
(see Section \ref{obj fxn subsec}). 
In some cases, there will be no 
mini-flash crash when our agents are sufficiently 
averse to these risks; otherwise, there will be one
(see Lemmas \ref{det A no root no blow up}
and \ref{te study lem SS FIN}).

\item \label{first human error F}
Every agent has the opportunity to 
update the drift parameter in 
his price model based upon his
observations (see Section \ref{inv model subsec}). 
In some cases, a mini-flash crash will unfold 
because our agents are too easily persuaded 
to revise their priors 
(see Lemmas \ref{det A no root no blow up}
and \ref{te study lem SS FIN}).

\item \label{last human error}
Every agent has a model for how prices 
are affected by the temporary impact 
of trades (see Section \ref{inv model subsec}). 
In some cases, there will be a mini-flash 
crash if our agents sufficiently 
underestimate the role of aggregate temporary impact. 
No such disturbance will occur otherwise. 
Our agents may be more prone to induce 
mini-flash crashes in this way 
when there are many of them
(see Lemmas \ref{det A no root no blow up}
and \ref{te study lem SS FIN}).

\end{enumerate}
Notice that some of our agents' human errors directly 
cause mini-flash crashes, though not all (see Proposition \ref{det A no root no blow up}).
We highlight this observation in Figures \ref{fig: Opt_Inv_NoRoot_MAX} - \ref{fig: ExcPri_NoRoot_MAX}. 
Implicitly, the occasional absence of mini-flash crashes also agrees with (\ref{human errors crash}.
Despite the regularity of these disruptions 
on a market-wide basis, 
individual securities may rarely experience such an event. 
Similarly, traders' models and strategies do
roughly achieve their intended goals much of the time, 
which we observe as well (see Proposition \ref{det A no root no blow up}).

Several key ideas from the 
endogenous feedback loop theory 
are present in our paper. 
For example, if a mini-flash crash does occur, 
it almost surely does so because of 
``endogenous triggering effects.''
Specifically, our mini-flash crashes arise when some of our agents 
buy or sell at faster and faster rates, which they only do 
because they started trading more rapidly in the first place
(see Section \ref{exc price sect FIN} and Lemma \ref{act 1st order ODE lem}). 
As predicted by this theory, some of our agents also ``converge on the same strategy''
during mini-flash crashes: In certain cases, the agents driving these
events all buy or sell together with the same (exploding) growth rate
(see Theorems \ref{SS main blow up lem FIN} and \ref{no Inv expl lem SS FIN}). 
Figures \ref{fig: Opt_Spd_PosEV_Expl_MIN_MAX}, 
\ref{fig: Opt_Spd_NegEV_Expl_MIN_MAX_DOWN}, 
and \ref{fig: Opt_Spd_NegEV_Expl_MIN_MAX_UP}
graphically illustrate this partial synchronization.

We do not explicitly model liquidity providers in our framework, 
as we view our agents as submitting market orders to a single
venue (see Section \ref{exc price sect FIN}). 
We still view our paper as reflecting (\ref{modern liq provision}, 
at least in some sense, since our mini-flash crashes
can be accompanied by both high and low trading volumes
(see Corollary \ref{cert agents opt strats} and 
Theorems \ref{SS main blow up lem FIN} and \ref{no Inv expl lem SS FIN}). 
Visualizations of this point are provided 
in Figures 
\ref{fig: Opt_Inv_PosEV_Expl_MAX}, 
\ref{fig: ExcPri_PosEV_Expl_MIN_MAX},
\ref{fig: Opt_Inv_NegEV_Expl_MIN_MAX_DOWN}, 
\ref{fig: ExcPri_NegEV_Expl_MIN_MAX_DOWN},
\ref{fig: Opt_Inv_NegEV_Expl_MIN_MAX_UP},
and \ref{fig: ExcPri_NegEV_Expl_MIN_MAX_UP}.

\section{Agents and the Execution Price}\label{mod details sec}

In this section we describe our 
agents and their individual 
optimal execution problems. We end this section by introducing the ``true" execution price.

\subsection{Agents' Models}\label{inv model subsec}

Our agents trade continuously by optimally selecting a trading rate from a particular class of admissible strategies. 
To motivate our specifications of their choices and objectives, we first define their models and beliefs.

All trades submitted at time $t$ are executed immediately at the price $S_t^{exc}$. 
At each time $t$, every agent observes the correct value of $S_t^{exc}$. 
No agent knows the true dynamics of the stochastic process $S^{exc}$, though.

In our framework, Agent $j \in \{1,\cdots, N\}$'s models the the risky asset without his own trading by 
\begin{align}\label{fund price model j}
S_{j,t}^{unf} = S_{j,0} + \beta_j t +  W_{j,t}, \qquad t \in \left[ 0, T \right], 
\end{align}
where $W_j$ is a Wiener process and $\beta_j$ is a normally distributed with mean $\mu_j$  and variance $\nu_j^2$ which is independent of the Wiener process. (In what follows we will write down $P_j$ for the probability measure of agent $j$.)
This drift term represents the price pressure that Agent $j$ believes will arise due to the trades of 
(other) institutional investors. 
Agent $j$ approximates the average behavior of uninformed or noise traders using the Brownian term.\footnote{
Almgren \& Lorenz provide further details regarding the interpretation of (\ref{fund price model j}) (\cite{almgren2007adaptive}). 
A possible extension of our work could replace (\ref{fund price model j}) with one of the more recent
models considered in the literature on optimal trading problems with a learning aspect
(\cite{almgren2007adaptive}, \cite{eks+vaic+drift}, \cite{jaim+algo+learn}, \cite{garl+ped+dyn}, \cite{pass+vazq}, 
\cite{frey+gabith+port}, \cite{cola+eksi+shall}).  
} If $\nu_j^2 > 0$, then we call Agent $j$ an {\it uncertain} agent.
If $\nu_j^2 = 0$, then we call Agent $j$ a {\it certain} agent.
Regardless of whether he is certain or uncertain in this sense, 
we will soon see that Agent $j$ can always be viewed as certain about many things, 
e.g, he will not change the 
form of his models, objectives, or admissible strategies on $\left[ 0 , T \right]$. From this 
perspective, one might partially connect our work on mini-flash crashes 
to explanations of longer term financial bubbles based on overconfident investors (\cite{scheink+xion+34}).

Intuitively, Agent $j$'s selection of (\ref{fund price model j}) makes the most sense
when $N$ is large and $T$ is short. 
Notice that Agent $j$ makes no attempt to precisely estimate the number of other market participants, 
nor their individual goals or beliefs. 
That he believes he cannot improve the predictive accuracy of (\ref{fund price model j}) by doing so appears to suggest
that the population of traders is of sufficient size.
Together with the fact that real drifts and volatilities are non-constant, 
(\ref{fund price model j}) only seems even potentially plausible over short periods.

Let $\mathcal{A}_{j}$ be the space of $\mathcal{F}_{j,t}^{unf}$-adapted processes $\theta_{j}$ of trading speeds such that $\theta_{j,\cdot} \left( \omega \right)$ is continuous on $\left[ 0 , T \right]$ for $P_j$-almost surely, 
\begin{equation}\label{admiss strat int}
E^{P_j} \left[ \displaystyle\int_{0}^{T} \theta_{j,t} ^2 \, dt \right] < \infty, 
\end{equation}
and 
\begin{equation}\label{admiss strat term}
x_{j} + \displaystyle\int_{0}^T \theta_{j,t} \, dt = 0  \quad P_{j}-\text{a.s.} 
\end{equation}
For any $\theta_{j} \in \mathcal{A}_{j}$, we denote by 
\begin{equation}\label{x theta def}
X_{j,t}^{\theta_j} = x_j + \displaystyle\int_{0}^{t} \theta_{j,s} \, ds,
\end{equation}
the agent's inventory.

The agent models the execution price
as $S_{j,\theta_j}^{exc}$, which is given by
\begin{align}\label{exc price model j}
S_{j,\theta_{j},t}^{exc} = S_{j,t}^{unf} + \eta_{j,per} \displaystyle\int_{0}^{t} \theta_{j,s} \, ds + \displaystyle\frac{1}{2} \eta_{j,tem} \theta_{j,t}, \qquad t \in \left[ 0, T \right]. 
\end{align}
Agent $j$ has chosen the deterministic positive constants $\eta_{j,per}$ and $\eta_{j,tem}$ in (\ref{exc price model j}) prior to time $t =0$.
Agent $j$ could be viewed as 
taking into account his own effects on the execution price 
via an Almgren-Chriss reduced-form model (\cite{almgren2001optimal}, \cite{almgren1999value}, \cite{almgren2003optimal}). 
$\eta_{j,per}$ would denote Agent $j$'s estimate for his permanent price impact parameter, 
while he would approximate his temporary price impact parameter with $\eta_{j,tem}$.

\subsection{Each Agent's Optimization Problem}\label{obj fxn subsec}

Agent $j$'s objective is to maximize the following objective function:

\begin{align}\label{orig max prob 0}
E^{P_j} \left[ - \displaystyle\int_{0}^T \theta_{j,t} \, S_{j,\theta_{j},t}^{exc}  \, dt  -  \displaystyle\frac{\kappa_{j}}{2} \displaystyle\int_{0}^{T}  \left( X^{\theta_{j}}_{j,t} \right)^2 \, dt \right].
\end{align}
This frequently used criteria balances 
his realized trading revenue and risks associated with delayed liquidation.  Agent $j$ selects the deterministic risk aversion parameter $\kappa_j > 0$
based upon his appetite. 

Let us denote 
\begin{equation} \label{tau j notn}
\tau_{j} \left( t \right) \triangleq  \displaystyle\sqrt{\displaystyle\frac{ \kappa_{j}}{ \eta_{j,tem}}} \left( T - t \right) , \quad t \in \left[ 0 , T \right].
\end{equation}

\begin{lem}\label{1 pl soln acc to play model}

(\ref{orig max prob 0}) has a unique optimizer  $\theta_{j}^{\star} \in \mathcal{A}_{j}$ almost surely.
When $\omega \in \Omega_j$ is chosen such that $W_{j,\cdot} \left( \omega \right)$ is continuous on $\left[ 0 , T \right]$, 
$X^{\theta^\star_j}_j \left( \omega \right)$ satisfies the linear ODE 
\begin{align}\label{x theta star}
\theta^{\star}_{j,t} \left( \omega \right) &= - \displaystyle\sqrt{\displaystyle\frac{ \kappa_{j}}{ \eta_{j,tem}}}  
\coth \left( \tau_{j} \left( t \right) \right) X^{\theta^\star_j}_{j,t}  \left( \omega \right) 
\notag \\
&\qquad + 
\displaystyle\frac{  \tanh \left( \tau_{j} \left(t \right) /2 \right) \left[ \mu_{j}  + \nu_{j}^2 \left( S_{j,t}^{unf}  \left( \omega \right) - S_{j,0}\right) \right]  }{ \displaystyle\sqrt{ \eta_{j,tem} \kappa_{j}}   \left( 1 + \nu_{j}^2 t \right)} 
, \quad t \in \left( 0, T \right)
  \notag \\
X^{\theta^{\star}_j}_{j,0}  \left( \omega \right) &= x_j .
\end{align}


\end{lem}

\begin{proof}
See Appendix \ref{1 pl soln acc to play model PRF SS}.
\end{proof}

\begin{rem}\label{terms in agents strat rem}

The first term in (\ref{x theta star}) arises from our constraint that 
Agent $j$ must liquidate by the terminal time (see (\ref{admiss strat term})).
In fact, the weighting factor 
\begin{equation*}
- \displaystyle\sqrt{\displaystyle\frac{ \kappa_{j}}{ \eta_{j,tem}}}  
\coth \left( \tau_{j} \left( t \right) \right) 
\end{equation*}
tends to $- \infty$ as $t \uparrow T$. 
Intuitively, the reason that
Agent $j$ believes that 
 $X^{\theta^\star_j}_{j,t}$ and $\theta_{j,t}^{\star}$
remain finite as $t \uparrow T$ is that $X^{\theta^\star_j}_{j,t}$ tends very rapidly to zero.

Agent $j$ thinks that he learns about $\beta_j$'s realized value over time, which is captured by the second term 
in (\ref{x theta star}) 
since 
\begin{equation}\label{cond beta raw filt form}
E^{P_j} \left[ \beta_j \big| \mathcal{F}_{j,t}^{unf} \right] = \displaystyle\frac{ \mu_{j}  + \nu_{j}^2 \left( S_{j,t}^{unf} - S_{j,0}\right) }{1 + \nu_{j}^2 t} \quad P_j-\text{a.s.}
\end{equation}
(\cite{lipt+shiry+i}). 
The factor 
\begin{equation*}
\displaystyle\frac{  \tanh \left( \tau_{j} \left(t \right) /2 \right)}{ \displaystyle\sqrt{ \eta_{j,tem} \kappa_{j}}}
\end{equation*}
is bounded by $1 / \displaystyle\sqrt{ \eta_{j,tem} \kappa_{j}}$ and tends to zero as $t \uparrow T$.

The second term may either dampen or amplify the effects of the first. 
Agent $j$ believes that the weighting factors 
reflect that his need to liquidate must 
eventually overwhelm his desire to profit by trading in the direction 
of the risky asset's drift.

\end{rem}

An immediate observation from Lemma \ref{1 pl soln acc to play model} is the following observation for certain agents, which we record as a corollary for ease of referencing.

\begin{cor}\label{cert agents opt strats}

If $\nu_j^2 = 0$, then $X^{\theta^\star_j}_j$ 
does not depend on $S^{unf}_j$. In particular, it is 
deterministic and satisfies the linear ODE 
\begin{align}\label{x theta star DET}
\theta^{\star}_{j,t}  &= - \displaystyle\sqrt{\displaystyle\frac{ \kappa_{j}}{ \eta_{j,tem}}}  
\coth \left( \tau_{j} \left( t \right) \right) X^{\theta^\star_j}_{j,t} 
 + 
\displaystyle\frac{ \mu_{j}     \tanh \left( \tau_{j} \left(t \right) /2 \right)   }{ \displaystyle\sqrt{ \eta_{j,tem} \kappa_{j}}   } 
, \quad t \in \left( 0, T \right)
  \notag \\
X^{\theta^{\star}_j}_{j,0}  &= x_j .
\end{align}

\end{cor}

\subsection{Execution Price}\label{exc price sect FIN}

We now specify
how $S^{exc}$ actually evolves. 
While each agent observes the 
same realized path of this process, in general,
no agent knows the correct dynamics.\footnote{
There is a single trivial case where
this is not true. 
If $N = 1$, $\tilde{\beta} = \beta$, $\nu^2_1 = 0$, 
$\tilde{\eta}_{1,tem} = \eta_{1,tem}$, and $\tilde{\eta}_{1,per} = \eta_{1,per}$, 
our lone agent's model would be exactly right.
} 
An agent's trading decisions 
are entirely determined by 
his beliefs, preferences,
and observations of a single realized path of $S^{exc}$. 

Let
$\left( \tilde{\Omega}, \tilde{\mathcal{F}} , \left\{ \tilde{\mathcal{F}}_t  \right\}_{0 \leq t \leq T} , \tilde{P} \right)$
be a filtered probability space satisfying the usual conditions.
The space is equipped with an $\tilde{\mathcal{F}}_t$-Wiener process under $\tilde{P}$, which we denote by $\tilde{W}$.
We also have the following deterministic real constants:
\begin{equation*}
\tilde{\beta}, \quad S_0, \quad \tilde{\eta}_{1,per}, \dots, \tilde{\eta}_{N,per}, 
\quad \text{and} \quad 
\tilde{\eta}_{1,tem}, \dots, \tilde{\eta}_{N,tem} .
\end{equation*}
$\tilde{\beta}$ can be arbitrary; however, the remaining constants are strictly positive. 

The true execution price $S^{exc}$ under $\tilde{P}$
 is the $\tilde{\mathcal{F}}_t$-adapted process 
\begin{equation}\label{rt true dynamics}
S_t^{exc} = S_{0} +  \tilde{\beta}  t + \displaystyle\sum_{i = 1}^N \tilde{\eta}_{i,per}\left( X^{\theta_{i}^\star}_{i,t} -  x_{i} \right) + \frac{1}{2} \displaystyle\sum_{i = 1}^N \tilde{\eta}_{i,tem} \theta_{i,t}^{\star} + \tilde{W}_t, \qquad  t \in \left[ 0, T \right].
\end{equation}

Equation \eqref{rt true dynamics} can be viewed as a multi-agent extension of the Almgren-Chriss model (\cite{almgren2001optimal}, \cite{almgren1999value}, \cite{almgren2003optimal}).  
Models of this form, particularly when the $\tilde{\eta}_{j,tem}$'s ($\tilde{\eta}_{j,per}$'s) are all identical, 
have been applied in the context of predatory trading (\cite{carlin2007episodic}). On the other hand, \cite{Cardaliaguet2018, 2018arXiv180304094C} consider a mean-field game model where the interactions between the players are through the price as it is here. 
Although both of these papers address latency and learning in their models misspecification of agents's models is an important element in our framework. Moreover, our agents do not observe each other or know each other's parameters. In fact, in our finite player set-up they do not even know the number of players $N$.

In what follows we will say that a {\it mini-flash crash} occurs, if the $S_t^{exc} $
tends to either $+\infty$ or $-\infty$ on our time horizon.

The parameters $\tilde{\eta}_{j,per}$ and $\tilde{\eta}_{j,tem}$ are the {\it correct} 
values of Agent $j$'s permanent and temporary price impact parameters, respectively. 
We allow these quantities to have arbitrary relationships to Agent $j$'s corresponding {\it estimates}  
$\eta_{j,per}$ and $\eta_{j,tem}$. 
For instance, Agent $j$ might underestimate his permanent impact ($\eta_{j,per} < \tilde{\eta}_{j,per}$)
but perfectly estimate his temporary impact ($\eta_{j,tem} = \tilde{\eta}_{j,tem}$). 
Similarly, Agent $j$'s prior $\beta_j$ for the {\it correct} drift $\tilde{\beta}$ may be accurate or severely mistaken.
Comparing our descriptions of $S_{j,\theta_j}^{exc}$ in (\ref{exc price model j}) and
$S^{exc}$ in (\ref{rt true dynamics}), we see that Agent $j$ proxies each term in (\ref{rt true dynamics}) as follows:
\begin{align*}
\eta_{j,per} \left( X^{\theta_{j}^\star}_{j,t} -  x_{j} \right) 
\quad
&\longleftrightarrow 
\quad
\tilde{\eta}_{j,per} \left( X^{\theta_{j}^\star}_{j,t} -  x_{j} \right)
\\
\frac{1}{2} \eta_{j,tem} \theta_{j,t}^{\star}    
\quad
&\longleftrightarrow 
\quad
 \frac{1}{2} \tilde{\eta}_{j,tem} \theta_{j,t}^{\star}  
\\
S_{j,0} + \beta_j t +  W_{j,t}
\quad
&\longleftrightarrow 
\quad
 S_{0} +  \tilde{\beta}  t + \displaystyle\sum_{i \not = j}   \tilde{\eta}_{i,per}\left( X^{\theta_{i}^\star}_{i,t} -  x_{i} \right) 
+ \frac{1}{2} \displaystyle\sum_{i \not = j} \tilde{\eta}_{i,tem} \theta_{i,t}^{\star} + \tilde{W}_t  . 
\end{align*}

\section{Analysis of the Dynamical System}\label{true price dyn section}

When our agents implement the strategies that they believe are optimal (see Lemma \ref{1 pl soln acc to play model})
but $S^{exc}$ has the dynamics in (\ref{rt true dynamics}), what happens? 
The goal of Section \ref{true price dyn section} is to offer some general answers to this question.

To simplify our presentation, we begin by introducing and analyzing additional notation
(see Definition \ref{Phi j defn} and Lemma \ref{Phi j props lem}). 
We find that our agents' inventories
and trading rates evolve
according 
 to a particular ODE system with stochastic coefficients (see Lemma \ref{act 1st order ODE lem}). 
Under certain conditions, the system can have a singular point 
(see Lemma \ref{ODE HOM sing pt mult lem}). 
For convenience, we study what unfolds 
when this singular point is of the first kind (see Proposition \ref{opt trad speeds uncer}). 
We also examine the case in which there is no singular point (Proposition \ref{det A no root no blow up}).

We will have an even mix of deterministic 
and stochastic maps.
In what follows, we always explicitly indicate
$\omega$-dependence  
to distinguish between the two.
Our equations are solved
pathwise, so we do not encounter probabilistic 
concerns.  We will fix $\omega \in \tilde{\Omega}$ 
such that $\tilde{W}_{\cdot} \left( \omega \right)$ has a continuous path.

\begin{defn}\label{Phi j defn}
Define the maps
\begin{align*}
\begin{array}{lll}
\Phi_i : \left[ 0 , T \right] &\longrightarrow &\mathbb{R} \\
A : \left[ 0 , T \right] &\longrightarrow  &M_K \left( \mathbb{R} \right) \\
B : \left[ 0 , T \right) &\longrightarrow  &M_K \left( \mathbb{R} \right) \\
C \left( \cdot, \omega \right) : \left[ 0 , T \right] &\longrightarrow  &\mathbb{R}^K \\
\end{array}
\end{align*}
by 
\begin{align*}
\Phi_i \left( t \right) &\triangleq  \displaystyle\frac{  \tanh \left( \tau_{i} \left(t \right)  / 2 \right)  \nu_{i}^2    }{ \sqrt{ \eta_{i,tem}  \kappa_{i}}   \left( 1 + \nu_{i}^2 t \right)} \\
A_{ik} \left( t  \right) &\triangleq  
\left\{
\begin{array}{cc}
 1 -   \displaystyle\frac{1}{2} \left( \tilde{\eta}_{i,tem} - \eta_{i,tem}\right) \Phi_i \left( t \right)  
 & \quad \text{if } i = k \\
 -  \displaystyle\frac{1}{2}   \tilde{\eta}_{k,tem} \Phi_i \left( t \right) 
  & \quad \text{if } i \not =  k \\
\end{array}
\right.\\
B_{ik} \left( t  \right) &\triangleq  
\left\{
\begin{array}{cc}
\left( \tilde{\eta}_{i,per}  - \eta_{i,per} \right)  \Phi_i \left( t \right)  - \displaystyle\sqrt{\displaystyle\frac{ \kappa_{i}}{ \eta_{i,tem}}}   \coth \left( \tau_{i} \left( t \right) \right) 
 & \quad \text{if } i = k \\
\tilde{\eta}_{k,per}  \Phi_i \left( t \right)   
  & \quad \text{if } i \not =  k \\
\end{array}
\right.\\
C_{i} \left( t , \omega \right) &\triangleq  
 \Phi_i \left( t \right) \left[ \displaystyle\frac{ \mu_{i}}{\nu_{i}^2} +  \left( S_{0} - S_{i,0} \right)+  \tilde{\beta}  t 
 - \displaystyle\sum_{ \substack{ k \leq K \\ k \not = i }  } \tilde{\eta}_{k,per}   x_{k}  - x_{i} \left( \tilde{\eta}_{i,per}  - \eta_{i,per} \right)  \right. \\
 &\qquad \qquad \qquad \left.
  + \displaystyle\sum_{ k > K } \tilde{\eta}_{k ,per}\left( X^{\theta_{k}^\star}_{k,t}  -  x_{k} \right) 
      + \frac{1}{2} \displaystyle\sum_{  k > K  }  \tilde{\eta}_{k,tem} \theta_{k,t}^{\star}  + \tilde{W}_t \left( \omega \right) 
     \right] . 
\end{align*}
Here, $i \in \left\{ 1, \dots, K \right\}$ are the uncertain agents, whose behavior we are set out to characterize. The behavior of the certain agents are already described. They are not influenced by the execution price but they do have an influence on it.

\end{defn}

Observe that we can now write the dynamics in 
(\ref{x theta star}) as 
\begin{align*}
\theta^{\star}_{j,t} \left( \omega \right) &= - \displaystyle\sqrt{\displaystyle\frac{ \kappa_{j}}{ \eta_{j,tem}}}  
  \coth \left( \tau_{j} \left( t \right) \right) X^{\theta^\star_j}_{j,t}  \left( \omega \right) 
   + \Phi_j \left( t \right) \left[ \displaystyle\frac{ \mu_{j}}{\nu_{j}^2}  + \left( S_{j,t}^{unf}  \left( \omega \right) - S_{j,0}\right) \right] 
\end{align*}
when Agent $j$ is uncertain.

\begin{defn}\label{te defn FIN 1}
When $\det A$ has a root on $\left[ 0, T \right]$, we let $t_e$ denote the smallest one (see Lemma \ref{Phi j props lem}). 
\end{defn}

\begin{lem}\label{act 1st order ODE lem}

Suppose that $\det A$ has a root on $\left[ 0, T \right]$. 
If $t_e > 0$, then $S^{exc} \left( \omega \right)$, the $X^{\theta_{j}^\star}_{j} \left( \omega \right)$'s and 
the $\theta^{\star}_{j} \left( \omega \right)$'s are all uniquely defined and continuous on $\left[ 0, t_e \right)$. 
Moreover, 
the uncertain agents' strategies are characterized by
\begin{align}\label{DG syst lin ODE}
A \left(t \right)  \theta_{t}^{u,\star} \left( \omega \right) &= 
B \left( t \right) X^{u,\theta^\star}_{t} \left(\omega \right)  + C \left( t, \omega \right) , \qquad t \in \left( 0, t_e \right) \notag \\
X^{u,\theta^\star}_{0} \left(\omega \right) &= x^u,
\end{align}
where 
$\theta_{t}^{u,\star} \left( \omega \right)$ denotes the first $K$-entries of $\theta_{t}^{\star} \left( \omega \right)$.
When $\det A$ does not have a root on $\left[ 0, T \right]$, the same statements hold after replacing $t_e$ with $T$. 

\end{lem}

\begin{proof}
See Appendix \ref{act 1st order ODE lem PRF SS}.
\end{proof}

Lemma \ref{act 1st order ODE lem} does not address 
the behavior of our uncertain agents' inventories
and trading rates as $t \uparrow t_e$ or $t \uparrow T$. 
The difficulties are that $A$ is non-invertible at $t_e$, 
while $B$'s entries explode at $T$ (see Lemma \ref{Phi j props lem}).

The approach for resolving these issues is well-established (see Chapter 6 of \cite{codd+carl+ode}). 
We sketch the key points when $\det A$ has a root on $\left[ 0, T \right]$ and $t_e > 0$.
Analyzing 
the effects of $B$'s explosion at $T$ is similar (see Proposition \ref{det A no root no blow up}).

We begin by considering the homogeneous equation corresponding to (\ref{DG syst lin ODE}):
\begin{align}\label{DG syst lin ODE HOM}
A \left(t \right)  \dot{X}_t^u \left( \omega \right) &= 
B \left( t \right) X_t^u \left(\omega \right)   , \qquad t \in \left( 0, t_e \right) \notag \\
X_0^u \left(\omega \right) &= x^u. 
\end{align}
We change notation to emphasize that (\ref{DG syst lin ODE HOM}) no longer 
describes the uncertain agents' optimal strategies. 
We next write (\ref{DG syst lin ODE HOM}) in a more convenient form.

\begin{lem}\label{ODE HOM sing pt mult lem}

Suppose that $\det A$ has a root on $\left[ 0, T \right]$ and $t_e > 0$.
Near $t_e$, the solution of (\ref{DG syst lin ODE HOM})
satisfies 
\begin{align}\label{DG syst lin ODE HOM REG}
\left(t - t_e \right)^{m+1} \dot{X}_{t}^{u} \left( \omega \right) &= 
D \left( t \right) X_{t}^{u}\left(\omega \right)   .
\end{align}
In (\ref{DG syst lin ODE HOM REG}), $m$ is a nonnegative integer such that the multiplicity of the zero of 
$\det A$ at $t_e$ is $\left( m + 1 \right)$.
$D$ is a particular analytic map for which $D \left( t_e \right)$ has rank 0 or 1 (see (\ref{D defn FIN})).

\end{lem}

\begin{proof}
See Appendix \ref{ODE HOM sing pt mult lem PRF SS}.
\end{proof}

\begin{defn}\label{D f m}

Let us denote the unique non-zero eigenvalue of $D \left( t_e \right)$ in the above lemma by $\lambda$.
\end{defn}

\begin{prop}\label{opt trad speeds uncer}

Suppose that $\det A$ has a root on $\left[ 0, T \right]$, $t_e > 0$, 
and $m = 0$. 
If $\lambda \not \in \mathbb{Z}$,
then for some 
small $\rho > 0$, 
\begin{align}\label{gen limit diag R main blow up KZ}
 X^{u,\theta^\star}_{t} \left( \omega \right) &=
 P \left( t \right) 
 \left[
\displaystyle\sum_{j = 1}^{K-1} 
\left( y_j \left( \omega \right) -   \displaystyle\int_{t_e - \rho }^t 
\displaystyle\frac{ F_j \left( s, \omega \right)}
{ \left| s - t_e \right| }\, ds 
\right) 
v_{j} \right. \notag \\
&\qquad \qquad +  \left.
\left| t - t_e \right|^{ \lambda  }   \left( y_K \left( \omega \right)  -
 \displaystyle\int_{t_e - \rho }^t 
\displaystyle\frac{  F_K \left( s, \omega \right)   }
{ \left| s - t_e \right|^{1+ \lambda} }\, ds 
\right)
v_{K}
 \right] 
\end{align}
for $t \in \left( t_e - \rho, t_e \right)$. 
Here,
\begin{itemize}[label=$\bullet$]

\item $\left\{v_1, \dots, v_{K} \right\}$ is an eigenbasis for $D \left( t_e \right)$ ($v_K$ corresponds to $\lambda$); 

\item $P$ is a (non-singular-)matrix-valued analytic function on $\left[ t_e - \rho, t_e \right]$ such that $P \left( t_e \right) = I_K$
 (see (\ref{fund soln homog R nonzer FIN1}));

\item $\left\{ y_1 \left( \omega \right) , \dots , y_K \left( \omega \right) \right\}$ are constants (see (\ref{y F defn eign})); 

\item and $\left\{ F_1 \left( \cdot , \omega \right), \dots , F_K \left( \cdot , \omega \right) \right\}$ are continuous 
real-valued functions on $\left[ t_e - \rho , t_e \right]$ (see (\ref{y F defn eign})).

\end{itemize}
We get $\theta^{u,\star} \left( \omega \right)$ and $S^{exc} \left( \omega \right)$ on $\left( t_e - \rho, t_e \right)$ by differentiating (\ref{gen limit diag R main blow up KZ}) and by substituting 
 $X^{\theta^\star} \left( \omega \right)$ and $\theta^{\star} \left( \omega \right)$ into (\ref{rt true dynamics}), respectively.

\end{prop}

\begin{proof}
See Appendix \ref{opt trad speeds uncer PRF SS}.
\end{proof}

\begin{prop}\label{det A no root no blow up}

Suppose that $\det A$ does not have a root on $\left[ 0, T \right]$. 
Then $S^{exc} \left( \omega \right)$, the $X^{\theta_{j}^\star}_{j} \left( \omega \right)$'s and 
the $\theta^{\star}_{j} \left( \omega \right)$'s are all uniquely defined and continuous on $\left[ 0, T \right]$. 
Moreover, 
\begin{equation}\label{liq in right mod too}
\displaystyle\lim_{t \uparrow T} X^{\theta^\star}_{t} \left( \omega \right) = 0. 
\end{equation}

\end{prop}

\begin{rem}

Each agent believes that his terminal inventory will be zero almost surely (see (\ref{admiss strat term})). 
Proposition \ref{det A no root no blow up} specifies conditions 
under which the agents are effectively correct in this regard.

\end{rem}

\begin{proof}
See Appendix \ref{det A no root no blow up PRF SS}.
\end{proof}

\section{Explicit Characterizations of Flash Crashes for Semi-Symmetric Uncertain Agents}\label{exmp sect}

In this section we thoroughly analyze
a broad but tractable class of scenarios. 
This will enable us to both theoretically
and numerically investigate the 
occurrence of mini-flash crashes. 
We specify that our uncertain agents' parameters are 
identical, except for their 
initial inventories $x_j$, means of their initial drift priors $\mu_j$, 
and their initial estimates for the fundamental price $S_{j,0}$. 
Such agents are nearly symmetric, so we call them {\it semi-symmetric}.

\begin{defn}\label{semi symm defn}

We say that our uncertain agents are {\it semi-symmetric} 
when there are positive constants 
\begin{align*}
\tilde{\eta}_{tem}, \quad \eta_{tem}, \quad \tilde{\eta}_{per}, \quad \eta_{per}, \quad \nu^2,  \quad \text{and} \quad \kappa
\end{align*}
such that for each $i \in \left\{ 1, \dots, K \right\}$
\begin{align*}
\begin{array}{lll}
\tilde{\eta}_{i,tem} = \tilde{\eta}_{tem},  &\eta_{i,tem} = \eta_{tem},  &\tilde{\eta}_{per} = \tilde{\eta}_{i,per}, \\
\eta_{i,per} = \eta_{per}, &\nu_{i}^2 = \nu^2 ,  &\kappa_{i} = \kappa . 
\end{array}
\end{align*}
\end{defn}
Since $\tau_j$'s and the $\Phi_j$'s are the same for $j \leq K$ 
(see Definitions \ref{tau j notn}, \ref{Phi j defn}, and \ref{semi symm defn}). 
We denote these functions by $\tau$ and $\Phi$, respectively.

Definition \ref{semi symm defn} implies that the diagonal entries of $A$ are identical, 
as are the off-diagonal entries. 
The same is true for $B$ (see Definition \ref{Phi j defn}). 
Such a simplification considerably reduces the difficulties in computing $\det A$, $\lambda$, 
and an eigenbasis for $D \left( t_e \right)$
(see (\ref{det A semi symm tild FIN}) and Lemma \ref{EV comps SS FIN}). 
The $x_j$'s, $\mu_j$'s, and $S_{j,0}$'s only enter in $C$, 
which also has a nice structure (see (\ref{evals adj A t eqn 1})).

\subsection{Results}\label{ss inv exmp sect}

\begin{lem}\label{te study lem SS FIN}

Suppose that the uncertain agents are semi-symmetric. 
Then $\det A$ has a root on $\left[ 0 , T \right]$ and $t_e > 0$ if and only if
\begin{equation}\label{semi-symmetric gull te in 0 T FIN REWRITE}
 \displaystyle\frac{  \nu^2  \left( K \tilde{\eta}_{tem} - \eta_{tem} \right)  \tanh \left( \displaystyle\frac{T}{2} \displaystyle\sqrt{\displaystyle\frac{ \kappa }{ \eta_{tem}}} \right)     }{ \sqrt{ \eta_{tem}  \kappa }  } 
> 2. 
\end{equation}
In this case, the zero of $\det A \left( \cdot \right)$ at $t_e$ is of multiplicity 1. 

\end{lem}

\begin{rem}\label{t_e exists rem FIN}

By varying our parameters in (\ref{semi-symmetric gull te in 0 T FIN REWRITE}) one at a time we can obtain the following interpretations discussed in Section \ref{back contrib sect}:

\begin{enumerate}[label=\roman*)]

\item \label{a SS te exists lem} 
(\ref{semi-symmetric gull te in 0 T FIN REWRITE}) holds when $\nu^2$ is high. 
Since $\nu^2$ is the variance of the uncertain agents' drift priors, 
we are led to (\ref{first human error F} in Section \ref{back contrib sect}.

\item \label{b SS te exists lem} 
(\ref{semi-symmetric gull te in 0 T FIN REWRITE}) holds when $\left( K \tilde{\eta}_{tem} - \eta_{tem} \right)$ is high. 
A given uncertain agent believes that his own temporary impact parameter is $\eta_{tem}$, 
while
the actual collective temporary impact parameter induced by the uncertain agents is 
$K \tilde{\eta}_{tem}$. 
Then $\left( K \tilde{\eta}_{tem} - \eta_{tem} \right)$ is large whenever
each uncertain agent severely underestimates his own temporary impact or 
there are many uncertain agents, giving (\ref{last human error} in Section \ref{back contrib sect}.

\item \label{c SS te exists lem} 
(\ref{semi-symmetric gull te in 0 T FIN REWRITE}) holds when $T$ is high. 
Since $\left[ 0 , T \right]$ is our time horizon, we get 
(\ref{first human error D} in Section \ref{back contrib sect}.
Note that $T$ must be small enough for our 
agents' modeling rationale to hold (see Section \ref{mod details sec}); 
however, $T$ need not be too large here, as the value of
$\tanh$ reaches $95\%$ of its supremum 
on $\left[ 0 , \infty \right)$ for arguments 
greater than 1.8.

\item \label{d SS te exists lem} 
(\ref{semi-symmetric gull te in 0 T FIN REWRITE}) holds when $\kappa$ is low.
We conclude 
(\ref{first human error E} in Section \ref{back contrib sect}, 
as $\kappa$ measures our uncertain agents' aversion to volatility risks (see Section \ref{obj fxn subsec}). 
Observe that both the numerator and the denominator of the LHS in (\ref{semi-symmetric gull te in 0 T FIN}) 
roughly look like 
$\sqrt{\kappa}$ for small $\kappa$; however, when $\kappa$ is large, 
the whole LHS looks like $1 / \sqrt{\kappa}$
since $\tanh$ is bounded by 1 on $\left[ 0 , \infty \right)$.

\end{enumerate}

\end{rem}

\begin{proof}
See Appendix \ref{te study lem SS FIN PRF SS}. 
\end{proof}

\begin{rem}\label{te impl formula FIN}

As observed in (\ref{key eqn Phi te FIN}), when $\det A$ has a root on $\left[  0, T \right]$, we have
\begin{align}\label{spec Phi te relat FIN}
 \Phi \left( t_e \right)   = 
 \displaystyle\frac{2}{ K \tilde{\eta}_{tem} - \eta_{tem} } .
\end{align}
No agent would think to compute $t_e$ since they all believe that a mini-flash crash is a null event; however,
(\ref{spec Phi te relat FIN}) makes it especially clear that they could not do so anyway.

\end{rem}

\begin{lem}\label{EV comps SS FIN}

Suppose that the uncertain agents are semi-symmetric and (\ref{semi-symmetric gull te in 0 T FIN REWRITE}) holds. 
Then
\begin{align}\label{lambda formula FIN}
\lambda &= 
\displaystyle\frac{
2
\left[
\displaystyle\sqrt{\displaystyle\frac{ \kappa }{ \eta_{tem}}}   \coth \left( \tau \left( t_e \right) \right)  
-  2 \left(  \displaystyle\frac{ K \tilde{\eta}_{per} - \eta_{per} }{ K \tilde{\eta}_{tem} - \eta_{tem} } \right)
\right] 
}
{ \left( K \tilde{\eta}_{tem} - \eta_{tem} \right) \dot{\Phi} \left( t_e \right)  } 
\end{align}
and the corresponding eigenvector is $v_K = \left[ 1, \dots, 1 \right]^{\top}$. 
By slightly perturbing $\tilde{\eta}_{per}$ and/or $\eta_{per}$, if necessary, 
we can ensure that $\lambda \not \in \mathbb{Z}$. 
In this case, $D \left( t_e \right)$ is diagonalizable and the remaining vectors in an 
eigenbasis for $D \left( t_e \right)$ (all with the eigenvalue zero) are 
\begin{align*}
v_1 = \left[ -1, 1, 0, \dots, 0 \right]^{\top}
\quad , \dots, \quad
  v_{K-1} = \left[ -1, 0, \dots, 0, 1 \right]^{\top}. 
\end{align*}

\end{lem}

\begin{rem}\label{lambda sign remark}

With the exceptions of $\tilde{\eta}_{per}$ and $\eta_{per}$, all parameters
in (\ref{lambda formula FIN}) determine whether or not $\det A$ has a root on $\left[ 0, T \right]$ 
(see Lemma \ref{te study lem SS FIN}). 
They also fix the value of $t_e$ (see Remark \ref{te impl formula FIN}). 
Hence, to interpret (\ref{lambda formula FIN}), we only consider the roles 
of $\tilde{\eta}_{per}$ and $\eta_{per}$. 
These parameters enter (\ref{lambda formula FIN}) via 
\begin{equation}\label{mistake ration FIN}
 \displaystyle\frac{ K \tilde{\eta}_{per} - \eta_{per} }{ K \tilde{\eta}_{tem} - \eta_{tem} } . 
\end{equation}

Intuitively, (\ref{mistake ration FIN}) can be viewed as the ratio of two terms:
The numerator measures how far a given uncertain agent's estimate of his own permanent impact is
from the uncertain agents' actual collective permanent impact. 
The denominator, which must be positive due to Lemma \ref{te study lem SS FIN}, is the corresponding measure for the temporary impact.
One might call (\ref{mistake ration FIN}) a {\it mistake ratio}.

Since $\dot{\Phi} \left( t_e \right) < 0$ by Lemma \ref{Phi j props lem}, $\lambda$ is positive only when
 (\ref{mistake ration FIN}) is high enough. 
 We have $\lambda < 0$ when the uncertain agents' total permanent impact and a single uncertain agent's' estimate of his own 
permanent impact are too close or when his estimate exceeds the cumulative permanent impact. 
 More precisely, 
\begin{align}\label{mistake rat lamb sign 3}
&\left\{ \lambda > 0 \right\} 
\quad \iff \quad 
\left\{
\frac{1}{2} \displaystyle\sqrt{\displaystyle\frac{ \kappa }{ \eta_{tem}}}   \coth \left( \tau \left( t_e \right) \right)  
\left( K \tilde{\eta}_{tem} - \eta_{tem} \right) 
<   K \tilde{\eta}_{per} - \eta_{per} 
\right\} \notag \\
&\left\{ \lambda < 0 \right\} 
\quad \iff \quad 
\left\{
\frac{1}{2} \displaystyle\sqrt{\displaystyle\frac{ \kappa }{ \eta_{tem}}}   \coth \left( \tau \left( t_e \right) \right)  
\left( K \tilde{\eta}_{tem} - \eta_{tem} \right) 
>   K \tilde{\eta}_{per} - \eta_{per} 
\right\} . 
\end{align}
Whether a mini-flash crash is accompanied by high or low trading volumes 
is effectively determined by which inequality in (\ref{mistake rat lamb sign 3}) 
holds (see 
Theorems \ref{SS main blow up lem FIN} and \ref{no Inv expl lem SS FIN}
and Sections 
 \ref{back contrib sect}, \ref{case 2 subsect}, and \ref{case 3 subsect}).

\end{rem}

\begin{proof}
See Appendix \ref{EV comps SS FIN PRF SS}. 
\end{proof}

\begin{thm}\label{SS main blow up lem FIN}

Suppose that the uncertain agents are semi-symmetric and (\ref{semi-symmetric gull te in 0 T FIN REWRITE}) holds. 
Assume that $\lambda \not \in \mathbb{Z}$ and $\lambda < 0$ (see Lemma \ref{EV comps SS FIN}). 
Let $\rho$, $y_K \left( \omega \right)$, and $F_K \left( \cdot , \omega \right)$ be defined as in Proposition \ref{opt trad speeds uncer}. 
Then 
\begin{align}\label{neg Lamb SS FIN pos expl}
&\left\{
y_K \left( \omega \right)
>  \displaystyle\lim_{t \uparrow t_e} 
 \displaystyle\int_{ t_e - \rho   }^t 
\displaystyle\frac{  F_K \left( s, \omega \right)}
{ \left| s - t_e \right|^{1  +  \lambda } }\, ds 
\right\} \\
&\quad \implies \quad
  \left\{  \displaystyle\lim_{t \uparrow t_e} X^{u, \theta^\star}_{t} \left( \omega \right)
 =  \displaystyle\lim_{t \uparrow t_e} \theta^{u,\star}_{t} \left( \omega \right) 
    = \left[ + \infty, \dots, + \infty \right]^{\top} , \quad
    \displaystyle\lim_{t \uparrow t_e} S^{exc}_{t} \left( \omega \right) = + \infty
  \right\} \notag
  \end{align}
  and 
 \begin{align}\label{neg Lamb SS FIN neg expl}
&\left\{
y_K \left( \omega \right)
< \displaystyle\lim_{t \uparrow t_e} 
 \displaystyle\int_{   t_e - \rho  }^t 
\displaystyle\frac{  F_K \left( s, \omega \right)}
{ \left| s - t_e \right|^{1  +  \lambda } }\, ds 
\right\}     \\
&\quad \implies \quad
  \left\{  \displaystyle\lim_{t \uparrow t_e} X^{u, \theta^\star}_{t} \left( \omega \right)
 =  \displaystyle\lim_{t \uparrow t_e} \theta^{u,\star}_{t} \left( \omega \right) 
    = \left[ - \infty, \dots, - \infty \right]^{\top} , \quad
    \displaystyle\lim_{t \uparrow t_e} S^{exc}_{t} \left( \omega \right) = - \infty
  \right\} \notag.
  \end{align}
 Moreover, 
 \begin{enumerate}[label=\roman*)]
  
  \item  \label{neg EV SS lem A FIN} 
  The integral limits in (\ref{neg Lamb SS FIN pos expl}) and (\ref{neg Lamb SS FIN neg expl})
exist and are finite. 

\item \label{neg EV SS lem B FIN} 
Either (\ref{neg Lamb SS FIN pos expl}) or (\ref{neg Lamb SS FIN neg expl}) holds $\tilde{P}$-a.s.

\item \label{neg EV SS lem C FIN} 
At $t_e - \rho$, 
the events (\ref{neg Lamb SS FIN pos expl}) and (\ref{neg Lamb SS FIN neg expl})
both have positive $\tilde{P}$-probability; however, the $\tilde{P}$-probability of one event tends
to 1 (while the other tends to 0) if we let $\rho \downarrow 0$. 

 \end{enumerate}

\end{thm}

\begin{rem}\label{neg lambda rem FIN}

Since $P \left( t_e \right) = I_K$ (see Proposition \ref{opt trad speeds uncer}), 
(\ref{y F defn eign}) and Lemma \ref{EV comps SS FIN} imply that $y_K \left( \omega \right)$ will be 
large and positive when the uncertain agents hold significant, similar long positions. 
$y_K \left( \omega \right)$ will be of high magnitude but negative, if the uncertain agents 
carry substantial, similarly-sized short positions. 
Hence, a spike in $S^{exc} \left( \omega \right)$ is more likely when the uncertain agents 
are synchronized aggressive buyers, while the odds of a collapse improve
when they are synchronized heavy sellers.
These effects play the deciding role as $t \uparrow t_e$, as the integral limits in (\ref{neg Lamb SS FIN pos expl})
and (\ref{neg Lamb SS FIN neg expl}) are finite.

Still, due to how we can decompose $F_K$ in our case (see (\ref{Fj1 Fj2 defn})), 
large fluctuations in the fundamental price 
can make the mini-flash crash's direction unclear until just before $t_e$ (see Figure \ref{fig: ExcPri_NegEV_Expl_MIN_MAX_DOWN}).

\end{rem}

\begin{proof}
See Appendix \ref{SS main blow up lem FIN PRF SS}.
\end{proof}

\begin{thm}\label{no Inv expl lem SS FIN}

Suppose that the uncertain agents are semi-symmetric and 
(\ref{semi-symmetric gull te in 0 T FIN REWRITE}) holds. 
Assume that $\lambda \not \in \mathbb{Z}$ and $\lambda > 0$ (see Lemma \ref{EV comps SS FIN}). 
Then $\tilde{P}$-a.s.,  
\begin{equation*}
\displaystyle\lim_{t \uparrow t_e} X^{\theta^\star}_{t} \left( \omega \right) 
\end{equation*}
exists in $\mathbb{R}^N$. 
If any coordinates of $\theta^{u,\star}_{t} \left( \omega \right)$ explode, then 
$S_t^{exc} \left( \omega \right)$ and all coordinates of $\theta^{u,\star}_{t} \left( \omega \right)$ 
explode in the same direction.  
For instance, when $\lambda > 1$, 
\begin{align}\label{lamb > 1 final lem limit 1}
&\left\{ \displaystyle\lim_{t \uparrow t_e} \left[
 \left| t - t_e \right|^{\lambda - 1}  
 \displaystyle\int_{t_e - \rho }^t \displaystyle\frac{ \tilde{W}_s \left( \omega \right)  - \tilde{W}_t \left( \omega \right) }
{ \left| s - t_e \right|^{1+ \lambda} }\, ds 
\right] = + \infty \right\} \\
&\quad \implies \quad 
  \left\{  \displaystyle\lim_{t \uparrow t_e} \theta^{u,\star}_{t} \left( \omega \right) 
    = \left[ + \infty, \dots, + \infty \right]^{\top} , \quad
    \displaystyle\lim_{t \uparrow t_e} S^{exc}_{t} \left( \omega \right) = + \infty
  \right\} \notag 
 \end{align}
 and 
 \begin{align}\label{lamb > 1 final lem limit 2}
  &\left\{ \displaystyle\lim_{t \uparrow t_e} \left[
 \left| t - t_e \right|^{\lambda - 1}  
 \displaystyle\int_{t_e - \rho }^t \displaystyle\frac{ \tilde{W}_s \left( \omega \right)  - \tilde{W}_t \left( \omega \right) }
{ \left| s - t_e \right|^{1+ \lambda} }\, ds 
\right] = - \infty \right\} \\
&\quad \implies \quad 
  \left\{  \displaystyle\lim_{t \uparrow t_e} \theta^{u,\star}_{t} \left( \omega \right) 
    = \left[ - \infty, \dots, - \infty \right]^{\top} , \quad
    \displaystyle\lim_{t \uparrow t_e} S^{exc}_{t} \left( \omega \right) = - \infty
  \right\} \notag.
\end{align}
 Moreover, 
 \begin{enumerate}[label=\roman*)] 

\item \label{POS EV SS lem B FIN} 
Either (\ref{lamb > 1 final lem limit 1}) or (\ref{lamb > 1 final lem limit 2}) holds $\tilde{P}$-a.s.

\item \label{POS EV SS lem C FIN} 
At $t_e - \rho$, 
the events (\ref{lamb > 1 final lem limit 1}) and (\ref{lamb > 1 final lem limit 2})
both have positive $\tilde{P}$-probability; however, the $\tilde{P}$-probability of one event tends
to 1 (while the other tends to 0) if we let $\rho \downarrow 0$. 

 \end{enumerate} 

\end{thm}

\begin{rem}\label{lamb > 0 but bet 0 1 rem}
We make no rigorous statement regarding the $\lambda \in \left( 0 , 1 \right)$ case. 
Most of Theorem~\ref{no Inv expl lem SS FIN}'s proof 
would still be valid (see Appendix \ref{SS main blow up lem FIN PRF SS}); 
however, the final estimates are especially convenient 
when $\lambda > 1$ 
(see (\ref{term by term explo pos lamb}) - (\ref{int by parts lamb > 1})). 
The over-arching purpose of Theorem~\ref{no Inv expl lem SS FIN}
is only to illustrate that mini-flash crashes
can occur in low trading volume environments (see Section \ref{back contrib sect}). 
Nevertheless, we suspect that mini-flash crashes might unfold when $\lambda \in \left( 0 , 1 \right)$, e.g., 
see Section \ref{case 2 subsect} and (\ref{term by term explo pos lamb}) - (\ref{int by parts lamb > 1}).

\end{rem}

\begin{proof}
See Appendix \ref{SS main blow up lem FIN PRF SS}.
\end{proof}

\section{Numerical illustrations}\label{sec:numerical}

\subsection{Example 1: No mini-flash crash}\label{case 1 subsect}

Our mini-flash crashes do not always occur
(see Lemmas \ref{det A no root no blow up} and \ref{te study lem SS FIN}). 
In Section \ref{case 1 subsect}, 
we illustrate this by numerically 
simulating a scenario in which 
$\det A$ has no root on $\left[ 0 , T \right]$.

By Lemma \ref{te study lem SS FIN} and (\ref{det A semi symm tild FIN}), 
we know that $\det A$ is non-vanishing on $\left[ 0 , T \right]$ if and only if 
\begin{equation}\label{no root detA exmp 1 cond}
\left( K \tilde{\eta}_{tem} - \eta_{tem} \right) \Phi \left( 0 \right) < 2. 
\end{equation}
One selection of parameters for which (\ref{no root detA exmp 1 cond}) is satisfied is
\begin{align}\label{no root det A params}
\begin{array}{llll}
N = 3,&  K = 2, & T = 1, & S_0 = 100, \\
 \tilde{\beta} = 1, & \tilde{\eta}_{tem} = 1,  & \eta_{tem} = 0.75,  &\tilde{\eta}_{per} = 1, \\
\eta_{per} = 1, & \nu^2 = 2,  &\kappa = 5 , & x_1 = 2,\\
  x_2 = -2, & \mu_1 = 15, &\mu_2 = -10, & S_{1,0} = 100,\\
  S_{2,0} = 100, & \tilde{\eta}_{3,tem} = 1,  & \eta_{3,tem} = 1,  &\tilde{\eta}_{3,per} = 1, \\
 \mu_3 = -3, & \nu_3^2 = 2,  &\kappa_3 = 5, & x_3 = 2 . \\
\end{array}
\end{align}
In fact, the LHS of (\ref{no root detA exmp 1 cond}) then equals 1.1095. 
Observe that there is no need to specify $\eta_{3,per}$ and $S_{3,0}$ as they are irrelevant 
(see  Corollary \ref{cert agents opt strats}, 
Definition \ref{Phi j defn}, and Lemma \ref{act 1st order ODE lem}).
Again, our purposes are only illustrative here, 
and we leave the reproduction of a specific practically meaningful scenario for a future work.

Since $K = 2$ and $N = 3$, we have two uncertain agents and one certain agent in 
the coming plots. 
We label the corresponding curves with $U1$, $U2$, and $C1$. 
For example, the label $U1$ will signify a quantity for Agent 1, the first uncertain agent. 
In Figures \ref{fig: Opt_Inv_NoRoot_MAX}
and  \ref{fig: Opt_Spd_NoRoot_MAX}, 
we plot 
inventories and trading rates. 
The execution price is 
depicted in Figure \ref{fig: ExcPri_NoRoot_MAX}.

The diagrams exhibit all of the important qualities that 
we expect based upon our theoretical results. 
Here are a few key features:
\begin{enumerate}[label=\roman*)]

\item All agents liquidate their positions by the terminal time $T$ 
(see (\ref{liq in right mod too}) and Figure \ref{fig: Opt_Inv_NoRoot_MAX}). 

\item $S^{exc} \left( \omega \right)$, the $X^{\theta_{j}^\star}_{j} \left( \omega \right)$'s and 
the $\theta^{\star}_{j} \left( \omega \right)$'s are all continuous on $\left[ 0, T \right]$ 
(see Proposition \ref{det A no root no blow up} and Figures \ref{fig: Opt_Inv_NoRoot_MAX} - \ref{fig: ExcPri_NoRoot_MAX}).  

\item The uncertain agents' trading rates appear to 
exhibit a Brownian component (see Lemma \ref{1 pl soln acc to play model} and Figure \ref{fig: Opt_Spd_NoRoot_MAX}).

\item The certain agent's trading rate appears to be smooth on $\left[ 0 , T \right]$ 
(see Corollary \ref{cert agents opt strats} and Figure \ref{fig: Opt_Spd_NoRoot_MAX}).

\item
The agents need not 
either strictly buy or strictly sell throughout $\left[ 0, T \right]$ 
(see Figure \ref{fig: Opt_Spd_NoRoot_MAX}).

\item
Even so, the agents may decide 
to strictly buy or strictly sell throughout $\left[ 0, T \right]$ 
(see Figure \ref{fig: Opt_Spd_NoRoot_MAX}).

\item 
The uncertain agents' trading rates do not appear to synchronize 
(see Figure \ref{fig: Opt_Spd_NoRoot_MAX}).

\end{enumerate}



\begin{figure}[!tbp]
  \centering
 \includegraphics[scale=0.37]{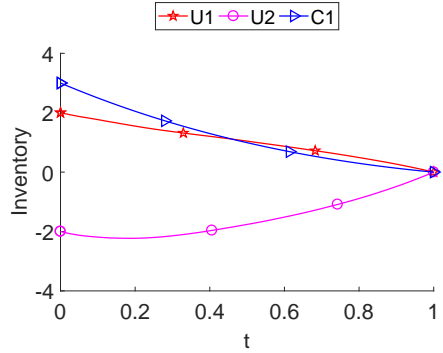}
  \caption{Depiction of the agents' inventories in Section \ref{case 1 subsect}.}
  \label{fig: Opt_Inv_NoRoot_MAX}
\end{figure}

\begin{figure}[!tbp]
  \centering
 \includegraphics[scale=0.37]{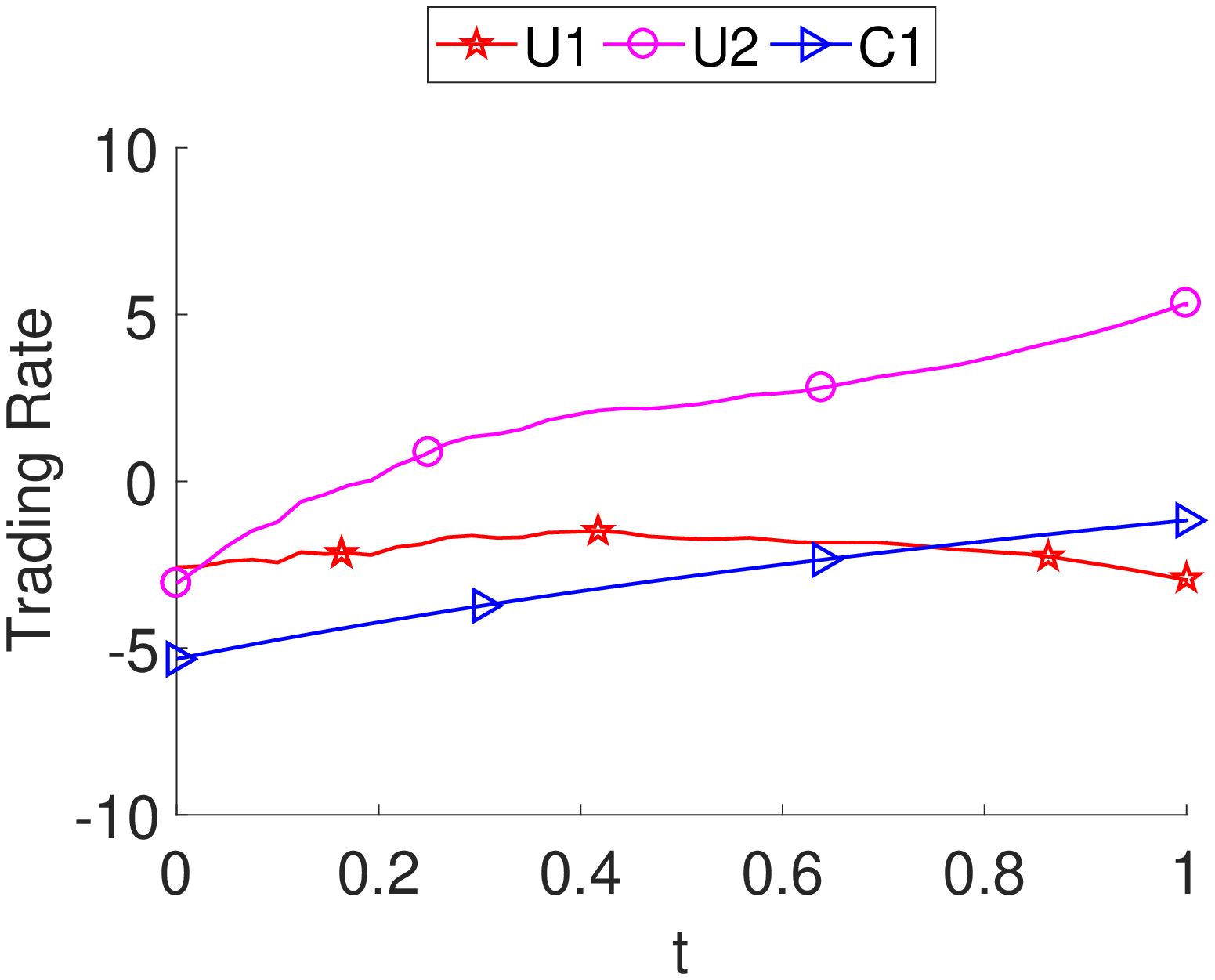}
  \caption{Depiction of the agents' trading rates in Section \ref{case 1 subsect}.}
  \label{fig: Opt_Spd_NoRoot_MAX}
\end{figure}


\begin{figure}[!tbp]
  \centering
 \includegraphics[scale=0.37]{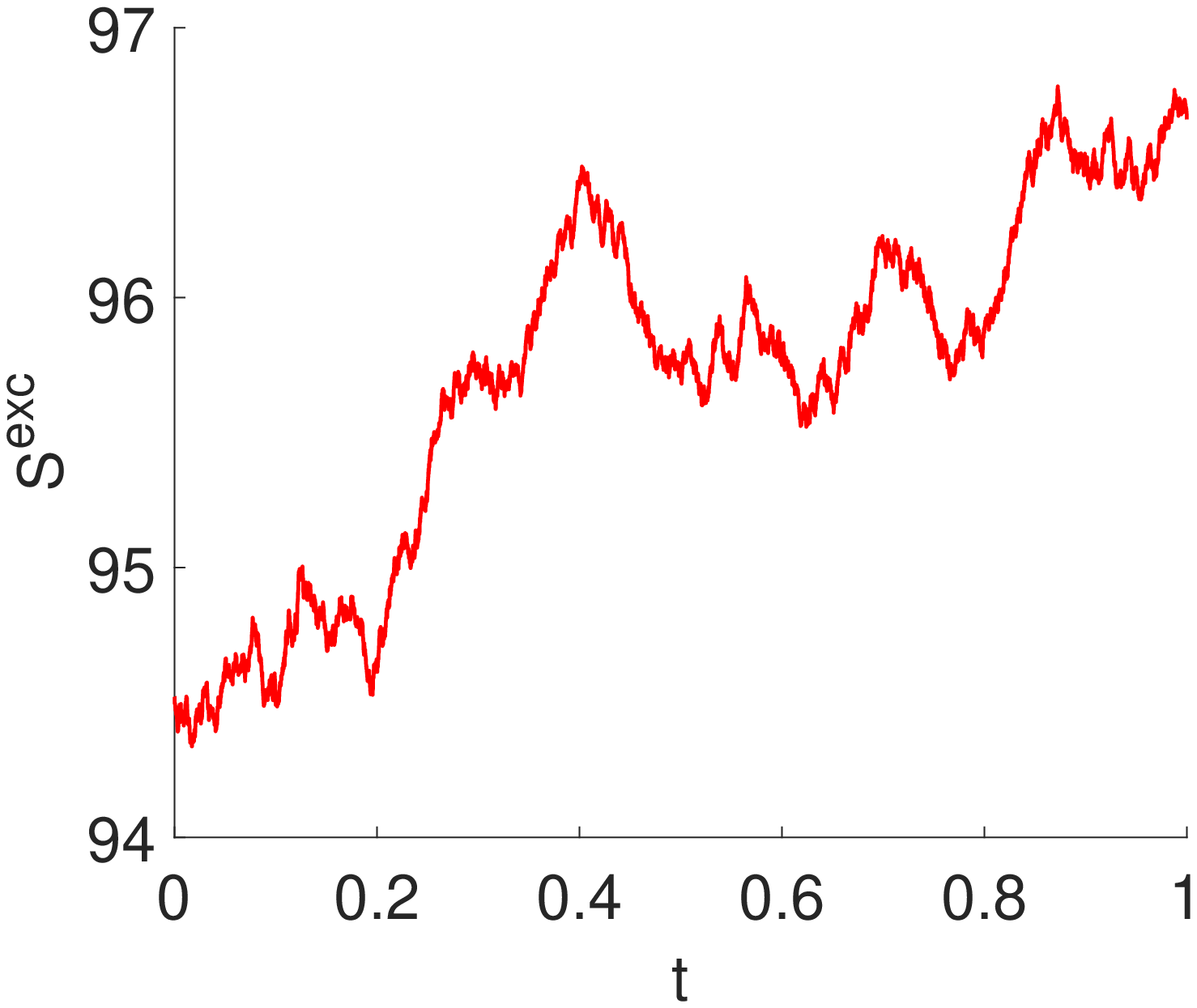}
  \caption{Depiction of the execution price in Section \ref{case 1 subsect}.}
  \label{fig: ExcPri_NoRoot_MAX}
\end{figure}

\begin{figure}[!tbp]
  \centering
\includegraphics[scale=0.37]{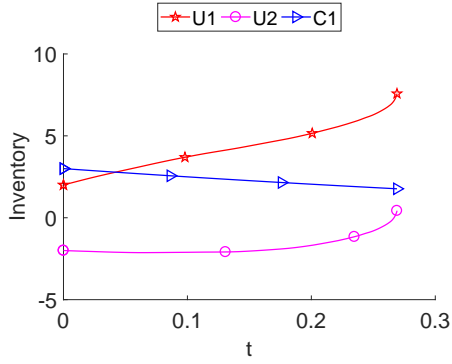}
  \caption{Depiction of the agents' inventories in Section \ref{case 2 subsect}.}
  \label{fig: Opt_Inv_PosEV_Expl_MAX}
\end{figure}

\subsection{Example 2: A mini-flash crash with low trading volume}\label{case 2 subsect}

Our mini-flash crashes can be accompanied by low trading volumes 
(see Theorem~\ref{no Inv expl lem SS FIN}). 
In Section \ref{case 2 subsect}, we visualize this 
by studying a concrete scenario in which 
$\det A$ has a root on $\left[ 0 , T \right]$; 
$t_e > 0$; 
the zero of $\det A$ at $t_e$ is of multiplicity 1; 
$\lambda \not \in \mathbb{Z}$; 
and $\lambda > 0$.
The behavior of the $X^{\theta_{j}^\star}_{j} \left( \omega \right)$'s
is then characterized by 
Corollary \ref{cert agents opt strats}
and
Theorem~\ref{no Inv expl lem SS FIN}. 
Theorem~\ref{no Inv expl lem SS FIN} would rigorously 
describe $S^{exc}_t \left( \omega \right)$
and the $\theta^{\star}_{j,t} \left( \omega \right)$'s
as $t \uparrow t_e$,
if $\lambda > 1$. 
To improve the quality of our plots, 
we consider a situation where $\lambda \in \left( 0 ,1 \right)$
instead
(see Remark \ref{lamb > 0 but bet 0 1 rem}).

By Lemmas \ref{te study lem SS FIN} and \ref{EV comps SS FIN}, 
we must select parameters such that (\ref{semi-symmetric gull te in 0 T FIN REWRITE})
is satisfied and 
\begin{equation}\label{lambda calc exmp SS subsect 2}
\lambda = 
\displaystyle\frac{
2
\left[
\displaystyle\sqrt{\displaystyle\frac{ \kappa }{ \eta_{tem}}}   \coth \left( \tau \left( t_e \right) \right)  
-  2 \left(  \displaystyle\frac{ K \tilde{\eta}_{per} - \eta_{per} }{ K \tilde{\eta}_{tem} - \eta_{tem} } \right)
\right] 
}
{ \left( K \tilde{\eta}_{tem} - \eta_{tem} \right) \dot{\Phi} \left( t_e \right)  } 
\end{equation}
is a positive non-integer. 
We can keep most of our choices in (\ref{no root det A params}) the same and only make a few revisions:
\begin{align}\label{pos Lambda new params}
\begin{array}{lll}
\tilde{\eta}_{tem} = 0.5,  & \eta_{tem} = 0.2,  &\tilde{\eta}_{per} = 0.8, \\
\eta_{per} = 0.025, & \nu^2 = 3,  &\kappa = 1 . 
\end{array}
\end{align}
As in Section \ref{case 1 subsect}, we do not seek to replicate a particular historical situation. 
We immediately get (\ref{semi-symmetric gull te in 0 T FIN REWRITE}), as its LHS is 4.3302. 
Using Remark \ref{te impl formula FIN} and (\ref{lambda calc exmp SS subsect 2}), we can show that 
\begin{equation*}
t_e = 0.2691
\quad \text{and} \quad 
\lambda = 0.5939 .
\end{equation*}

Again,  
we have two uncertain agents and one certain agent. 
We retain the $\left\{ U1, U2, C1\right\}$- labeling system from Section \ref{case 1 subsect}. 
The inventories, trading rates, and execution price are plotted in
Figures \ref{fig: Opt_Inv_PosEV_Expl_MAX} - \ref{fig: ExcPri_PosEV_Expl_MIN_MAX}. 
To aid our illustration, 
we truncate the time domains 
in Figures \ref{fig: Opt_Spd_PosEV_Expl_MIN_MAX} - \ref{fig: ExcPri_PosEV_Expl_MIN_MAX}
to
\begin{equation*}
\left[ 0, 0.75 \left( t_e - 10^{-6} \right) \right]
\quad \text{and} \quad 
\left[ 0, t_e - 10^{-6} \right]
\end{equation*}
for the left and right plots, respectively.


The qualities that we expect based upon Theorem~\ref{no Inv expl lem SS FIN}, 
and Remark \ref{lamb > 0 but bet 0 1 rem} are all present. 
We offered
some applicable 
comments in Section \ref{case 1 subsect}, 
so we only add a few new observations here. 

\begin{enumerate}[label=\roman*)]

\item All agents' inventories approach a finite limit as $t \uparrow t_e$
(see Theorem~\ref{no Inv expl lem SS FIN} and Figure \ref{fig: Opt_Inv_PosEV_Expl_MAX}). 

\item The execution price and the uncertain agents' trading rates explode as $t \uparrow t_e$
(see 
Theorem~\ref{no Inv expl lem SS FIN}, 
Remark \ref{lamb > 0 but bet 0 1 rem} and 
Figures \ref{fig: Opt_Spd_PosEV_Expl_MIN_MAX} - \ref{fig: ExcPri_PosEV_Expl_MIN_MAX}).

\item The uncertain agents' trading 
rates synchronize as $t \uparrow t_e$ 
(see 
Theorem~\ref{no Inv expl lem SS FIN}, 
Remark \ref{lamb > 0 but bet 0 1 rem}, and Figure \ref{fig: Opt_Spd_PosEV_Expl_MIN_MAX}).

\item
That an explosion in $S^{exc} \left( \omega \right)$ will occur as well as its direction becomes increasingly obvious as $t \uparrow t_e$; however, it is not necessarily clear at first 
(see
Theorem~\ref{no Inv expl lem SS FIN}, 
Remark \ref{lamb > 0 but bet 0 1 rem},
 and Figure \ref{fig: ExcPri_PosEV_Expl_MIN_MAX}).

\end{enumerate}



\begin{figure}[!tbp]
  \centering
  \subfloat{\includegraphics[scale=0.37]{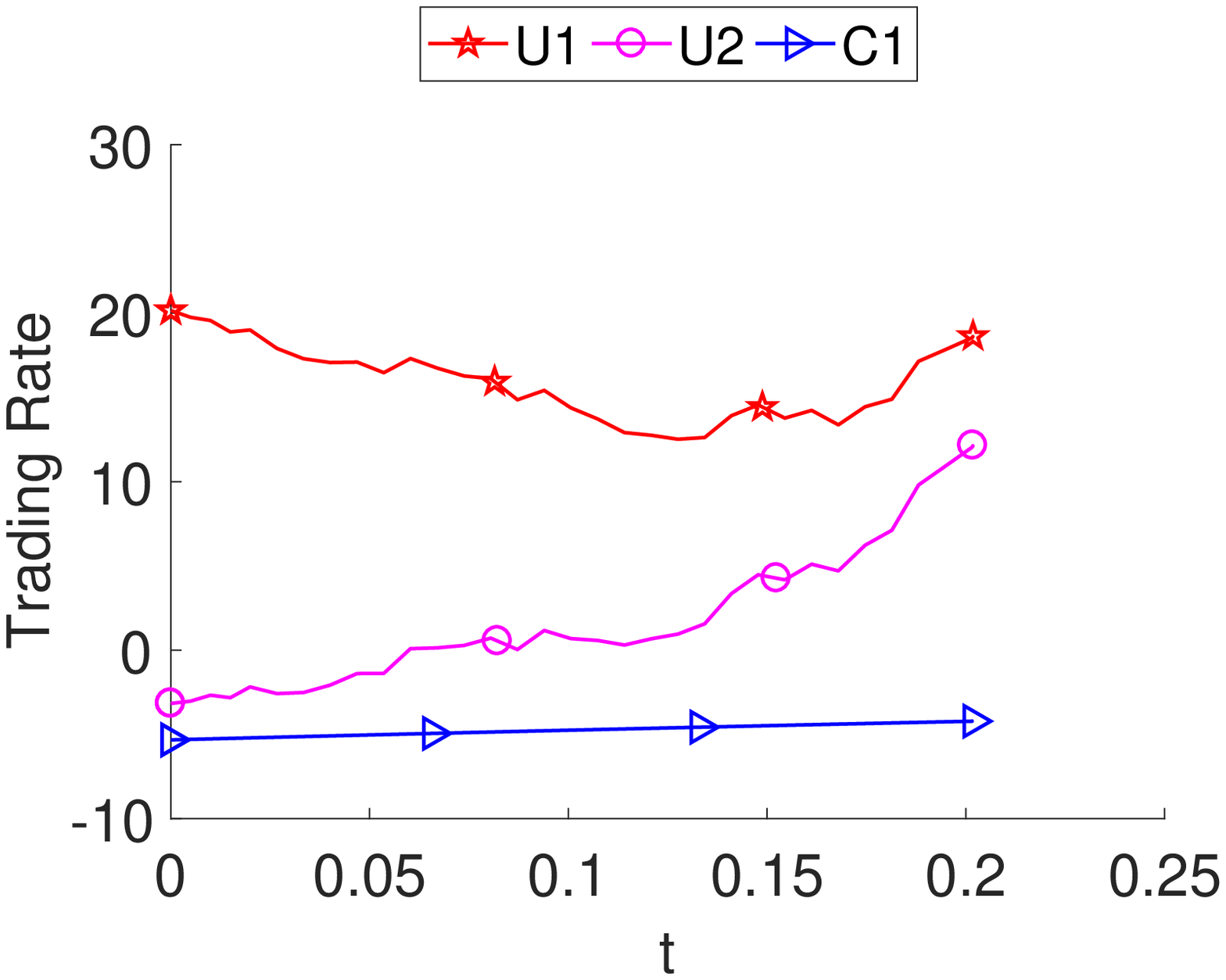}}
  \hfill
  \subfloat{\includegraphics[scale=0.37]{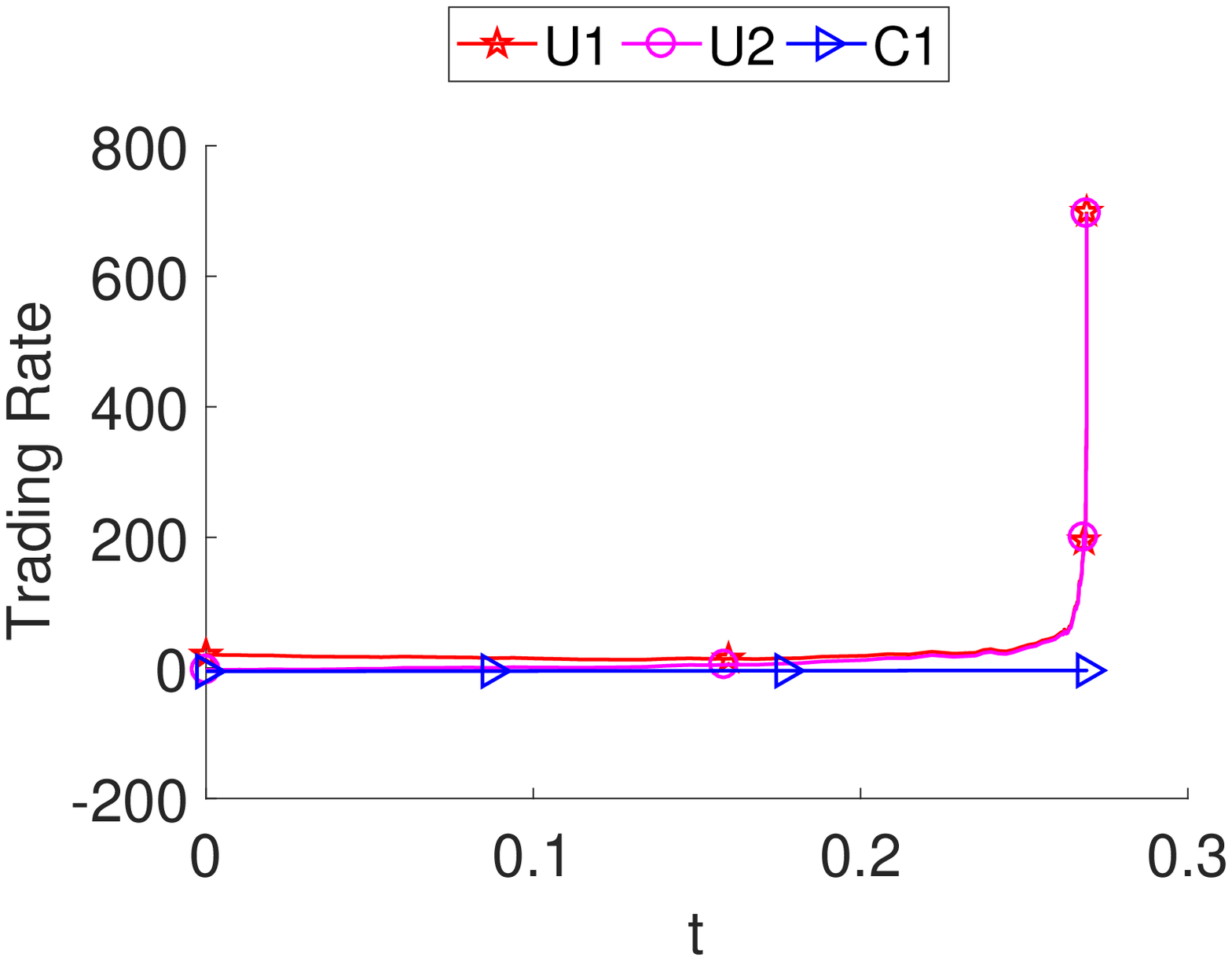}}
  \caption{Depiction of the agents' trading rates in Section \ref{case 2 subsect}.}
  \label{fig: Opt_Spd_PosEV_Expl_MIN_MAX}
\end{figure}



\begin{figure}[!tbp]
  \centering
  \subfloat{\includegraphics[scale=0.37]{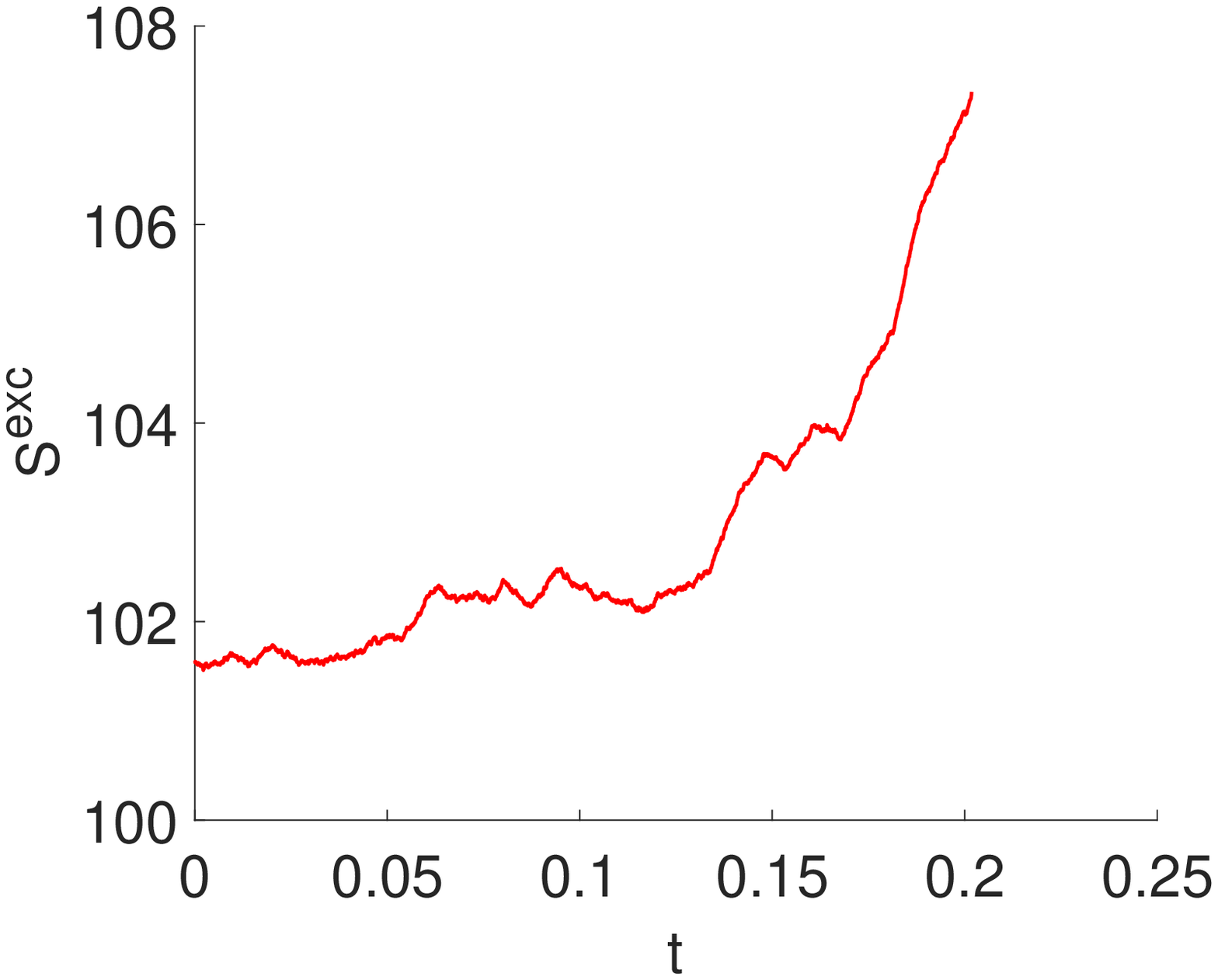}}
  \hfill
  \subfloat{\includegraphics[scale=0.37]{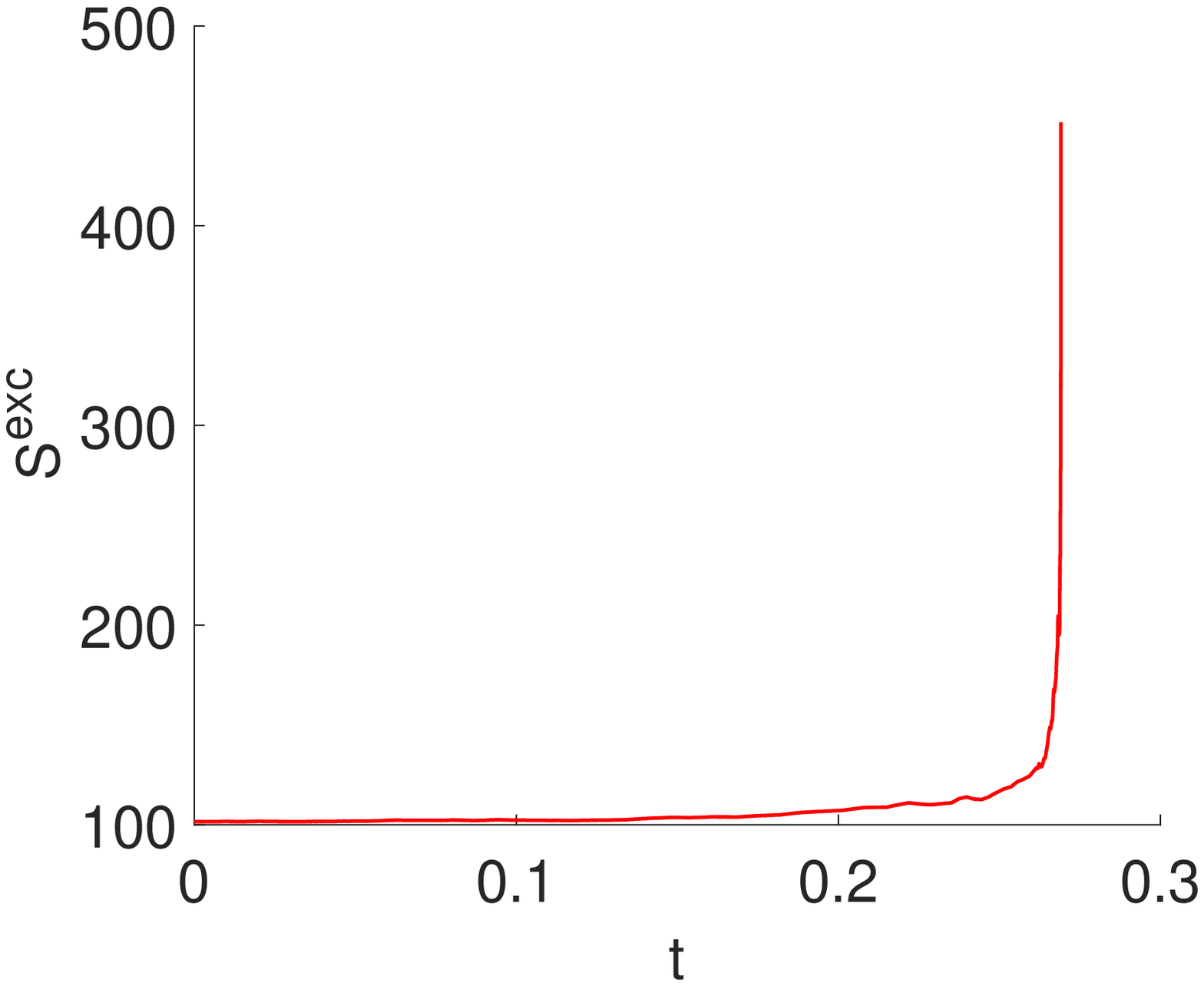}}
  \caption{Depiction of the execution price in Section \ref{case 2 subsect}.}
  \label{fig: ExcPri_PosEV_Expl_MIN_MAX}
\end{figure}


\subsection{Example 3: A mini-flash crash with high trading volume}\label{case 3 subsect}

Our mini-flash crashes can also be accompanied by high 
trading volumes (see Theorem~\ref{SS main blow up lem FIN}). 
We illustrate this in Section \ref{case 3 subsect}
by simulating a case in which 
$\det A$ has a root on $\left[ 0 , T \right]$; 
$t_e > 0$; 
the zero of $\det A$ at $t_e$ is of multiplicity 1; 
$\lambda \not \in \mathbb{Z}$; 
and $\lambda < 0$.
The behaviors of 
$S^{exc} \left( \omega \right)$,
the  $X^{\theta_{j}^\star}_{j} \left( \omega \right)$'s, 
 and the $\theta^{\star}_{j} \left( \omega \right)$'s
 are then described by 
 Corollary \ref{cert agents opt strats}
and 
Theorem~\ref{SS main blow up lem FIN}.

We especially wish to emphasize
the stochastic explosion direction and
do this in two ways.

First, we choose the same deterministic 
parameters to 
create
Figures \ref{fig: Opt_Inv_NegEV_Expl_MIN_MAX_DOWN}
- \ref{fig: ExcPri_NegEV_Expl_MIN_MAX_UP}. 
The difference is that one realization of $\tilde{W}_{\cdot}$ is used
in Figures \ref{fig: Opt_Inv_NegEV_Expl_MIN_MAX_DOWN} - \ref{fig: ExcPri_NegEV_Expl_MIN_MAX_DOWN}, 
while another is used in
Figures \ref{fig: Opt_Inv_NegEV_Expl_MIN_MAX_UP} - \ref{fig: ExcPri_NegEV_Expl_MIN_MAX_UP}.
We denote the corresponding $\omega$'s by $\omega_{up}$ and $\omega_{dn}$, 
since there 
are spikes and crashes
in the former and latter plots, respectively.

Second, Figures \ref{fig: Opt_Inv_NegEV_Expl_MIN_MAX_DOWN} - \ref{fig: ExcPri_NegEV_Expl_MIN_MAX_DOWN}
themselves suggest that the explosion direction is random. 
This is particularly true in 
Figures \ref{fig: Opt_Spd_NegEV_Expl_MIN_MAX_DOWN} - \ref{fig: ExcPri_NegEV_Expl_MIN_MAX_DOWN}, 
since we initially notice that 
the price rapidly rises as 
the uncertain agents' buying rates synchronize. 
Only moments before the mini-flash crash
do we see the price collapsing
and the
uncertain agents' aggressively selling together.

Now, we need to choose parameters such that (\ref{semi-symmetric gull te in 0 T FIN REWRITE})
is satisfied and 
\begin{equation*}
\lambda = 
\displaystyle\frac{
2
\left[
\displaystyle\sqrt{\displaystyle\frac{ \kappa }{ \eta_{tem}}}   \coth \left( \tau \left( t_e \right) \right)  
-  2 \left(  \displaystyle\frac{ K \tilde{\eta}_{per} - \eta_{per} }{ K \tilde{\eta}_{tem} - \eta_{tem} } \right)
\right] 
}
{ \left( K \tilde{\eta}_{tem} - \eta_{tem} \right) \dot{\Phi} \left( t_e \right)  } 
\end{equation*}
is a negative non-integer due to Lemmas \ref{te study lem SS FIN} and \ref{SS main blow up lem FIN}. 
Compared to Section \ref{case 2 subsect}, we set
\begin{equation*}
\tilde{\eta}_{per} = 0.5, \quad \eta_{per} = 0.5
\end{equation*}
and keep every other parameter the same. 
As in Sections \ref{case 1 subsect} - \ref{case 2 subsect}, 
we do not have in mind a special historical example here. 
Since we have only changed $\tilde{\eta}_{per}$ and $\eta_{per}$, the 
values of 
$\left( K \tilde{\eta}_{tem} - \eta_{tem} \right) \Phi \left( 0 \right)$
and $t_e$ 
do not differ from Section \ref{case 2 subsect}; however, $\lambda$ is now negative:
\begin{equation*}
\left( K \tilde{\eta}_{tem} - \eta_{tem} \right) \Phi \left( 0 \right) = 4.3302, 
\quad 
t_e = 0.2691, 
\quad \text{and} \quad
\lambda = -0.4531. 
\end{equation*}

The numbers of uncertain and certain agents 
are still two and one, respectively.
We also retain the $\left\{ U1, U2, C1 \right\}$-labeling system
from Sections \ref{case 1 subsect} - \ref{case 2 subsect}.
 Figures \ref{fig: Opt_Inv_NegEV_Expl_MIN_MAX_DOWN} and \ref{fig: Opt_Inv_NegEV_Expl_MIN_MAX_UP}
depict the agents' inventories. 
We plot the agents' trading rates in Figures \ref{fig: Opt_Spd_NegEV_Expl_MIN_MAX_DOWN} and 
\ref{fig: Opt_Spd_NegEV_Expl_MIN_MAX_UP}. 
The execution price appears in Figures \ref{fig: ExcPri_NegEV_Expl_MIN_MAX_DOWN} and
\ref{fig: ExcPri_NegEV_Expl_MIN_MAX_UP}. 
To help with our visualization, the time domains in the 
left plots in Figures \ref{fig: Opt_Inv_NegEV_Expl_MIN_MAX_DOWN} - \ref{fig: ExcPri_NegEV_Expl_MIN_MAX_DOWN}
and 
Figures \ref{fig: Opt_Inv_NegEV_Expl_MIN_MAX_UP} - \ref{fig: ExcPri_NegEV_Expl_MIN_MAX_UP}
are truncated to 
$\left[ 0 , 0.94\left( t_e - 10^{-6} \right) \right]$
and 
$\left[ 0 , 0.75\left( t_e - 10^{-6} \right) \right]$, 
respectively.


Our observations regarding 
Figures \ref{fig: Opt_Inv_NegEV_Expl_MIN_MAX_DOWN} - \ref{fig: ExcPri_NegEV_Expl_MIN_MAX_UP}
are in agreement 
with Corollary \ref{cert agents opt strats} and Theorem~\ref{SS main blow up lem FIN}. 
We have already made note of many important aspects in Sections \ref{case 1 subsect} - \ref{case 2 subsect}
and only remark upon the new details.

\begin{enumerate}[label=\roman*)]

\item The execution price, as well as the uncertain agents' inventories and trading rates, all explode in 
the same direction as $t \uparrow t_e$ 
(see Theorem~\ref{SS main blow up lem FIN} and Figures \ref{fig: Opt_Inv_NegEV_Expl_MIN_MAX_DOWN} - \ref{fig: ExcPri_NegEV_Expl_MIN_MAX_UP}).

\item The explosions take place at the deterministic time $t_e$
(see Theorem~\ref{SS main blow up lem FIN} and Figures \ref{fig: Opt_Inv_NegEV_Expl_MIN_MAX_DOWN} - \ref{fig: ExcPri_NegEV_Expl_MIN_MAX_UP}).

\item The explosion direction depends on $\omega \in \tilde{\Omega}$ 
(see Theorem~\ref{SS main blow up lem FIN} and Figures \ref{fig: Opt_Inv_NegEV_Expl_MIN_MAX_DOWN} - \ref{fig: ExcPri_NegEV_Expl_MIN_MAX_UP}).

\item The 
explosion direction cannot be known with complete certainty before $t_e$ 
(see Theorem~\ref{SS main blow up lem FIN} and Figures \ref{fig: Opt_Inv_NegEV_Expl_MIN_MAX_DOWN} - \ref{fig: ExcPri_NegEV_Expl_MIN_MAX_UP}).

\item
The explosion rates in the price and uncertain agents' trading rates in Section \ref{case 2 subsect}
are slower than in Section \ref{case 3 subsect}
(see 
Figures \ref{fig: Opt_Spd_PosEV_Expl_MIN_MAX} - \ref{fig: ExcPri_PosEV_Expl_MIN_MAX}, 
Figures \ref{fig: Opt_Spd_NegEV_Expl_MIN_MAX_DOWN} - \ref{fig: ExcPri_NegEV_Expl_MIN_MAX_DOWN}, 
and Figures \ref{fig: Opt_Spd_NegEV_Expl_MIN_MAX_UP} - \ref{fig: ExcPri_NegEV_Expl_MIN_MAX_UP}). 
We did not explicitly state this previously; however, this is to be expected 
since trading rates are integrable in Section \ref{case 2 subsect} but not in Section \ref{case 3 subsect}.

\end{enumerate}



\begin{figure}[!tbp]
  \centering
  \subfloat{\includegraphics[scale=0.37]{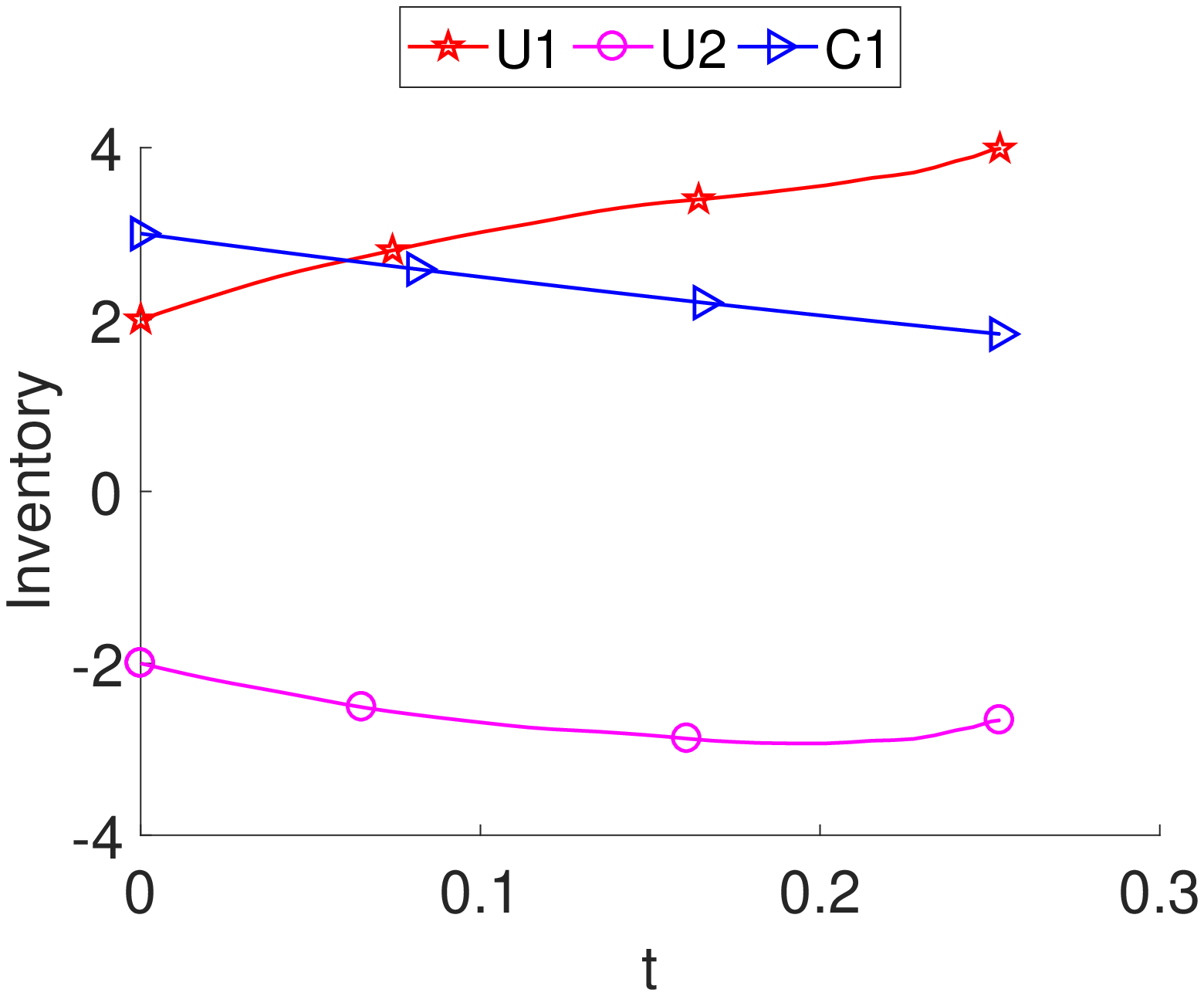}}
  \hfill
  \subfloat{\includegraphics[scale=0.37]{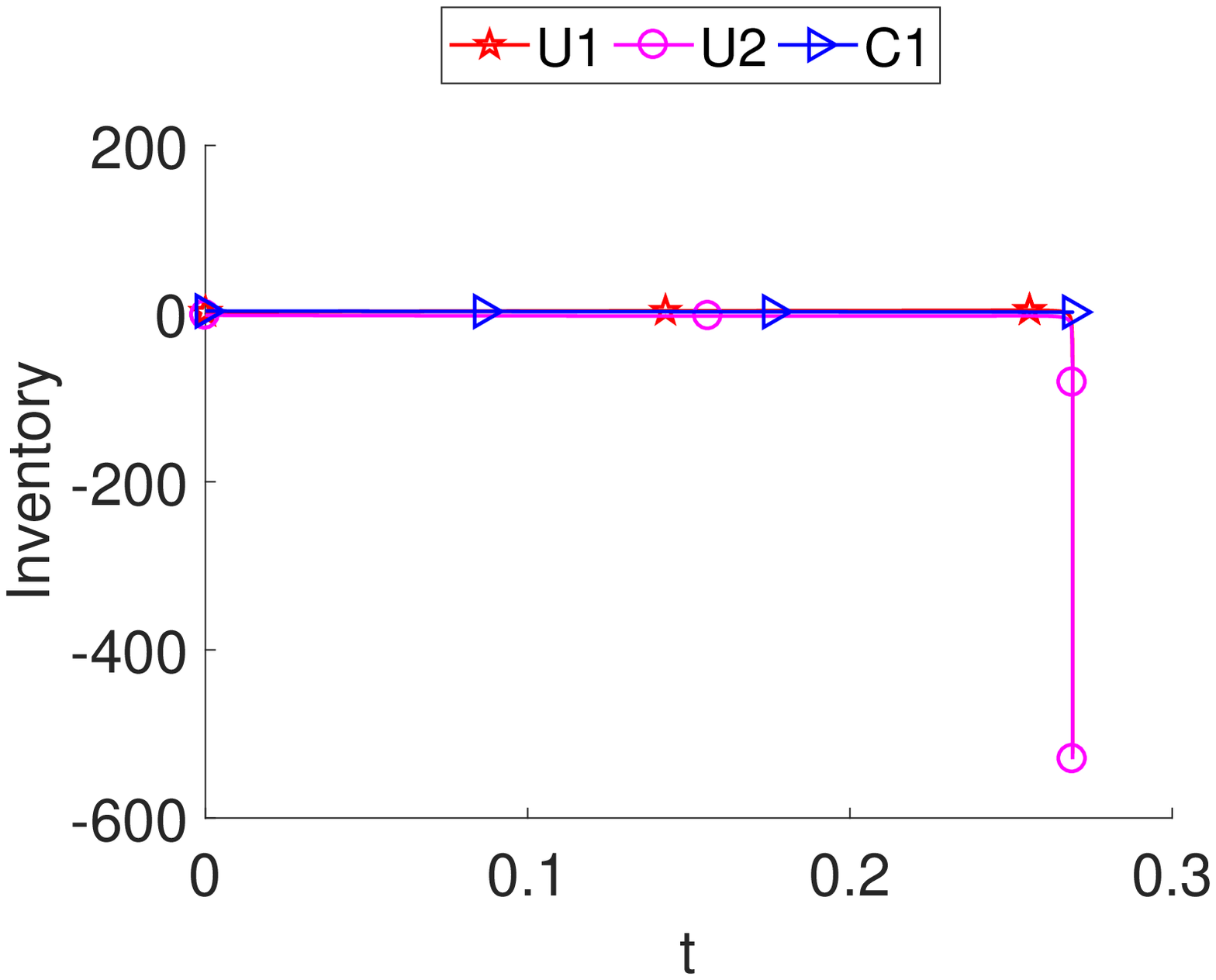}}
  \caption{Depiction of the agents' inventories for $\omega_{dn}$ in Section \ref{case 3 subsect}.}
  \label{fig: Opt_Inv_NegEV_Expl_MIN_MAX_DOWN}
\end{figure}



\begin{figure}[!tbp]
  \centering
  \subfloat{\includegraphics[scale=0.37]{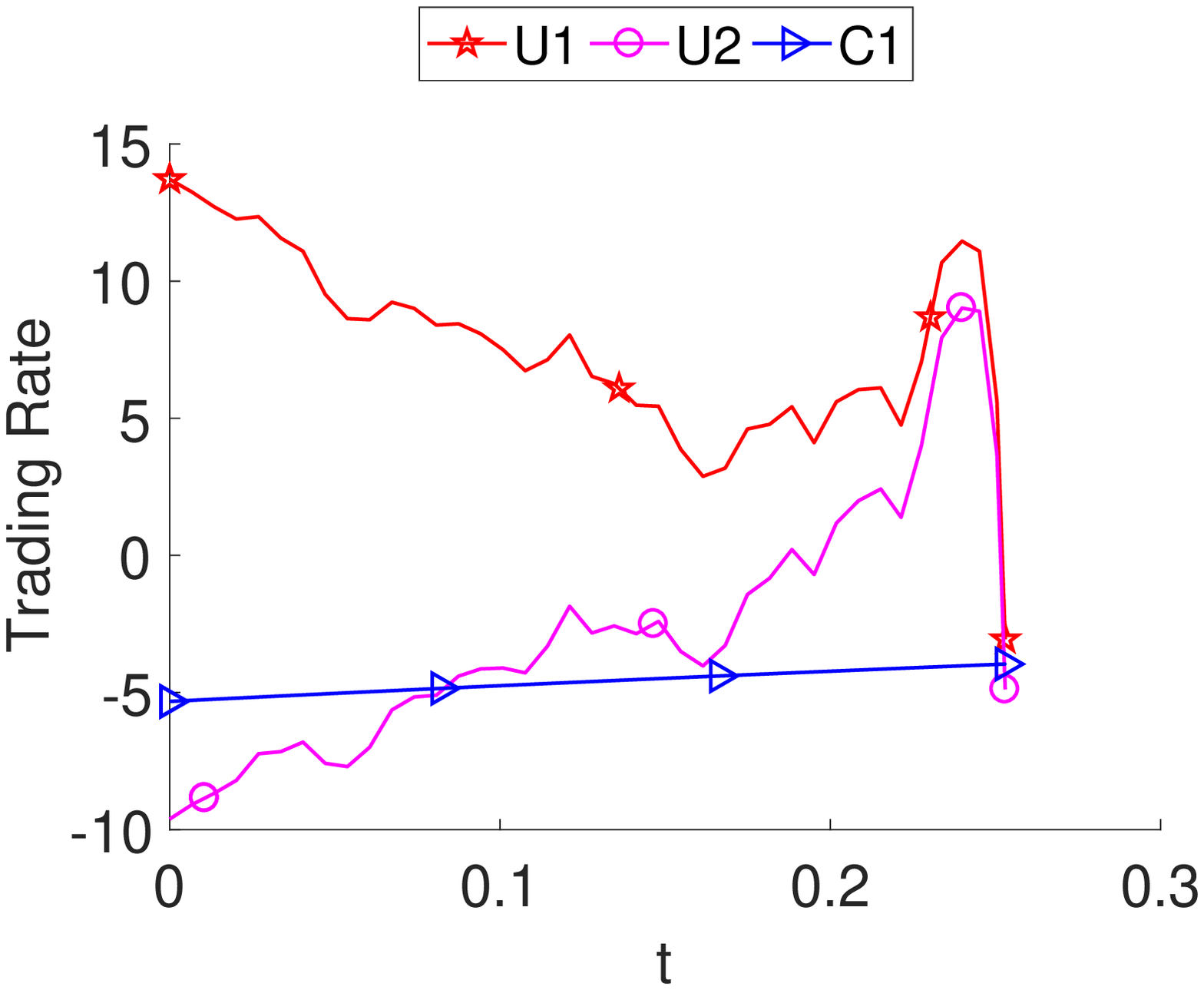}}
  \hfill
  \subfloat{\includegraphics[scale=0.37]{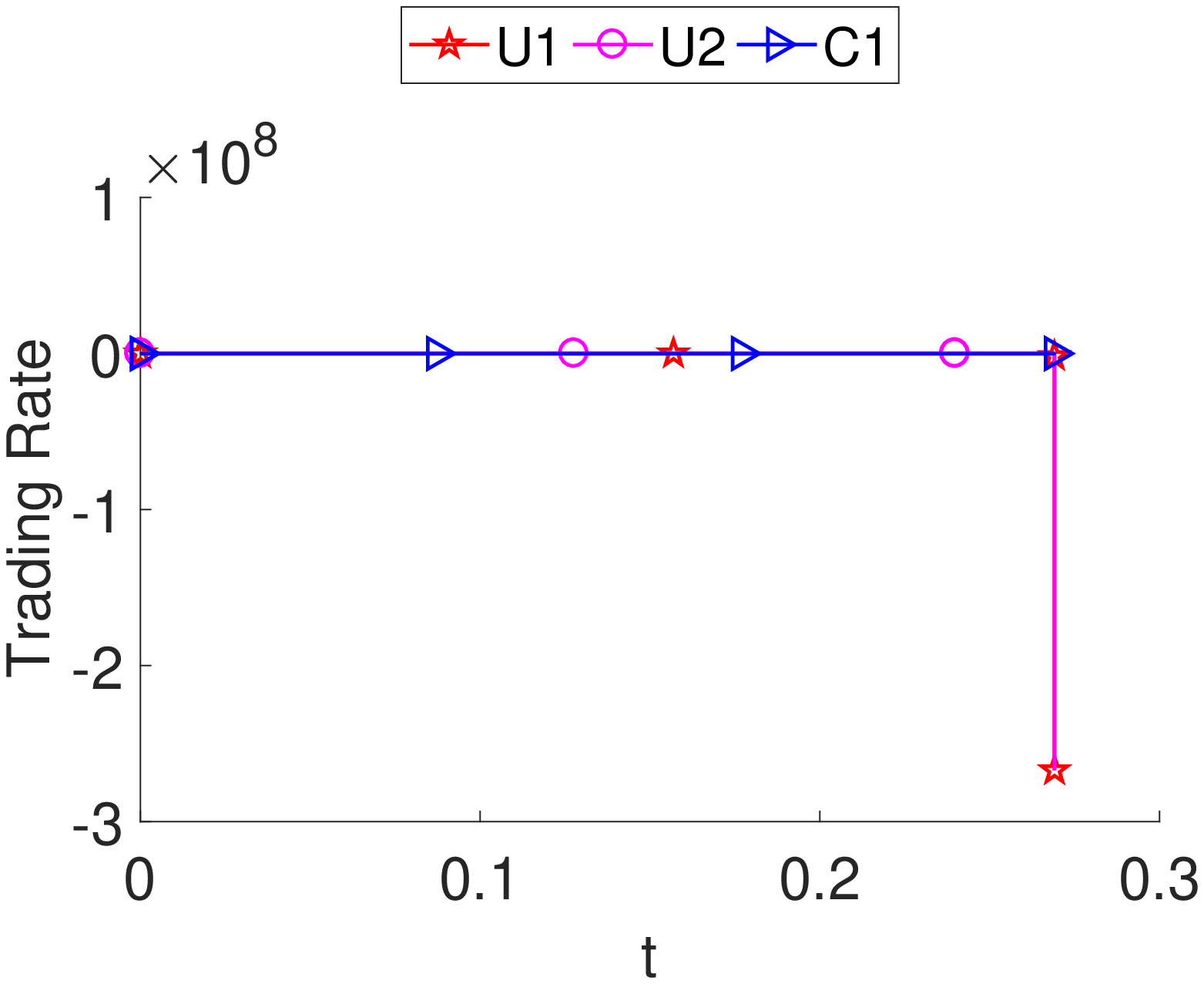}}
  \caption{Depiction of the agents' trading rates for $\omega_{dn}$ in Section \ref{case 3 subsect}.}
  \label{fig: Opt_Spd_NegEV_Expl_MIN_MAX_DOWN}
\end{figure}



\begin{figure}[!tbp]
  \centering
  \subfloat{\includegraphics[scale=0.37]{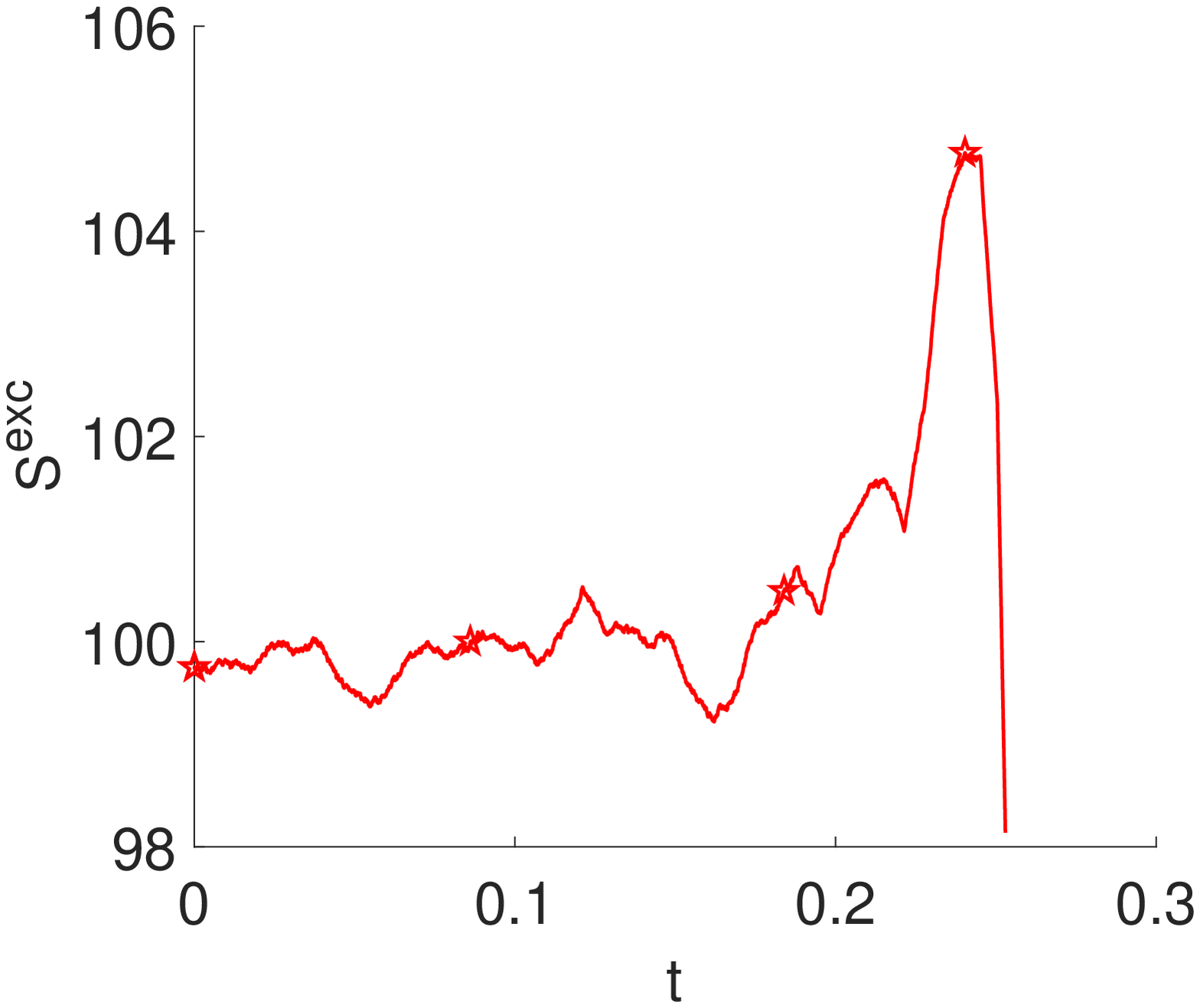}}
  \hfill
  \subfloat{\includegraphics[scale=0.37]{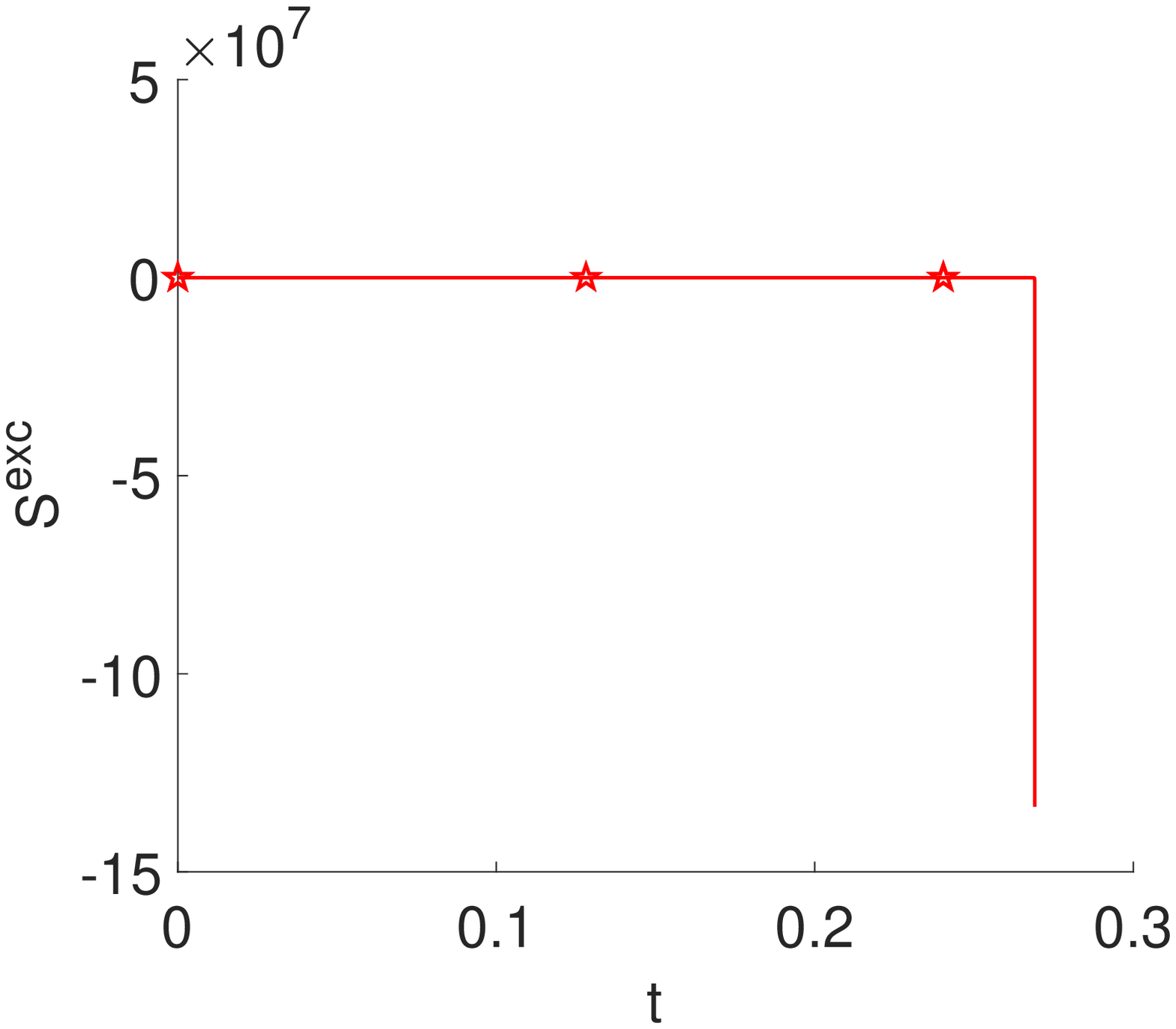}}
  \caption{Depiction of the execution price for $\omega_{dn}$ in Section \ref{case 3 subsect}.}
  \label{fig: ExcPri_NegEV_Expl_MIN_MAX_DOWN}
\end{figure}




\begin{figure}[!tbp]
  \centering
  \subfloat{\includegraphics[scale=0.37]{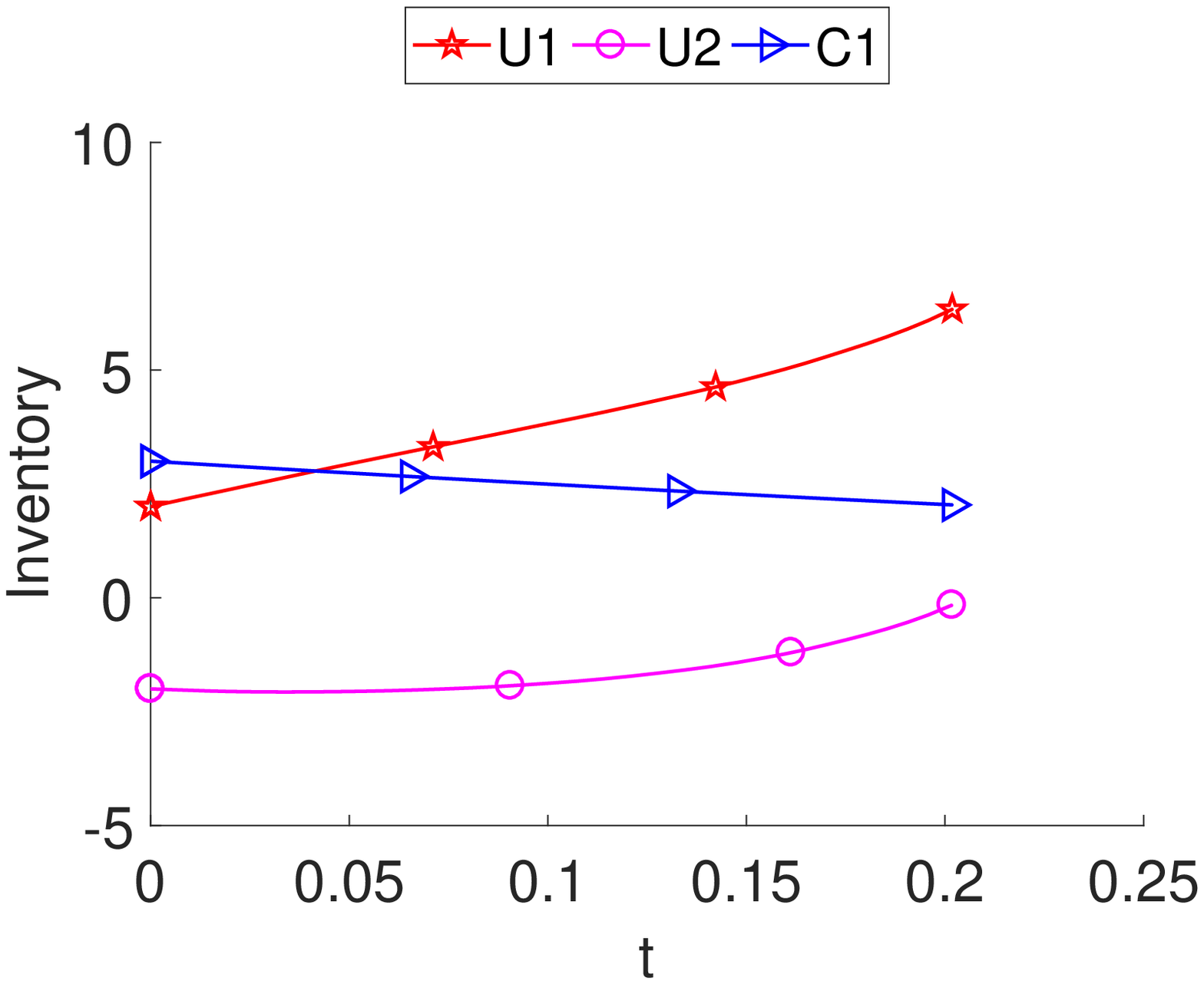}}
  \hfill
  \subfloat{\includegraphics[scale=0.37]{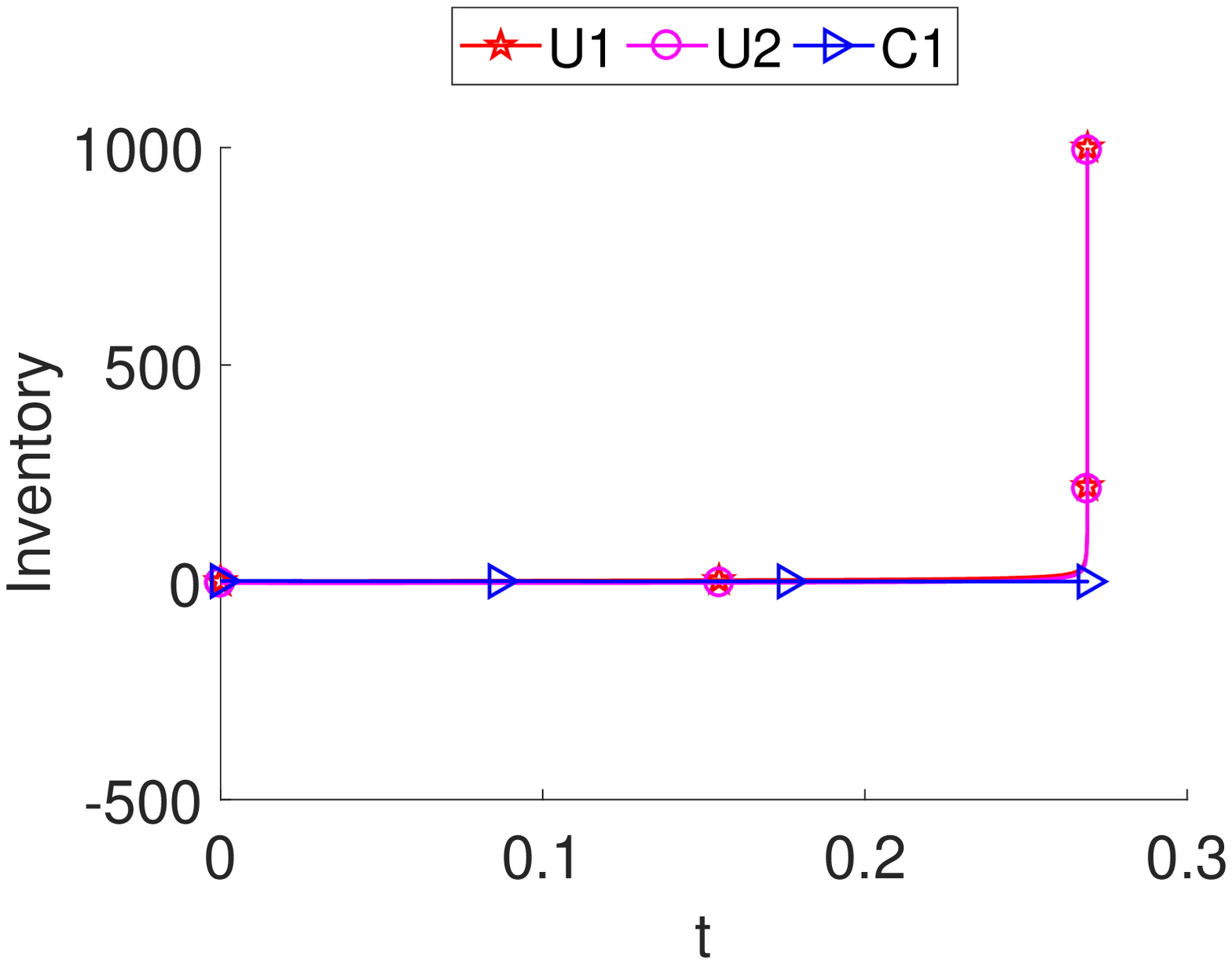}}
  \caption{Depiction of the agents' inventories for $\omega_{up}$ in Section \ref{case 3 subsect}.}
  \label{fig: Opt_Inv_NegEV_Expl_MIN_MAX_UP}
\end{figure}



\begin{figure}[!tbp]
  \centering
  \subfloat{\includegraphics[scale=0.37]{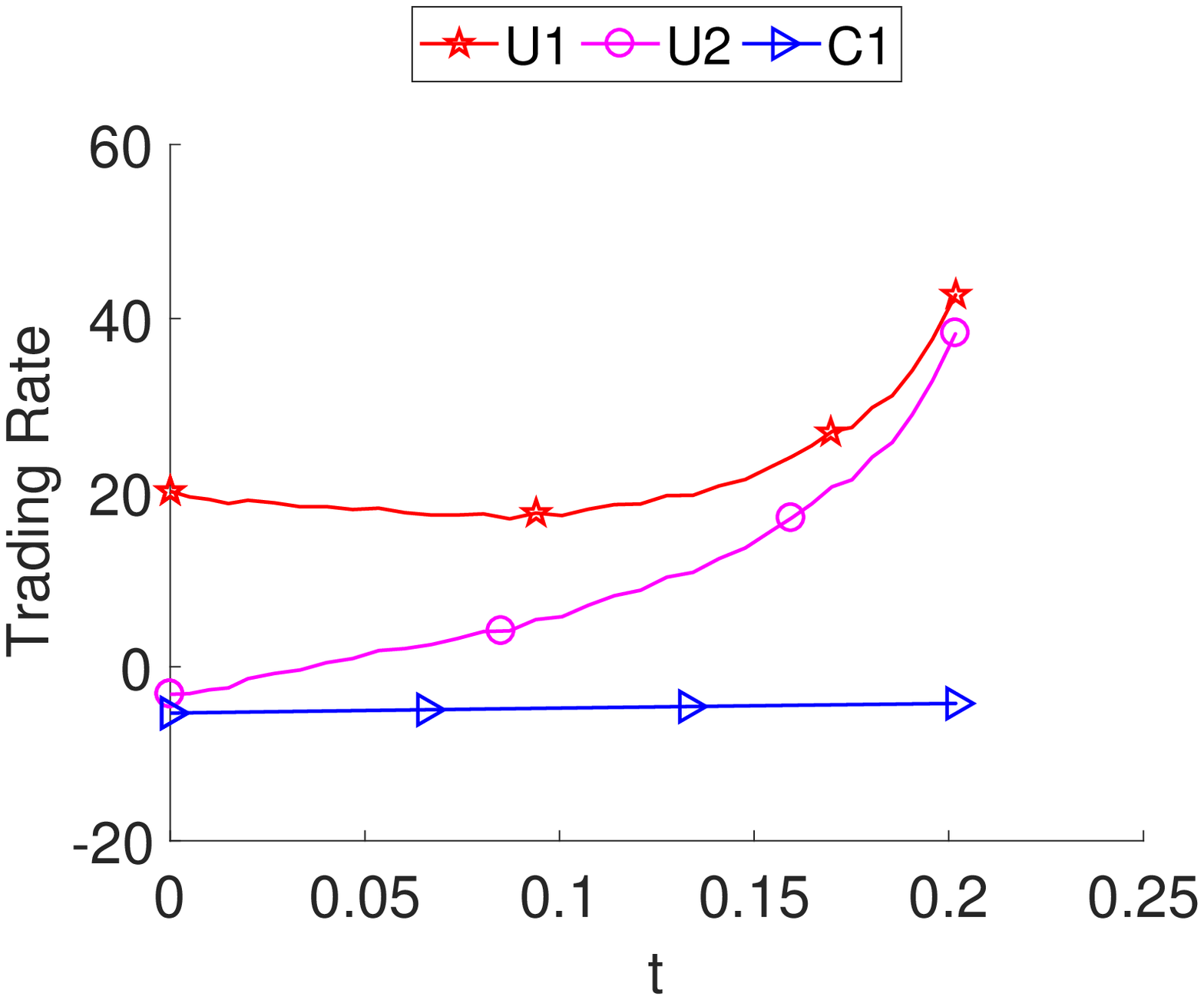}}
  \hfill
  \subfloat{\includegraphics[scale=0.37]{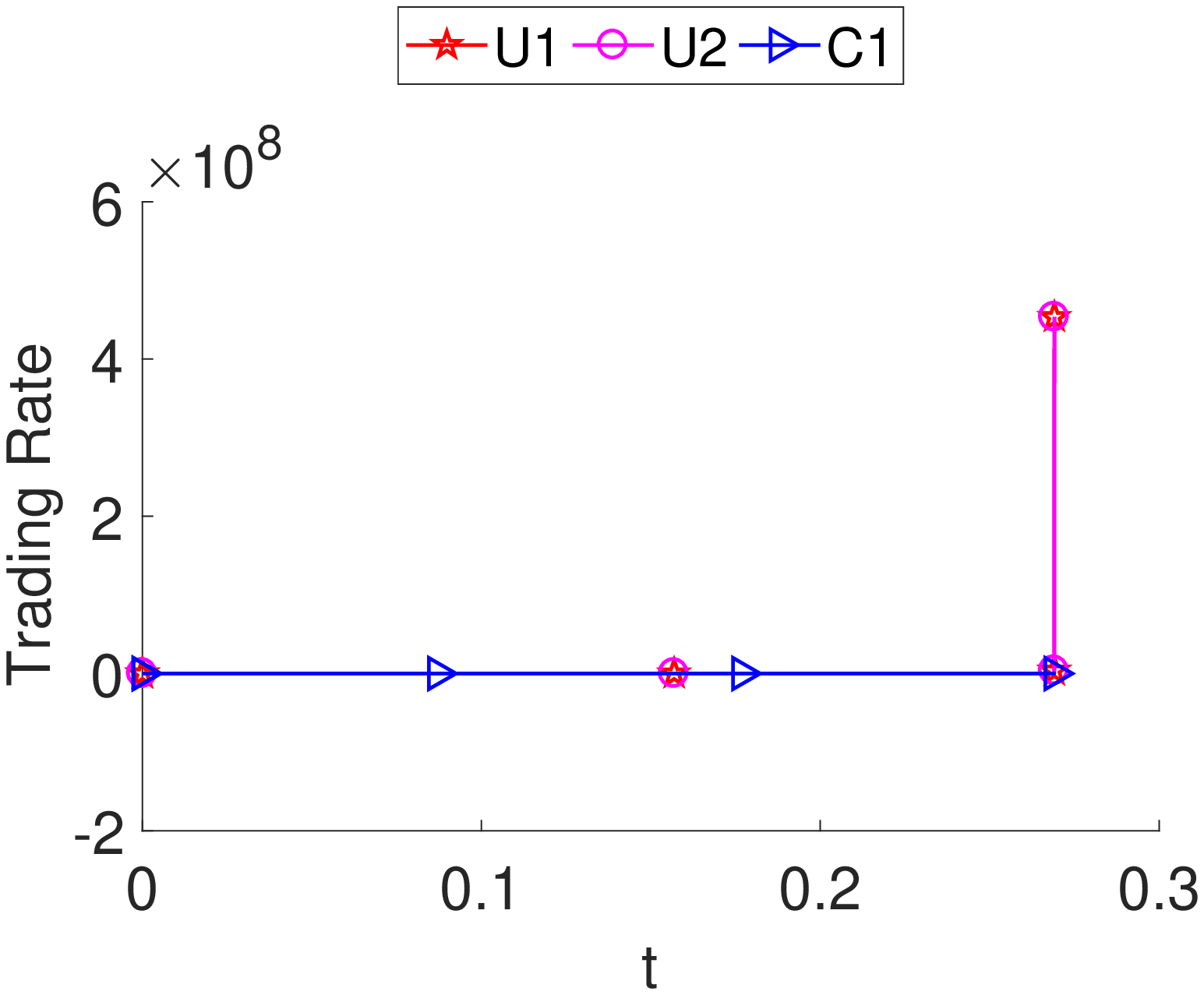}}
  \caption{Depiction of the agents' trading rates for $\omega_{up}$ in Section \ref{case 3 subsect}.}
  \label{fig: Opt_Spd_NegEV_Expl_MIN_MAX_UP}
\end{figure}



\begin{figure}[!tbp]
  \centering
  \subfloat{\includegraphics[scale=0.37]{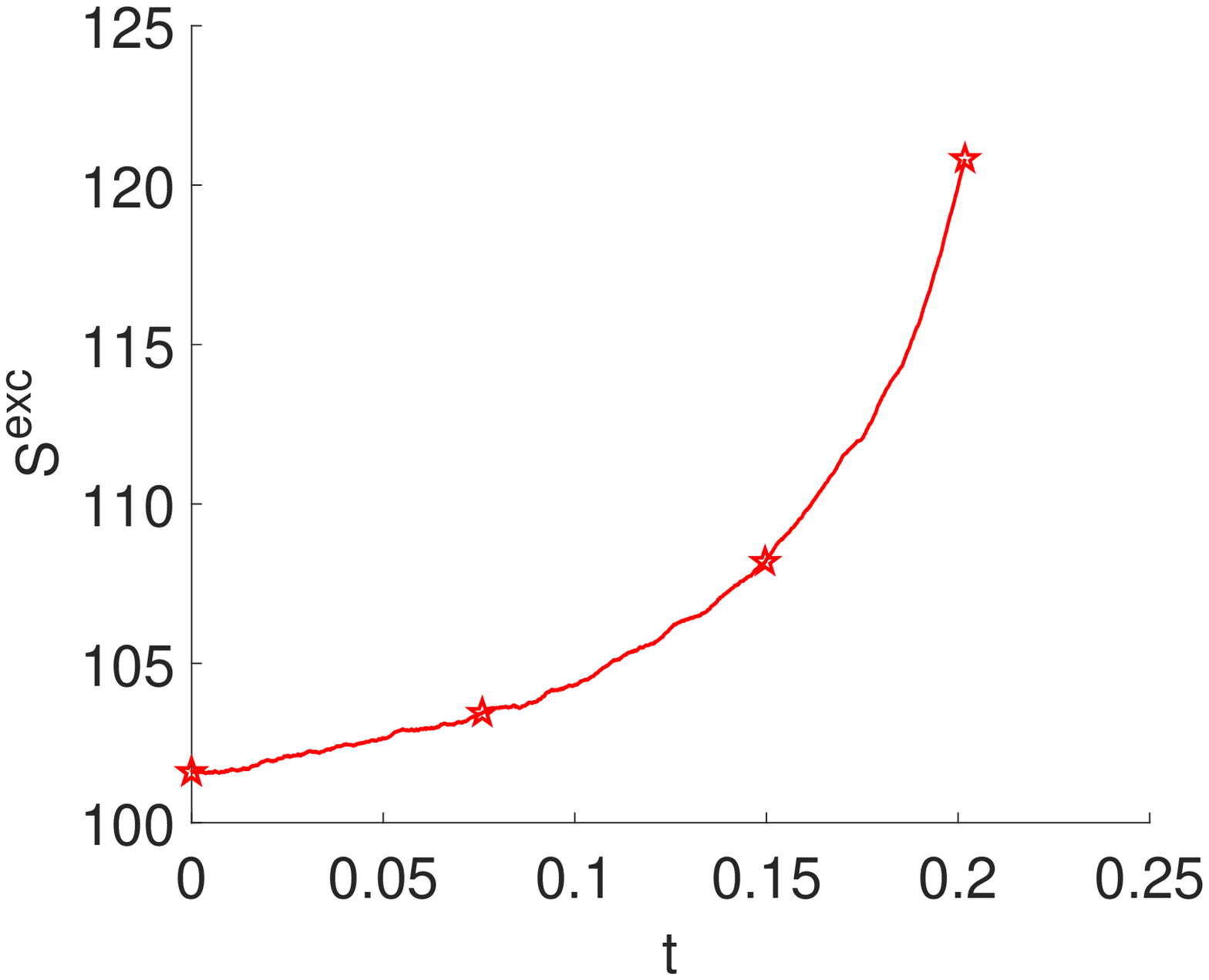}}
  \hfill
  \subfloat{\includegraphics[scale=0.37]{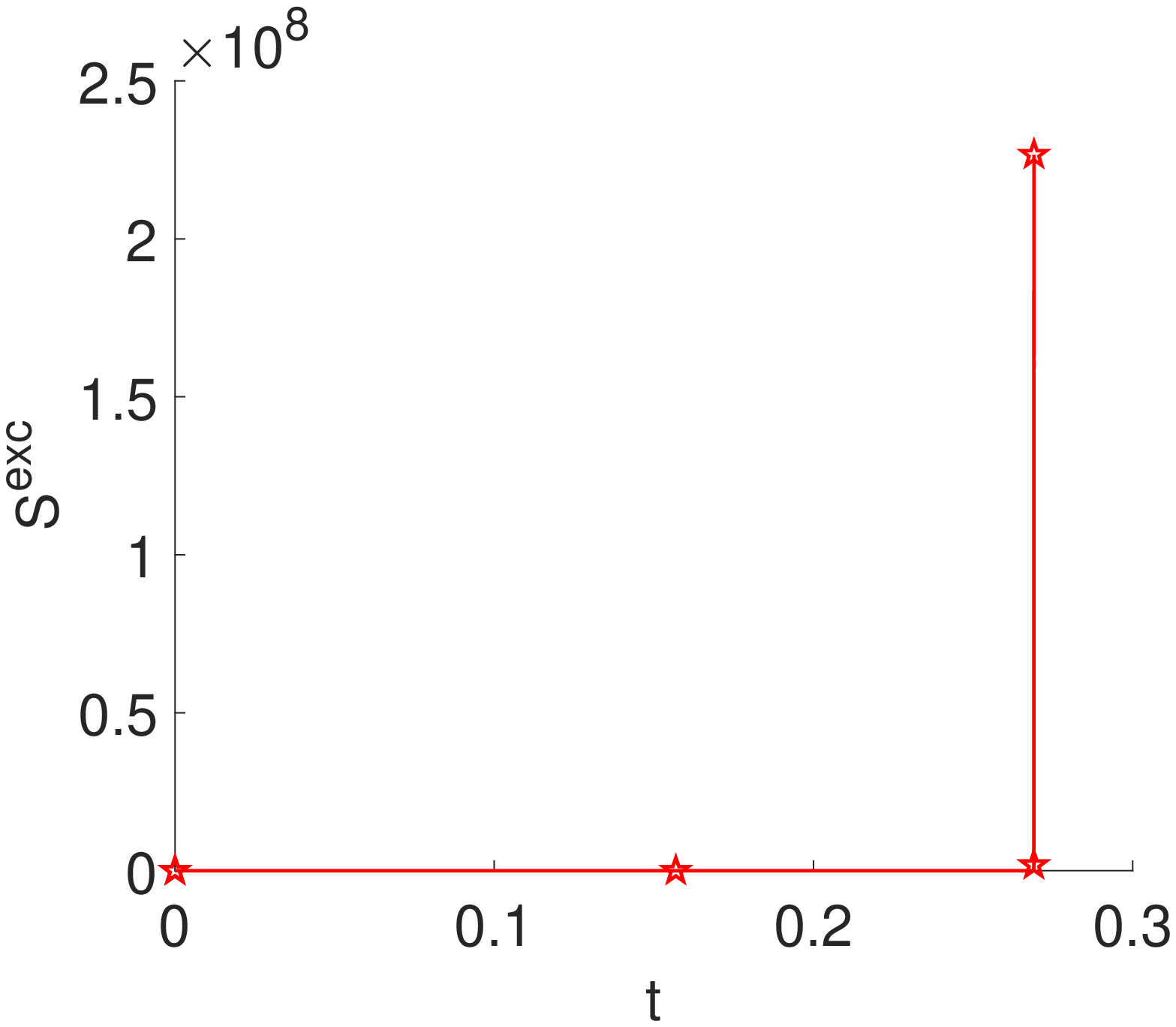}}
  \caption{Depiction of the execution price for $\omega_{up}$ in Section \ref{case 3 subsect}.}
  \label{fig: ExcPri_NegEV_Expl_MIN_MAX_UP}
\end{figure}

\newpage

\section{conclusion}

In this paper we show how mini-flash crashes might occur when agents learn and make their decisions 
based upon misspecified models. We give a necessary and sufficient condition for a mini-flash crash to occur and observe that if the agents

\begin{enumerate}[label=\roman*)]

\item are too uncertain about their prior information
\item sufficiently underestimate the aggregate temporary impact
\item have long trading periods
\item have low risk aversion

\end{enumerate}
then mini flash crashes occur.

Our numerical section has three main examples illustrating our three main results. In the first example, we see that not all human errors directly 
cause mini-flash crashes illustrating Proposition \ref{det A no root no blow up}. Despite the regularity of the human errors 
on a market-wide basis, 
individual securities may rarely experience such an event. 
Similarly, traders' models and strategies do
roughly achieve their intended goals much of the time, 
as it is observed in the real markets.

The two other examples are about the nature of the mini-flash crashes.
If a mini-flash crash does occur, 
it almost surely does so because of 
``endogenous triggering effects.''
As predicted by Theorems \ref{SS main blow up lem FIN} and \ref{no Inv expl lem SS FIN}), some of our agents also ``converge on the same strategy''
during mini-flash crashes: In certain cases, the agents driving these
events all buy or sell together with the same (exploding) growth rate. The two theorems and the examples illustrating them are to demonstrate that mini-flash crashes
can be accompanied by both high and low trading volumes.

\appendix

\section{Section \ref{mod details sec} Proofs}\label{sect mod details sec proofs app}

\subsection{Proof of Lemma \ref{1 pl soln acc to play model}}\label{1 pl soln acc to play model PRF SS}

{\bf Step \mystepTwo{step 1 orig}}:
Denote the usual $P_j$-augmentation of $\left\{ \mathcal{F}_{j,t}^{unf} \right\}_{0 \leq t \leq T}$ by $\left\{ \tilde{\mathcal{F}}_{j,t}^{unf} \right\}_{0 \leq t \leq T}$. 
Let $\tilde{\mathcal{A}}_{j}$ be the space of $\tilde{\mathcal{F}}_{j,t}^{unf}$-progressively measurable processes $\theta_{j}$ such that (\ref{admiss strat int}) and (\ref{admiss strat term}) hold. 
We, again, define the process $X^{\theta_j}_j$ by (\ref{x theta def}) for 
any strategy $\theta_{j} \in \tilde{\mathcal{A}_{j}}$. 
Agent $j$'s auxiliary problem is to maximize 
\begin{align}\label{orig max prob AUX}
 E^{P_j} \left[ - \displaystyle\int_{0}^T \theta_{j,t} \, S_{j,\theta_{j},t}^{exc}  \, dt  -  \displaystyle\frac{\kappa_{j}}{2} \displaystyle\int_{0}^{T}  \left( X^{\theta_{j}}_{j,t} \right)^2 \, dt \right]
\end{align}
 over $\theta_{j} \in \tilde{\mathcal{A}}_{j}$. 
It is not difficult to show that
\begin{equation}\label{eq:ldnl}
E^{P_j} \left[ \beta_j \Big| \tilde{\mathcal{F}}_{j,t}^{unf} \right] = E^{P_j} \left[ \beta_{j} \Big| \mathcal{F}_{j,t}^{unf} \right]
\quad P_{j}-\text{a.s.}
\end{equation}

{\bf Step \mystepTwo{step 4 orig}:}
By (\ref{exc price model j}) and (\ref{x theta def}),
\begin{align*}
- \displaystyle\int_{0}^T \theta_{j,t} \, S_{j,\theta_{j},t}^{exc}  \, dt &= - \displaystyle\int_{0}^T \theta_{j,t} \, S_{j,t}^{unf}  \, dt - \displaystyle\int_{0}^T \theta_{j,t} \left[ \eta_{j,per} \left(  X_{j,t}^{\theta_{j}} - x_{j}   \right) + \displaystyle\frac{1}{2} \eta_{j,tem}\theta_{j,t} \right] dt 
\end{align*}
for $\theta_j \in \tilde{\mathcal{A}}_j$. Section 7.4 of \cite{lipt+shiry+i} and (\ref{fund price model j}) imply that 
the process $ \left\{ \overline{W}_{j,t} \right\}_{0 \leq t \leq T}$
with 
\begin{align*}
\overline{W}_{j,t}  &\triangleq S_{j,t}^{unf} - S_{j,0} - \displaystyle\int_0^t E^{P_j} \left[ \beta_j \big| \mathcal{F}_{j,s}^{unf} \right] \, ds  
\end{align*}
is an $\mathcal{F}_{j,t}^{unf} $-Wiener process under $P_j$
 and
\begin{align}\label{ce sunaff dyn}
S_{j,t}^{unf} = S_{j,0}+  \displaystyle\int_0^t E^{P_j} \left[ \beta_j \big| \mathcal{F}_{j,s}^{unf}  \right] \, ds +\overline{W}_{j,t} .
\end{align} 
After integrating by parts and recalling (\ref{x theta def}) and (\ref{ce sunaff dyn}), we get 
\begin{align*}
&E^{P_{j}} \left[ - \displaystyle\int_{0}^T \theta_{j,t} \, S_{j,t}^{unf}  \, dt  \right] \\
&\qquad =  E^{P_{j}} \left[  -X_{j,T}^{\theta_{j}}  S_{j,T}^{unf}   + \displaystyle\int_{0}^T X_{j,t}^{\theta_{j}} \, E^{P_j} \left[ \beta_j \big| \mathcal{F}_{j,t}^{unf} \right] \, dt  \right] + x_{j} S_{j,0}. 
\end{align*}
We also have 
\begin{align*}
&E^{P_{j}} \left[ - \displaystyle\int_{0}^T \theta_{j,t} \left[ \eta_{j,per} \left(  X_{j,t}^{\theta_{j}} - x_{j}   \right) + \displaystyle\frac{1}{2} \eta_{j,tem}\theta_{j,t} \right] dt \right]\\
&\qquad = E^{P_{j}} \left[  - \displaystyle\frac{1}{2} \eta_{j,per} \left( X_{j,T}^{\theta_{j}}  -  x_{j} \right)^2  - \displaystyle\frac{1}{2} \eta_{j,tem}\displaystyle\int_{0}^T \theta_{j,t}^2 \, dt  \right]. 
\end{align*}

Now $X_{j,T}^{\theta_{j}} =0 $ $P_{j}$-a.s. by the definition of $\tilde{\mathcal{A}}_j$ in Step \ref{step 1 orig}.
 Since $x_{j}$, $S_{j,0}$ and $S_{j,T}^{unf}$ do not depend on Agent $j$'s choice of $\theta_{j} \in \tilde{\mathcal{A}}_{j}$, \eqref{eq:ldnl} implies that $\theta_j^{\star}$ maximizes (\ref{orig max prob AUX}) over $\theta_j \in \tilde{\mathcal{A}}_j$ if and only if it maximizes 
\begin{align}\label{orig max prob AUX 2}
E^{P_j} \left[ \displaystyle\int_{0}^T X_{j,t}^{\theta_{j}} \, E^{P_j} \left[ \beta_j \big| \tilde{\mathcal{F}}_{j,t}^{unf} \right] \, dt  - \displaystyle\frac{1}{2} \eta_{j,tem}\displaystyle\int_{0}^T \theta_{j,t}^2 \, dt -  \displaystyle\frac{\kappa_{j}}{2} \displaystyle\int_{0}^{T}  \left( X^{\theta_{j}}_{j,t} \right)^2 \, dt\right] .
\end{align}

Clearly, $\theta_j^{\star}$ maximizes (\ref{orig max prob AUX 2})
over $\theta_j \in \tilde{\mathcal{A}}_j$ if and only if it minimizes 
 \begin{align}\label{orig max prob AUX FIN}
E^{P_j}  \left[ \frac{1}{2} \displaystyle\int_{0}^{T}   \left(   X_{j,t}^{\theta_j} - \frac{ E^{P_j} \left[ \beta_j \big| \tilde{\mathcal{F}}_{j,t}^{unf}\right]  }{\kappa_{j}} \right)^2 \, dt + \displaystyle\frac{\eta_{j,tem}}{2 \kappa_{j}}  \displaystyle\int_{0}^{T} \theta_{j,t}^2   \, dt \right] .
 \end{align}

{\bf Step \mystepTwo{step 5 orig}:}
 Recall the definition of $\tau_j(\cdot)$ from \eqref{tau j notn} and let 
 \begin{align}\label{beta hat kern defn FIN}
 \begin{array}{ll}
K_j \left( t , s \right) \triangleq  \displaystyle\sqrt{\displaystyle\frac{ \kappa_{j}}{ \eta_{j,tem}}} \left( \displaystyle\frac{ \sinh \left( \tau_{j} \left( s \right) \right)}{  \cosh\left(  \tau_{j} \left( t \right) \right) - 1} \right) , & \quad 0 \leq t \leq s < T  \\
& \\
\hat{\beta}_{j,t}  \triangleq E^{P_j} \left[ \displaystyle\frac{1}{\kappa_{j}} \left( 1 - \displaystyle\frac{1}{\cosh \left( \tau_{j} \left(t \right) \right)}  \right) \right.  & \\
&\\
\qquad \qquad \qquad \cdot \left. \displaystyle\int_t^T E^{P_j} \left[ \beta_j \big| \tilde{\mathcal{F}}_{j,s}^{unf}\right]   K_j \left( t , s \right) \, ds \, \Bigg| \, \tilde{\mathcal{F}}_{j,t}^{unf}\right] , &\quad  t \in \left[ 0 , T \right).
\end{array}
\end{align}
We see from Theorem 3.2 
 of \cite{bank+son+voss} 
 that (\ref{orig max prob AUX FIN}) has a unique solution $\theta_{j}^{\star} \in \tilde{\mathcal{A}}_{j}$.
 Moreover, the corresponding 
 optimal inventory process $X^{\theta^\star_j}_j$ satisfies the linear ODE 
\begin{align}\label{x theta star SON VOSS}
d X^{\theta^{\star}_j}_{j,t} &= \displaystyle\sqrt{\displaystyle\frac{ \kappa_{j}}{ \eta_{j,tem}}} \coth \left( \tau_{j} \left( t \right) \right) \left( \hat{\beta}_{j,t} - X^{\theta^\star_j}_{j,t} \right) \, dt \notag \\
X^{\theta^{\star}_j}_{j,0} &= x_j 
\end{align}
$d P_{j} \otimes dt$-a.s. on $\Omega_{j} \times \left[ 0 , T \right)$.

 Using Fubini's theorem and \eqref{eq:ldnl}, we get that 
  \begin{align}\label{beta hat simplif}
 E^{P_j} \left[  \displaystyle\int_t^T E^{P_j} \left[ \beta_j \Big| \tilde{\mathcal{F}}_{j,s}^{unf}\right]   K_j \left( t , s \right) \, ds \, \Bigg| \, \tilde{\mathcal{F}}_{j,t}^{unf}\right]  =  E^{P_j} \left[ \beta_{j} \Big| \mathcal{F}_{j,t}^{unf} \right] , 
 \quad P_j-\text{a.s.}
\end{align}
The $\tanh$ half-angle formula together with (\ref{cond beta raw filt form}) and (\ref{beta hat kern defn FIN}) imply that
(\ref{x theta star SON VOSS}) can be re-written as 
\begin{align}\label{x theta star no omega}
\theta^{\star}_{j,t} &= - \displaystyle\sqrt{\displaystyle\frac{ \kappa_{j}}{ \eta_{j,tem}}}  
\coth \left( \tau_{j} \left( t \right) \right) X^{\theta^\star_j}_{j,t}   \notag \\
&\qquad + 
\displaystyle\frac{  \tanh \left( \tau_{j} \left(t \right) /2 \right) \left[ \mu_{j}  + \nu_{j}^2 \left( S_{j,t}^{unf}   - S_{j,0}\right) \right]  }{ \displaystyle\sqrt{ \eta_{j,tem} \kappa_{j}}   \left( 1 + \nu_{j}^2 t \right)} 
, \quad t \in \left( 0, T \right)
  \notag \\
X^{\theta^{\star}_j}_{j,0}   &= x_j . 
\end{align}

{\bf Step \mystepTwo{step 7 orig}:}
We know that $\theta^{\star}_j$ satisfies (\ref{admiss strat int}) and (\ref{admiss strat term}), as
all strategies in $\tilde{\mathcal{A}}_j$ have these properties. 
Now $W_{j,\cdot} \left( \omega \right)$ is continuous on $\left[ 0 , T \right]$ 
for $P_j$-almost every $\omega \in \Omega_j$. When such an $\omega$ is chosen, 
(\ref{x theta star no omega}) becomes (\ref{x theta star}). 
The latter is a first order linear ODE with continuous coefficients, so 
$\theta^{\star}_{j, \cdot } \left( \omega \right)$ is continuous on $\left[ 0 , T \right)$
(e.g., by Chapter 1.2 of \cite{wolf+ode}). 

Since our terminal inventory constraint is deterministic, we observe that 
\begin{equation*}
\displaystyle\lim_{t \uparrow T} \theta^{\star}_{j,t} \left( \omega \right)
\end{equation*}
exists and is finite from (28) and (29) in the proof of Theorem 3.2 in \cite{bank+son+voss}, 
as well as (\ref{beta hat simplif}) in Step \ref{step 5 orig}. 
In particular, we can view the paths of $\theta^{\star}_{j}$ on $\left[ 0 , T \right]$ as 
$P_j$-a.s. continuous.\footnote{Alternatively, we could give an argument using singular point theory
as in Section \ref{true price dyn section}.} 
We conclude by noting that $\theta_j^{\star}$ is also $\mathcal{F}_{j,t}^{unf}$-adapted
by (28) and (29) in the proof of Theorem 3.2 in \cite{bank+son+voss}, 
(\ref{cond beta raw filt form}), 
and (\ref{beta hat simplif}) in Step \ref{step 5 orig}.

\qed

\section{Section \ref{true price dyn section} Proofs}\label{true price dyn section PRF APP}

We frequently reference various easy properties of the functions in Definition \ref{Phi j defn}. 
We collect these below for convenience. We will leave the proof to the reader.

\begin{lem}\label{Phi j props lem}

Fix $j \in \left\{ 1, \dots, K \right\}$. We have the following:

\begin{enumerate}[label=\roman*)]

\item \label{a easy props phi} 
$\Phi_{j}$ is a strictly decreasing nonnegative function on $\left[ 0 , T \right]$ with $\Phi_{j} \left( T \right) = 0$.  

\item  \label{b easy props phi}
The entries of $A$ are analytic on $\left[ 0 , T \right]$ and $A \left( T \right) = I_K$.

\item \label{c easy props phi}
If $\det A$ has a root on $\left[ 0, T \right]$, we can find the smallest one which we denote by $t_e$. 
In this case, $t_e < T$ and the zero of $\det A$ at $t_e$ is of finite multiplicity. 

\item \label{d easy props phi}
The entries of $B$ are analytic on $\left[ 0 , T \right)$ but 
\begin{equation*}
\displaystyle\lim_{t \uparrow T} B_{jj} \left( t \right) = -\infty. 
\end{equation*}

\item \label{e easy props phi} 
$C \left( \cdot, \omega \right)$'s entries are continuous on $\left[ 0, T \right]$.

\end{enumerate}

\end{lem}

\subsection{Proof of Lemma \ref{act 1st order ODE lem}}\label{act 1st order ODE lem PRF SS}

Let $j \in \left\{ 1, \dots, K \right\}$. 
At each time $t$, 
Agent $j$ observes the correct value of $S_t^{exc} \left( \omega \right)$,
interprets this value as the realized value of $S_{j,\theta_{j}^{\star},t}^{exc} \left( \omega \right)$, and computes $S_{j,t}^{unf} \left( \omega \right) $.\footnote{By abuse of notation, we evaluate $S_{j,\theta_{j}^{\star},t}^{exc}$ and $S_{j,t}^{unf}$ are evaluated at $\omega$; however, Agent $j$ would evaluate these quantities at some $\omega_j \in \Omega_j$. We adopt similar conventions in the sequel without further comment.} 
By (\ref{exc price model j}), it follows that
\begin{align}\label{exc price model unaff j ident}
S_t^{exc} \left(\omega \right) &= S_{j,\theta_{j}^{\star},t}^{exc} \left(\omega \right) \notag \\
&= S_{j,t}^{unf} \left(\omega \right) + \eta_{j,per} \left(  X_{j,t}^{\theta_{j}^{\star}} \left(\omega \right)  - x_{j}   \right)
+ \frac{1}{2} \eta_{j,tem}  \theta_{j,t}^{\star} \left(\omega \right)  . 
\end{align}
After substituting (\ref{rt true dynamics}) into (\ref{exc price model unaff j ident}), we have
\begin{align}\label{unf price chg actual}
&S_{j,t}^{unf} \left(\omega \right)  - S_{j,0} \notag \\
&\quad =  \left( S_{0} - S_{j,0} \right)+  \tilde{\beta}  t 
 + \displaystyle\sum_{  \substack{i \leq K \\ i \not = j  }} \tilde{\eta}_{i ,per}\left( X^{\theta_{i}^\star}_{i,t} \left(\omega \right) -  x_{i} \right) 
 + \displaystyle\sum_{ i > K } \tilde{\eta}_{i ,per}\left( X^{\theta_{i}^\star}_{i,t}  -  x_{i} \right) 
 \notag \\
 &\qquad
  + \frac{1}{2} \displaystyle\sum_{   \substack{i \leq K \\ i \not = j  }  }  \tilde{\eta}_{i,tem} \theta_{i,t}^{\star} \left(\omega \right) 
    + \frac{1}{2} \displaystyle\sum_{  i > K  }  \tilde{\eta}_{i,tem} \theta_{i,t}^{\star} 
  \notag \\
&\qquad + \left( \tilde{\eta}_{j,per}  - \eta_{j,per} \right)  \left(  X_{j,t}^{\theta_{j}^{\star}} \left(\omega \right) - x_{j}   \right) 
 + \frac{1}{2} \left( \tilde{\eta}_{j,tem} - \eta_{j,tem}\right)  \theta_{j,t}^{\star} \left(\omega \right) + \tilde{W}_t  \left(\omega \right)  . 
\end{align}

The quantity on the LHS of (\ref{unf price chg actual}) plays a role in determining Agent $j$'s strategy (see Lemma \ref{1 pl soln acc to play model}).  
Substituting (\ref{unf price chg actual}) into (\ref{x theta star}) and applying the half-angle formula for $\tanh \left(  \cdot \right)$, 
we get
\begin{align*}
&A_{jj}\left( t \right) \theta_{j,t}^{\star} \left(\omega \right) 
- \displaystyle\sum_{\substack{ i \leq K \\ i \not = j}} A_{ji}\left( t \right) \theta_{i,t}^{\star}\left(\omega \right) \notag \\
&\quad = B_{jj}\left( t \right)  X^{\theta_j^\star}_{j,t} \left(\omega \right)  
+ \displaystyle\sum_{\substack{ i \leq K \\ i \not = j}} B_{ji}\left( t \right)  X^{\theta_{i}^\star}_{i,t}  \left(\omega \right) + C_{j} \left( t , \omega \right) .
\end{align*}
It follows that the uncertain agents' strategies are characterized by the ODE system
\begin{align}\label{DG syst lin ODE 1}
A \left(t \right)  \theta_{t}^{u,\star} \left( \omega \right) &= 
B \left( t \right) X^{u,\theta^\star}_{t} \left(\omega \right)  + C \left( t, \omega \right) \notag \\
X^{u,\theta^\star}_{0} \left(\omega \right) &= x^u. 
\end{align}
 Corollary \ref{cert agents opt strats}, Lemma \ref{Phi j props lem} and a standard existence and uniqueness theorem 
 (see Sections 1.1 and 3.1 of \cite{barbu}) finish the argument. 
\qed

\subsection{Proof of Lemma \ref{ODE HOM sing pt mult lem}}\label{ODE HOM sing pt mult lem PRF SS}

As $t \uparrow t_e$, 
\begin{align*}
\begin{array}{c}
 \left\{ A \left(t \right)  \dot{X}_{t}^{u} \left( \omega \right) = 
B \left( t \right) X_{t}^{u}\left(\omega \right) \right\} \\
\iff \\
\left\{ 
\left[ \det  A \left(t \right) \right] \dot{X}_{t}^{u} \left( \omega \right) = 
\left[ \text{adj}  A \left(t \right) \right]
B \left( t \right) X_{t}^{u}\left(\omega \right) \right\} . 
\end{array}
\end{align*}
Here, adj denotes the usual adjugate operator.

We can find a non-negative integer $m$ such that the multiplicity of the zero of $\det A$ at $t_e$ is 
$\left( m + 1 \right)$ by Lemma \ref{Phi j props lem}. 
Hence, there is a unique non-vanishing analytic function $f$ such that 
\begin{equation}\label{f defn FIN}
\det A \left( t \right) = \left( t - t_e \right)^{m+1} f \left( t \right)
\end{equation}
on a small neighborhood of $t_e$. 
Note that $f$ is non-vanishing, as the zeroes of $\det A$ are isolated and $\det A \left( T \right) = 1$ (see Lemma \ref{Phi j props lem}). 
We then define the analytic (see Lemma \ref{Phi j props lem}) map $D$ by
\begin{equation}\label{D defn FIN}
D \left( t \right) \triangleq 
\left[ \text{adj}  A \left(t \right) \right] B \left( t \right) / f \left( t \right) 
\end{equation}
and arrive at (\ref{DG syst lin ODE HOM REG}).

Since $\det A \left( \cdot \right)$ has a root at $t_e$, 
the rank of $A \left( t_e \right)$ is no more than $K - 1$. 
We conclude by observing that adj $A \left( t_e \right)$ has rank 1 when $A \left( t_e \right)$ 
has rank $K - 1$; otherwise, adj $A \left( t_e \right)$ must be the zero matrix.
The comments about the rank of $D \left( t_e \right)$ immediately follow.
\qed

\subsection{Proof of Proposition \ref{opt trad speeds uncer}}\label{opt trad speeds uncer PRF SS}

$D \left( t_e \right) \not = 0$ since $\lambda \not = 0$. 
Then (\ref{DG syst lin ODE HOM}) has a singular point of the first kind at $t_e$ (see our discussion above). 
$\lambda \not \in \mathbb{Z}$ by hypothesis, so Theorem 6.5 of \cite{codd+carl+ode}
implies that a fundamental solution of (\ref{DG syst lin ODE HOM}) 
on $\left[ t_e - \rho , t_e \right)$ for some $\rho > 0$ is given by
\begin{align}\label{fund soln homog R nonzer FIN1}
P \left( t \right) \left| t - t_e \right|^{D \left( t_e \right)} . 
\end{align}
In (\ref{fund soln homog R nonzer FIN1}), 
$P\left( \cdot \right)$ is an analytic $M_K \left( \mathbb{R} \right)$-valued function with $P \left( t_e \right) = I_K$. 
Moreover, $P \left( t \right)$ is invertible for all $t \in \left[ t_e - \rho , t_e \right)$ and\footnote{Any fundamental solution of (\ref{DG syst lin ODE HOM REG})
is invertible everywhere, as are matrix exponentials.}
\begin{equation}\label{fund soln inv FIN1}
 \left( P \left( t \right) \left| t - t_e \right|^R \right)^{-1}
 = 
 \left| t - t_e \right|^{- R}   \left[ P \left( t \right) \right]^{-1} .
\end{equation}

The solution of (\ref{DG syst lin ODE}) satisfies
\begin{align*}
\left(t - t_e \right) \theta_{t}^{u,\star} \left( \omega \right)  &= 
D \left( t \right) X^{u,\theta^\star}_{t} \left(\omega \right) 
+  
\displaystyle\frac{  \text{adj} \left[ A \left( s \right) \right] C \left( s , \omega \right)}{ f \left( s \right)}
 .
\end{align*}
near $t_e$ (argue as in Lemma \ref{ODE HOM sing pt mult lem}). 
Since 
\begin{equation*}
P \left( t \right) \left| t - t_e \right|^{D \left( t_e \right)}  \rho^{-D \left( t_e \right)}  \left[ P \left( t_e - \rho \right)  \right]^{-1}
\end{equation*}
is also a fundamental solution of (\ref{DG syst lin ODE HOM REG}) on $\left[ t_e - \rho , t_e \right)$\footnote{See Theorem 2.5 of 
Coddington \& Carlson (\cite{codd+carl+ode}).}
and equals $I_K$ at $t_e - \rho$, 
we can apply variation of parameters\footnote{See Theorem 2.8 of Coddington \& Carlson (\cite{codd+carl+ode}).}
to obtain
\begin{align}\label{orig var param NO eign}
 &X^{u,\theta^\star}_{t} \left( \omega \right) \notag \\
 &\, = P \left( t \right) \left| t - t_e \right|^{D \left( t_e \right)} \left[ \rho^{-D \left( t_e \right)}  \left[ P \left( t_e - \rho \right)  \right]^{-1}\right]
 \cdot \Bigg[ X^{\theta^\star}_{t_e - \rho} \left( \omega \right) \notag \\
 &\quad 
+
 \displaystyle\int_{t_e - \rho }^t 
 \left( P \left( t_e - \rho \right)  \rho^{D \left( t_e \right)}
\left| s - t_e \right|^{-D \left( t_e \right) } \left[ P \left( s \right) \right]^{-1}
 \right) 
 \left( \displaystyle\frac{  \text{adj} \left[ A \left( s \right) \right] C \left( s , \omega \right)}{ \left( s - t_e \right) f \left( s \right)} \right)
 ds 
\Bigg] . 
\end{align}

We can find an eigenbasis $\left\{v_1, \dots, v_{K} \right\}$ for $D \left( t_e \right)$ such that $v_K$ corresponds to $\lambda$. 
We then define 
the continuous 
real-valued functions $\left\{ F_1 \left( \cdot , \omega \right), \dots , F_K \left( \cdot , \omega \right) \right\}$
on $\left[ t_e - \rho , t_e \right]$
and the constants 
$\left\{ y_1 \left( \omega \right) , \dots , y_K \left( \omega \right) \right\}$
as certain eigenbasis coordinates:
\begin{align}\label{y F defn eign}
\displaystyle\sum_{j = 1}^{K}  F_j \left( s, \omega \right) v_{j} &\triangleq  \displaystyle\frac{  \left[ P \left( s \right) \right]^{-1} \text{adj} \left[ A \left( s \right) \right] C \left( s , \omega \right)}{ f \left( s \right)} \notag \\
\displaystyle\sum_{j = 1}^{K}  y_j \left( \omega \right) v_{j} &\triangleq \rho^{- D \left( t_e \right)}  \left[ P \left( t_e - \rho \right)  \right]^{-1} X^{\theta^\star}_{t_e - \rho} \left( \omega \right) . 
\end{align}
Taken with (\ref{orig var param NO eign}), these definitions immediately give (\ref{gen limit diag R main blow up KZ})  
after recalling that for any matrix $Q \in M_K \left( \mathbb{R} \right)$ with 
eigenvalue $\gamma$ and corresponding eigenvector $v$, 
we have
\begin{align*}
\left| t - t_e \right|^Q v = \left| t - t_e \right|^{\gamma} v.
\end{align*}
\qed

\subsection{Proof of Proposition \ref{det A no root no blow up}}\label{det A no root no blow up PRF SS}

We know that $S^{exc} \left( \omega \right)$, the $X^{\theta_{j}^\star}_{j} \left( \omega \right)$'s and 
the $\theta^{\star}_{j} \left( \omega \right)$'s are all uniquely defined and continuous on $\left[ 0, T \right)$ 
(see Lemma \ref{act 1st order ODE lem}). 
Corollary \ref{cert agents opt strats} implies that 
$X^{\theta_j^\star}_{j} \left( \omega \right)$ and 
$\theta^{\star}_{j} \left( \omega \right)$
are continuous at $T$ for $j > K$ (the certain agents). It also gives us
\begin{equation*}
\displaystyle\lim_{t \uparrow T} X^{\theta_j^\star}_{j,t} \left( \omega \right) = 0
\end{equation*}
for $j > K$. 
It remains to show that 
\begin{equation}\label{suff to show no T blow up}
\displaystyle\lim_{t \uparrow T} X^{u,\theta^\star}_{t} \left( \omega \right) = 0 
\quad \text{and} \quad
\displaystyle\lim_{t \uparrow T} \theta^{u,\star}_{t} \left( \omega \right) \in \mathbb{R}^K. 
\end{equation}
As discussed above, one difficulty is that the diagonal entries of $B$ in (\ref{DG syst lin ODE})
explode at $T$ (see Lemma \ref{Phi j props lem}); however, the approach for resolving this issue 
is similar to that used to analyze solution behavior near $t_e$.

First, we show that (\ref{DG syst lin ODE HOM}) (after replacing $t_e$ with $T$)
has a singular point of the first kind at $T$. 
Now $ \sinh \left( \tau_{j} \left( \cdot \right) \right)$ has a zero of multiplicity 1 at $T$
since 
\begin{equation*}
\displaystyle\frac{d \sinh \left( \tau_{j} \left( t \right) \right) }{dt}    \Bigg|_{t = T}
= - \displaystyle\sqrt{\displaystyle\frac{ \kappa_{j}}{ \eta_{j,tem}}}   
\cosh \left( \tau_{j} \left( t \right) \right)  \Bigg|_{t = T} = - \displaystyle\sqrt{\displaystyle\frac{ \kappa_{j}}{ \eta_{j,tem}}}. 
\end{equation*}
Hence, there is a unique non-vanishing analytic function $g_j$ such that 
\begin{equation}\label{gj props FIN}
\sinh \left( \tau_{j} \left( t \right) \right) = \left(t - T \right) g_j \left( t \right)
\quad \text{and} \quad
g_j \left( T \right) = - \displaystyle\sqrt{\displaystyle\frac{ \kappa_{j}}{ \eta_{j,tem}}}
\end{equation}
on a small neighborhood of $T$. 
Near $T$, it follows that the entries of $\left( t - T \right) B \left( t  \right)$ are given by
\begin{align}\label{B with gj eqn}
\left( t - T \right) B_{ik} \left( t  \right) =
\left\{
\begin{array}{ll}
\left( t - T \right) \left( \tilde{\eta}_{i,per}  - \eta_{i,per} \right)  \Phi_i \left( t \right)  & \,\\
\qquad - \displaystyle\sqrt{\displaystyle\frac{ \kappa_{i}}{ \eta_{i,tem}}}  
\left(
\displaystyle\frac{ \cosh \left( \tau_{i} \left( t \right) \right)}{ g_i \left( t \right) }
 \right)
 & \quad \text{if } i = k \\
 &\\
\left( t - T \right) \tilde{\eta}_{k,per}  \Phi_i \left( t \right)   
  & \quad \text{if } i \not =  k \\
\end{array}
\right. 
\end{align}
(see Definition \ref{Phi j defn}). 
On this region, the solution of (\ref{DG syst lin ODE HOM})
satisfies 
\begin{align}\label{DG syst lin ODE HOM T}
\left(t - T \right) \dot{X}_{t}^{u} \left( \omega \right) &= 
A^{-1} \left( t \right) \left( t - T \right) B \left( t \right) 
 X_{t}^{u}\left(\omega \right)   .
\end{align}
By (\ref{B with gj eqn}) and Lemma \ref{Phi j props lem}, 
(\ref{DG syst lin ODE HOM T}) has a singular point of the first kind at $T$.

Second, we find a fundamental solution of (\ref{DG syst lin ODE HOM T}) near $T$. 
We know that
\begin{equation*}
A^{-1} \left( T \right) = \left( t - T \right) B \left( t  \right)  \Bigg|_{t = T} = I_{K}
\end{equation*}
by (\ref{gj props FIN}), (\ref{B with gj eqn}), and Lemma \ref{Phi j props lem}. 
Theorem 6.5 of \cite{codd+carl+ode}
implies that a fundamental solution of (\ref{DG syst lin ODE HOM T}) 
on $\left[ T - \delta , T \right)$ for some $\delta > 0$ is given by
\begin{align}\label{fund soln homog R nonzer FIN1 Bprf}
Q \left( t \right) \left| t - T \right|^{I_K}  = Q \left( t \right) \left| t - T \right|. 
\end{align}
In (\ref{fund soln homog R nonzer FIN1 Bprf}), 
$Q$ is an analytic $M_K \left( \mathbb{R} \right)$-valued function with $Q \left( T \right) = I_K$. 
Also, $Q \left( t \right)$ is invertible for all $t \in \left[ T  - \delta , T \right)$.\footnote{Any fundamental solution of (\ref{DG syst lin ODE HOM T})
is invertible everywhere.}

Finally, we use our fundamental solution to solve (\ref{DG syst lin ODE}) and conclude the proof. 
Notice that $\tanh \left( \tau_j \left( \cdot \right) \right)$ also has a zero
of multiplicity 1 at $T$ since
\begin{equation*}
\displaystyle\frac{d \tanh \left( \tau_{j} \left( t \right) \right) }{dt}    \Bigg|_{t = T}
=  - \displaystyle\frac{1}{2} \displaystyle\sqrt{\displaystyle\frac{ \kappa_{j}}{ \eta_{j,tem}}}   
\text{sech}^2 \left( \tau_{j} \left( t \right) /2 \right)
 \Bigg|_{t = T} 
=  - \displaystyle\frac{1}{2} \displaystyle\sqrt{\displaystyle\frac{ \kappa_{j}}{ \eta_{j,tem}}}   . 
\end{equation*}
There is a unique non-vanishing analytic function $h_j$ such that 
\begin{equation}\label{hj props FIN}
\tanh \left( \tau_{j} \left( t \right) / 2\right) = \left(t - T \right) h_j \left( t \right) 
\end{equation}
on a neighborhood of $T$. 
In particular, the entries of $C \left(t , \omega \right) / \left( t - T \right)$ near $T$ are given by
\begin{align}\label{C divd by t - T  eqn}
\displaystyle\frac{C_{i} \left( t , \omega \right)}{ \left( t - T \right) } &=
\left(  \displaystyle\frac{  h_i \left( t \right) \nu_{i}^2    }{ \sqrt{ \eta_{i,tem}  \kappa_{i}}   \left( 1 + \nu_{i}^2 t \right)} \right)
 \left[ \displaystyle\frac{ \mu_{i}}{\nu_{i}^2} +  \left( S_{0} - S_{i,0} \right)+  \tilde{\beta}  t 
 - \displaystyle\sum_{ \substack{ k \leq K \\ k \not = i }  } \tilde{\eta}_{k,per}   x_{k}   \right. \notag \\
 &\qquad \qquad \qquad 
 - x_{i} \left( \tilde{\eta}_{i,per}  - \eta_{i,per} \right) 
  + \displaystyle\sum_{ k > K } \tilde{\eta}_{k ,per}\left( X^{\theta_{k}^\star}_{k,t}  -  x_{k} \right) \notag \\
   &\qquad \qquad \qquad \left.
      + \frac{1}{2} \displaystyle\sum_{  k > K  }  \tilde{\eta}_{k,tem} \theta_{k,t}^{\star}  + \tilde{W}_t \left( \omega \right) 
     \right] . 
\end{align}
Since 
\begin{equation*}
Q \left( t \right)  \left| t - T \right| \delta^{-1} Q^{-1} \left( T - \delta \right) 
\end{equation*}
is also a fundamental solution of (\ref{DG syst lin ODE HOM REG}) on $\left[ T - \delta  , T \right)$\footnote{See Theorem 2.5 of 
Coddington \& Carlson (\cite{codd+carl+ode}).}
and equals $I_K$ at $T- \delta$, 
we can apply variation of parameters\footnote{See Theorem 2.8 of Coddington \& Carlson (\cite{codd+carl+ode}).}
to obtain
\begin{align}\label{orig var param NO eign T}
 &X^{u,\theta^\star}_{t} \left( \omega \right) \notag \\
 &\quad = Q \left( t \right) \left| t - T \right| \delta^{-1}  Q^{-1} \left( T - \delta \right) 
 \cdot \Bigg[ X^{\theta^\star}_{T - \delta} \left( \omega \right) \notag \\
 &\qquad 
+
 \displaystyle\int_{T - \delta }^t 
 \left( Q \left( T - \delta \right)  \delta
\left| s - T \right|^{-1 } Q^{-1} \left( s \right)
 \right) 
    A^{-1} \left( s \right) C \left( s , \omega \right)  
 ds 
\Bigg] . 
\end{align}
By (\ref{C divd by t - T  eqn}), (\ref{orig var param NO eign T}), and Corollary \ref{cert agents opt strats}, we get (\ref{suff to show no T blow up}). 
\qed

\section{Section \ref{exmp sect} Proofs}\label{exmp sect PRF APP}

\subsection{Proof of Lemma \ref{te study lem SS FIN}}\label{te study lem SS FIN PRF SS}

First observe that  \eqref{semi-symmetric gull te in 0 T FIN REWRITE} is equivalent to
\begin{equation}\label{semi-symmetric gull te in 0 T FIN}
\left( K \tilde{\eta}_{tem} - \eta_{tem} \right) \Phi \left( 0 \right) > 2,
\end{equation}
by Definitions \ref{tau j notn} and \ref{Phi j defn}. By Definitions \ref{Phi j defn} and \ref{semi symm defn}, 
we see that $A$ is now given by 
\begin{align}\label{A calc SS FIN}
A_{ik} \left( t  \right) &\triangleq  
\left\{
\begin{array}{cc}
 1 -   \displaystyle\frac{1}{2} \left( \tilde{\eta}_{tem} - \eta_{tem}\right) \Phi \left( t \right)  
 & \quad \text{if } i = k \\
 -  \displaystyle\frac{1}{2}   \tilde{\eta}_{tem} \Phi \left( t \right) 
  & \quad \text{if } i \not =  k \\
\end{array}
\right. .
\end{align}
A short calculation shows that 
\begin{align}\label{det A semi symm tild FIN}
\det  A \left(  t \right) &=
 \left[ 1 +  \displaystyle\frac{1}{2}  \eta_{tem} \Phi \left( t \right)  \right]^{K-1}
\left[  1 - \displaystyle\frac{1}{2} \left(   K  \tilde{\eta}_{tem} - \eta_{tem} \right) \Phi \left( t \right)    \right] .
\end{align}
The first term in (\ref{det A semi symm tild FIN}) is always at least 1.
The second term is non-zero at 0 but does have a root on $\left( 0, T \right]$ if and only if 
(\ref{semi-symmetric gull te in 0 T FIN}) holds.\footnote{In fact, $t_e$ is the unique root of $\det A$ in this case.}
Both of these observations come from Lemma \ref{Phi j props lem}.

Now, (\ref{semi-symmetric gull te in 0 T FIN}) 
implies that $K \tilde{\eta}_{tem} > \eta_{tem}$. 
Since $t_e$ is a zero of $\det A$, we must have that 
\begin{equation}\label{key eqn Phi te FIN}
 1 - \displaystyle\frac{1}{2} \left(   K  \tilde{\eta}_{tem} - \eta_{tem} \right) \Phi \left( t_e \right)   = 0. 
\end{equation}
Hence, by Lemma \ref{Phi j props lem}, 
\begin{align}\label{det deriv SS FIN}
\displaystyle\frac{d \left[ \det A \left( t \right) \right] }{dt} \Bigg|_{t = t_e} &= 
- \displaystyle\frac{1}{2} \left(   K  \tilde{\eta}_{tem} - \eta_{tem} \right)  \left[ 1 +  \displaystyle\frac{1}{2}  \eta_{tem} \Phi \left( t \right)  \right]^{K-1} \dot{\Phi} \left( t \right)  \Bigg|_{t = t_e} > 0. 
\end{align}
\qed

\subsection{Proof of Lemma \ref{EV comps SS FIN}}\label{EV comps SS FIN PRF SS}

By (\ref{f defn FIN}), (\ref{det deriv SS FIN}), and Lemma \ref{te study lem SS FIN}, 
\begin{align}\label{f te SS FIN}
f \left( t_e \right) &= \displaystyle\frac{d \left[ \det A \left( t \right) \right] }{dt} \Bigg|_{t = t_e} \notag \\
&= - \displaystyle\frac{1}{2} \left(   K  \tilde{\eta}_{tem} - \eta_{tem} \right)  \left[ 1 +  \displaystyle\frac{1}{2}  \eta_{tem} \Phi \left( t_e \right)  \right]^{K-1} \dot{\Phi} \left( t_e \right) . 
\end{align}
A short calculation shows that adj $A\left( t \right)$ is given by
\begin{align}\label{adj A SS FIN}
&\left[ \text{adj} A \left( t \right) \right]_{ik} \notag \\
\quad &=
\left( 1 +  \displaystyle\frac{1}{2}  \eta_{tem} \Phi \left( t \right)   \right)^{K-2}
\left\{
\begin{array}{ll}
 1 -   \displaystyle\frac{1}{2} \left[ \left( K - 1 \right) \tilde{\eta}_{tem} - \eta_{tem}\right] \Phi \left( t \right)  
 & \quad \text{if } i = k \\
  \displaystyle\frac{1}{2}   \tilde{\eta}_{tem} \Phi \left( t \right) 
  & \quad \text{if } i \not =  k \\
\end{array}
\right. . 
\end{align}
It follows that 
\begin{align}\label{adj A B form SS FIN}
&\left[ \left( \text{adj} A \left( t \right) \right] B \left( t \right) \right]_{ik} \notag \\
&\quad = \tilde{\eta}_{per}  \Phi \left( t \right) \left( 1 +  \displaystyle\frac{1}{2}  \eta_{tem} \Phi \left( t \right)   \right)^{K-1}  \notag \\
&\qquad   + \left( 1 +  \displaystyle\frac{1}{2}  \eta_{tem} \Phi \left( t \right)   \right)^{K-2} 
\left( \eta_{per}  \Phi \left( t \right)  + \displaystyle\sqrt{\displaystyle\frac{ \kappa }{ \eta_{tem}}}   \coth \left( \tau \left( t \right) \right)  \right) \notag \\
&\qquad \qquad \cdot 
\left\{
\begin{array}{ll}
 \displaystyle\frac{1}{2} \left[ \left( K - 1 \right) \tilde{\eta}_{tem} - \eta_{tem}\right] \Phi \left( t \right)  - 1
 & \quad \text{if } i = k \\
 &\\
 - \displaystyle\frac{1}{2}   \tilde{\eta}_{tem} \Phi \left( t \right) 
  & \quad \text{if } i \not =  k \\
\end{array}
\right. . 
\end{align}
One can then check that the only potentially non-zero eigenvalue of 
\begin{align*}
D \left(t_e \right) =  \displaystyle\frac{ \left[ \text{adj} A \left( t_e \right) \right] B \left( t_e \right) }{  f \left( t_e \right) }
\end{align*}
is given by 
\begin{align}\label{lambda SS nearly FIN}
\lambda &= - \displaystyle\frac{ 
2 \left[
\left( K \tilde{\eta}_{per} - \eta_{per}\right)  \Phi \left( t_e \right)  
-  \displaystyle\sqrt{\displaystyle\frac{ \kappa }{ \eta_{tem}}}   \coth \left( \tau \left( t_e \right) \right)  
\right] 
}
{
 \left(   K  \tilde{\eta}_{tem} - \eta_{tem} \right)   \dot{\Phi} \left( t_e \right)
} 
\end{align}
 with corresponding eigenvector $v_K$ as above. 
We get (\ref{lambda formula FIN}) from (\ref{lambda SS nearly FIN}) after applying (\ref{spec Phi te relat FIN}).

Recall that $\Phi \left( t_e \right) > 0$ and $\dot{\Phi} \left( t_e \right) < 0$ by Lemma \ref{Phi j props lem}. 
Since $t_e$, $\Phi$, and $\tau$ do not depend on $\tilde{\eta}_{per}$ or $\eta_{per}$, 
we can ensure that $\lambda \not \in \mathbb{Z}$ by perturbing the latter parameters. 
$D \left( t_e \right)$ is then diagonalizable as observed in Proposition \ref{opt trad speeds uncer}, 
and $v_1, \dots , v_{K-1}$ can be computed using (\ref{adj A B form SS FIN}). 
\qed

\subsection{Proof of Theorems \ref{SS main blow up lem FIN} and \ref{no Inv expl lem SS FIN}}\label{SS main blow up lem FIN PRF SS}

Since our uncertain agents are semi-symmetric,
\begin{align}\label{C SS FIN 1}
C_{i} \left( t , \omega \right) &=  \Phi \left( t \right)  \tilde{W}_t \left( \omega \right) \notag \\
&\qquad +  \Phi \left( t \right) \left[   \tilde{\beta}  t 
  + \displaystyle\sum_{ k > K } \tilde{\eta}_{k ,per}\left( X^{\theta_{k}^\star}_{k,t}  -  x_{k} \right) 
      + \frac{1}{2} \displaystyle\sum_{  k > K  }  \tilde{\eta}_{k,tem} \theta_{k,t}^{\star}  
     \right] \notag \\
&\qquad + \Phi \left( t \right) \left[ \displaystyle\frac{ \mu_{i}}{\nu^2} +  \left( S_{0} - S_{i,0} \right) 
 - \displaystyle\sum_{ \substack{ k \leq K \\ k \not = i }  } \tilde{\eta}_{per}   x_{k}  - x_{i} \left( \tilde{\eta}_{per}  - \eta_{per} \right)  \right] 
\end{align}
for $t \leq t_e$ by Definition \ref{Phi j defn}. 
For convenience, we introduce the following deterministic function\footnote{
The function $c$ is deterministic by Corollary \ref{cert agents opt strats}.}
$c$ and the constants $c_1, \dots, c_K$:
\begin{align}\label{aux fxns eigen expl SS lem eqn}
c \left( t \right) &\triangleq \left[   \tilde{\beta}  t 
  + \displaystyle\sum_{ k > K } \tilde{\eta}_{k ,per}\left( X^{\theta_{k}^\star}_{k,t}  -  x_{k} \right) 
      + \frac{1}{2} \displaystyle\sum_{  k > K  }  \tilde{\eta}_{k,tem} \theta_{k,t}^{\star}  
     \right] \notag \\
  \displaystyle\sum_{i = 1}^K   c_i v_i  &\triangleq 
  \left[ 
 \begin{array}{c}
 \displaystyle\frac{ \mu_{1}}{\nu^2} +  \left( S_{0} - S_{1,0} \right) 
 - \displaystyle\sum_{ \substack{ k \leq K \\ k \not = 1 }  } \tilde{\eta}_{per}   x_{k}  - x_{1} \left( \tilde{\eta}_{per}  - \eta_{per} \right) \\
 \vdots \\
 \displaystyle\frac{ \mu_{K}}{\nu^2} +  \left( S_{0} - S_{K,0} \right) 
 - \displaystyle\sum_{ \substack{ k \leq K \\ k \not = K }  } \tilde{\eta}_{per}   x_{k}  - x_{K} \left( \tilde{\eta}_{per}  - \eta_{per} \right) 
   \end{array} \right] .
\end{align}

Using (\ref{C SS FIN 1}), we get that 
\begin{align}\label{C SS FIN simplif}
C \left( t , \omega \right) &=  
\tilde{W}_t \left( \omega \right) \Phi \left( t \right) v_K 
+ c \left( t \right) \Phi \left( t \right) v_K 
+ \Phi \left( t \right)   \displaystyle\sum_{i = 1}^K   c_i v_i . 
\end{align}
By (\ref{adj A SS FIN}), $\left\{ v_1 , \dots, v_K \right\}$ is an eigenbasis for $\text{adj} \left[ A \left( t \right) \right]$.  
Moreover, 
\begin{equation}\label{evals adj A t eqn 1}
\left( 1 +  \displaystyle\frac{1}{2}  \eta_{tem} \Phi \left( t \right)   \right)^{K-2}
\left[  1 - \displaystyle\frac{1}{2} \left(   K  \tilde{\eta}_{tem} - \eta_{tem} \right) \Phi \left( t \right)    \right]
\end{equation}
is the eigenvalue corresponding to each of  $v_1 , \dots, v_{K-1}$, while 
\begin{equation}\label{evals adj A t eqn 2}
\left( 1 +  \displaystyle\frac{1}{2}  \eta_{tem} \Phi \left( t \right)   \right)^{K-1}
\end{equation}
corresponds to $v_K$.

By (\ref{y F defn eign}), it follows that 
\begin{align}\label{large decpmp for Fjs SS eqn}
&\displaystyle\sum_{j = 1}^{K}  F_j \left( t, \omega \right) v_{j} \notag \\
&\quad =  \displaystyle\frac{  \left[ P \left( t \right) \right]^{-1} \text{adj} \left[ A \left( t \right) \right] C \left( t , \omega \right)}{ f \left( t \right)  } \notag \\
&\quad  = 
\tilde{W}_t \left( \omega \right)  \left(
 \displaystyle\frac{ \Phi \left( t \right) \left( 1 +  \displaystyle\frac{1}{2}  \eta_{tem} \Phi \left( t \right)   \right)^{K-1} }
 {   f \left( t \right)  }
  \right)  \left[ P \left( t \right) \right]^{-1}  v_K \notag \\
 &\qquad   \quad +   
  \left(
 \displaystyle\frac{ \Phi \left( t \right) \left( 1 +  \displaystyle\frac{1}{2}  \eta_{tem} \Phi \left( t \right)   \right)^{K-1} \left( c \left( t \right) + c_K \right) }
 {   f \left( t \right)  }
  \right)  \left[ P \left( t \right) \right]^{-1}  v_K \notag \\
   &\qquad   \quad +   
  \left(
 \displaystyle\frac{\Phi \left( t \right) \left( 1 +  \displaystyle\frac{1}{2}  \eta_{tem} \Phi \left( t \right)   \right)^{K-2}
\left[  1 - \displaystyle\frac{1}{2} \left(   K  \tilde{\eta}_{tem} - \eta_{tem} \right) \Phi \left( t \right)    \right] }
 {   f \left( t \right)  }
  \right)  
\notag \\
 &\qquad   \qquad \qquad \cdot     \displaystyle\sum_{i = 1}^{K-1} c_i \left[ P \left( t \right) \right]^{-1}  v_i    .
\end{align}
It follows that we can find 
analytic deterministic functions $F_{j,1}$ and $F_{j,2}$ such that 
\begin{equation}\label{Fj1 Fj2 defn}
F_{j} \left( t , \omega \right) \triangleq \tilde{W}_t \left( \omega \right)    F_{j,1} \left( t \right) + F_{j,2} \left( t \right) 
\end{equation}
for each $j \in \left\{ 1, \dots, K \right\}$.\footnote{
Note that $c$ is continuously differentiable on $\left[ 0 , t_e \right]$ by Corollary \ref{cert agents opt strats}.} 
Since $P \left( t_e \right) = I_K$ (see Proposition \ref{opt trad speeds uncer}), (\ref{f te SS FIN}) and Remark \ref{te impl formula FIN}
further imply that 
\begin{equation}\label{other Fjs zero at t_e 100}
F_{j,1} \left( t_e  \right) = F_{j,2} \left( t_e  \right) =  \cdots = F_{K-1,1} \left( t_e \right) = F_{K-1,2} \left( t_e \right) = 0
\end{equation}
and 
\begin{align}\label{FK zero at t_e 100} 
F_{K,1} \left( t_e  \right) = 
  -  
 \displaystyle\frac{ \Phi^2 \left( t_e \right) }
 { \dot{\Phi} \left( t_e \right) } > 0
 \quad \text{and} \quad 
 F_{K,2} \left( t_e  \right) =  -  
 \displaystyle\frac{ \Phi^2 \left( t_e \right) }
 { \dot{\Phi} \left( t_e \right) } 
 \left(  c \left( t_e \right) + c_K  \right) .
\end{align}
While $F_{K,1} \left( t_e  \right) > 0$, determining the sign of $F_{K,2} \left( t_e  \right)$
 is difficult, in general, as it depends upon the sign of $c \left( t_e \right) + c_K$ 
(see (\ref{aux fxns eigen expl SS lem eqn})).

 We see from (\ref{Fj1 Fj2 defn}) and (\ref{other Fjs zero at t_e 100})
 that the expression
 \begin{equation}\label{deriv exp theta vj}
 \displaystyle\frac{ F_j \left( s, \omega \right)}
{ \left| s - t_e \right| } 
 \end{equation}
is bounded near $t_e$ for each $j < K$ and almost every $\omega \in \tilde{\Omega}$. 
In particular, the coordinates of both 
\begin{equation*}
\displaystyle\sum_{j = 1}^{K-1} 
\left( y_j \left( \omega \right) -   \displaystyle\int_{t_e - \rho }^t 
\displaystyle\frac{ F_j \left( s, \omega \right)}
{ \left| s - t_e \right| }\, ds 
\right) 
 P \left( t \right)  v_{j}
\end{equation*}
and its time derivative are bounded near $t_e$ for such $\omega$ as well.

Since $P \left( t_e \right) = I_K$, the $v_K$-coordinate of $P\left( t \right) v_K$
tends to 1 $t \uparrow t_e$. 
For $j < K$, the $v_j$-coordinate of $P\left( t \right) v_K$
tends to 0 as $t \uparrow t_e$. 
In each situation, we can also obtain Lipschitz bounds on the convergence. 
Due to (\ref{gen limit diag R main blow up KZ}) 
and (\ref{Fj1 Fj2 defn}), 
potential explosions in the coordinates of $X^{u,\theta^\star}_{t} \left( \omega \right)$ 
are characterized by 
\begin{equation}\label{Xu explosion last integral eqn SCALAR}
\displaystyle\lim_{t \uparrow t_e}  \left[ \left| t - t_e \right|^{ \lambda  }   \left( y_K \left( \omega \right)  -
 \displaystyle\int_{t_e - \rho }^t 
\displaystyle\frac{  \tilde{W}_s \left(\omega \right) F_{K,1} \left( s  \right) + F_{K,2} \left( s  \right)   }
{ \left| s - t_e \right|^{1+ \lambda} }\, ds 
\right) 
\right] . 
\end{equation}
Specifically, 
\begin{align}\label{Xu explosion last integral eqn SCALAR directs}
\left\{ 
\left| \text{(\ref{Xu explosion last integral eqn SCALAR})} \right| < +\infty
\right\}
\quad &\iff  \quad 
\left\{  \displaystyle\lim_{t \uparrow t_e}  X^{u,\theta^\star}_{t} \left( \omega \right) 
    \quad \text{exists in } \mathbb{R}^K 
  \right\} \notag \\
\left\{ 
\text{(\ref{Xu explosion last integral eqn SCALAR})}= +\infty
\right\}
\quad &\iff  \quad 
\left\{  \displaystyle\lim_{t \uparrow t_e} X^{u,\theta^\star}_{t} \left( \omega \right) 
    = \left[ + \infty, \dots, + \infty \right]^{\top} 
  \right\} \notag \\
 \left\{ 
\text{(\ref{Xu explosion last integral eqn SCALAR})}= -\infty
\right\}
\quad &\iff \quad 
\left\{  \displaystyle\lim_{t \uparrow t_e} X^{u,\theta^\star}_{t} \left( \omega \right) 
    = \left[ - \infty, \dots, - \infty \right]^{\top} 
  \right\} . 
\end{align}
To finish the proof, we separately consider the $\lambda < 0$ 
and $\lambda > 0$ cases.

\vspace{2mm}

\noindent {\bf $\lambda < 0$ Case.}

\vspace{2mm}

Assume that $\lambda < 0$. 
It follows that 
\begin{equation*}
\displaystyle\lim_{t \uparrow t_e} 
 \displaystyle\int_{ t_e - \rho   }^t
\displaystyle\frac{   \left| \tilde{W}_s \left(\omega \right) F_{K,1} \left( s  \right) \right| }
{ \left| s - t_e \right|^{1  +  \lambda } } 
\, ds  < \infty
\quad \text{and} \quad 
\displaystyle\lim_{t \uparrow t_e} 
 \displaystyle\int_{ t_e - \rho   }^t
\displaystyle\frac{   \left|  F_{K,2} \left( s  \right) \right| }
{ \left| s - t_e \right|^{1  +  \lambda } } 
\, ds  < \infty .
\end{equation*}
Clearly, 
\begin{equation*}
\displaystyle\lim_{t \uparrow t_e} 
\left| t - t_e \right|^{ \lambda  }  = + \infty , 
\end{equation*}
meaning that 
\begin{align}\label{neg Lamb SS FIN pos expl 1}
&\left\{
y_K \left( \omega \right)
 - \displaystyle\lim_{t \uparrow t_e} 
 \displaystyle\int_{ t_e - \rho   }^t 
\displaystyle\frac{  F_{K,2} \left( s \right)}
{ \left| s - t_e \right|^{1  +  \lambda } }\, ds 
>  \displaystyle\lim_{t \uparrow t_e} 
 \displaystyle\int_{ t_e - \rho   }^t 
\displaystyle\frac{  \tilde{W}_s \left( \omega \right) F_{K,1} \left( s  \right)}
{ \left| s - t_e \right|^{1  +  \lambda } }\, ds 
\right\} \\
&\quad \implies \quad
  \left\{  \displaystyle\lim_{t \uparrow t_e} X^{u, \theta^\star}_{t} \left( \omega \right)= \left[ + \infty, \dots, + \infty \right]^{\top}
  \right\} \notag
  \end{align}
  and 
 \begin{align}\label{neg Lamb SS FIN NEG expl 1}
&\left\{
y_K \left( \omega \right)
 - \displaystyle\lim_{t \uparrow t_e} 
 \displaystyle\int_{ t_e - \rho   }^t 
\displaystyle\frac{  F_{K,2} \left( s \right)}
{ \left| s - t_e \right|^{1  +  \lambda } }\, ds 
<  \displaystyle\lim_{t \uparrow t_e} 
 \displaystyle\int_{ t_e - \rho   }^t 
\displaystyle\frac{  \tilde{W}_s \left( \omega \right) F_{K,1} \left( s  \right)}
{ \left| s - t_e \right|^{1  +  \lambda } }\, ds 
\right\} \\
&\quad \implies \quad
  \left\{  \displaystyle\lim_{t \uparrow t_e} X^{u, \theta^\star}_{t} \left( \omega \right)= \left[ - \infty, \dots, - \infty \right]^{\top}
  \right\} \notag
  \end{align} 
Arguing as in our discussion of (\ref{deriv exp theta vj}), 
 we see that the hypotheses in 
(\ref{neg Lamb SS FIN pos expl 1})
and 
 (\ref{neg Lamb SS FIN NEG expl 1})
also imply that 
\begin{equation*}
\left\{  \displaystyle\lim_{t \uparrow t_e} \theta^{u,\star}_{t} \left( \omega \right) 
    = \left[ + \infty, \dots, + \infty \right]^{\top} 
  \right\} 
  \quad \text{and} \quad 
  \left\{  \displaystyle\lim_{t \uparrow t_e} \theta^{u,\star}_{t} \left( \omega \right) 
    = \left[ - \infty, \dots, - \infty \right]^{\top} 
  \right\} ,
\end{equation*}
  respectively.\footnote{
  In particular, the coordinates of $\theta^{u,\star}_{t} \left( \omega \right)$
  will asymptotically explode at the rate $\left| t - t_e \right|^{-\lambda - 1}$.}
 Conditional on $\tilde{\mathcal{F}}_{t_e - \rho}$, the RHS of the inequality in 
(\ref{neg Lamb SS FIN pos expl 1}) (and \ref{neg Lamb SS FIN NEG expl 1})
 is deterministic. 
 Since $F_{K,1} \left( t_e \right) > 0$
 (see (\ref{FK zero at t_e 100})), 
 we finish our proof of Theorem~\ref{SS main blow up lem FIN}.

\vspace{2mm}

\noindent {\bf $\lambda > 0$ Case.}

\vspace{2mm}

Assume that $\lambda > 0$. 
We can find a constant $R_0 \left( \omega \right)$ such that 
\begin{equation}\label{lam pos no expl prelim 1}
\left| y_K \left( \omega \right)  -
 \displaystyle\int_{t_e - \rho }^t 
\displaystyle\frac{  \tilde{W}_s \left(\omega \right) F_{K,1} \left( s  \right) + F_{K,2} \left( s  \right)   }
{ \left| s - t_e \right|^{1+ \lambda} }\, ds 
\right| \leq \displaystyle\frac{R_0 \left( \omega \right) }{ \left| t - t_e \right|^{\lambda}} . 
\end{equation}
Hence, (\ref{Xu explosion last integral eqn SCALAR}) is bounded as $t \uparrow t_e$ and 
\begin{equation*}
\displaystyle\lim_{t \uparrow t_e} X^{u,\theta^\star}_{t} \left( \omega \right) 
 \end{equation*} 
 exists in $\mathbb{R}^K$ by our previous comments.

By our discussion surrounding (\ref{deriv exp theta vj}), 
we see that explosions in the coordinates of $\theta^{u,\star}_{t} \left( \omega \right)$ 
are characterized by 
\begin{align}\label{lamb pos theta expl key term 1}
\displaystyle\lim_{t \uparrow t_e} 
&\left[
- \lambda \left| t - t_e \right|^{\lambda - 1} 
\left(
y_K \left( \omega \right)  -
 \displaystyle\int_{t_e - \rho }^t 
\displaystyle\frac{  \tilde{W}_s \left(\omega \right) F_{K,1} \left( s  \right) + F_{K,2} \left( s  \right)   }
{ \left| s - t_e \right|^{1+ \lambda} }\, ds 
\right)
\right. \notag \\
&\qquad \left.
- 
\left( 
\displaystyle\frac{  \tilde{W}_t \left(\omega \right) F_{K,1} \left( t \right) + F_{K,2} \left( t  \right)   }
{ \left| t - t_e \right| }
\right)
\right] . 
\end{align}
More precisely, 
\begin{align}\label{lamb pos theta expl key term 1 limits}
\left\{ 
\text{(\ref{lamb pos theta expl key term 1})}= +\infty
\right\}
\quad &\iff \quad 
\left\{  \displaystyle\lim_{t \uparrow t_e} \theta^{u,\star}_{t} \left( \omega \right) 
    = \left[ + \infty, \dots, + \infty \right]^{\top} 
  \right\} \notag \\
 \left\{ 
\text{(\ref{lamb pos theta expl key term 1})}= -\infty
\right\}
\quad &\iff \quad 
\left\{  \displaystyle\lim_{t \uparrow t_e} \theta^{u,\star}_{t} \left( \omega \right) 
    = \left[ - \infty, \dots, - \infty \right]^{\top} 
  \right\} . 
\end{align}

Suggestively, we first rewrite the expression in (\ref{lamb pos theta expl key term 1}) as
\begin{align}\label{term by term explo pos lamb}
&F_{K,2} \left( t  \right)  \left( 
\lambda \left| t - t_e \right|^{\lambda - 1} 
 \displaystyle\int_{t_e - \rho }^t \displaystyle\frac{  1 }
{ \left| s - t_e \right|^{1+ \lambda} }\, ds
- \displaystyle\frac{  1 }
{ \left| t - t_e \right| }
\right) \notag \\ 
&\quad + 
\lambda \left| t - t_e \right|^{\lambda - 1} 
 \displaystyle\int_{t_e - \rho }^t \displaystyle\frac{  F_{K,2} \left( s  \right)  - F_{K,2} \left( t  \right)  }
{ \left| s - t_e \right|^{1+ \lambda} }\, ds
 \notag \\ 
 &\quad - \lambda \left| t - t_e \right|^{\lambda - 1} y_K \left( \omega \right) \notag \\
 &\quad + \tilde{W}_t \left( \omega \right) F_{K,1} \left( t  \right)  \left( 
\lambda \left| t - t_e \right|^{\lambda - 1} 
 \displaystyle\int_{t_e - \rho }^t \displaystyle\frac{  1 }
{ \left| s - t_e \right|^{1+ \lambda} }\, ds
- \displaystyle\frac{  1 }
{ \left| t - t_e \right| }
\right) \notag \\ 
&\quad + 
\lambda \left| t - t_e \right|^{\lambda - 1} 
 \displaystyle\int_{t_e - \rho }^t \displaystyle\frac{ \tilde{W}_s \left( \omega \right)  \left[  F_{K,1} \left( s  \right)  - F_{K,1} \left( t  \right)  \right] }
{ \left| s - t_e \right|^{1+ \lambda} }\, ds
 \notag \\ 
  &\quad + 
 \lambda \left| t - t_e \right|^{\lambda - 1}  F_{K,1} \left( t  \right)
 \displaystyle\int_{t_e - \rho }^t \displaystyle\frac{ \tilde{W}_s \left( \omega \right)  - \tilde{W}_t \left( \omega \right) }
{ \left| s - t_e \right|^{1+ \lambda} }\, ds 
\end{align}
Let $R_1$ and $R_2$ be the deterministic Lipschitz coefficients for $F_{K,1}$ and $F_{K,2}$. 
The first two lines of (\ref{term by term explo pos lamb}) are deterministic, 
and we can obtain the following bounds:
\begin{equation*}
\begin{split}
& \left| F_{K,2} \left( t  \right)  \left(
\lambda \left| t - t_e \right|^{\lambda - 1} 
 \displaystyle\int_{t_e - \rho }^t \displaystyle\frac{  1 }
{ \left| s - t_e \right|^{1+ \lambda} }\, ds
- \displaystyle\frac{  1 }
{ \left| t - t_e \right| }
\right)
\right| \notag \\
&\qquad 
\leq 
\displaystyle\frac{ \left| F_{K,2} \left( t \right) \right| \left| t- t_e \right|^{\lambda - 1} }{\rho^{\lambda}} 
\notag \\ 
&\left| \lambda \left| t - t_e \right|^{\lambda - 1} 
 \displaystyle\int_{t_e - \rho }^t \displaystyle\frac{  F_{K,2} \left( s  \right)  - F_{K,2} \left( t  \right)  }
{ \left| s - t_e \right|^{1+ \lambda} }\, ds \right| \notag \\
&\qquad 
\leq 
\left( \displaystyle\frac{ \lambda R_2 }{1 - \lambda} \right) 
\left( \rho^{1 - \lambda} \left| t- t_e \right|^{\lambda - 1} - 1 \right) 
\end{split}
\end{equation*}
In (\ref{term by term explo pos lamb}), the third line is 
deterministic conditional on $\tilde{\mathcal{F}}_{t_e - \rho}$. 
Lines 4 - 6 of (\ref{term by term explo pos lamb})
are stochastic conditional on $\tilde{\mathcal{F}}_{t_e - \rho}$. 
Letting $R_3 \left( \omega \right)$ be the maximum of 
$\left| \tilde{W}_t \left( \omega \right) \right|$ on $\left[ t_e - \rho, t_e \right]$, 
we notice that
\begin{equation*}
\begin{split}
& \left| \tilde{W}_t \left( \omega \right) F_{K,1} \left( t  \right)  \left( 
\lambda \left| t - t_e \right|^{\lambda - 1} 
 \displaystyle\int_{t_e - \rho }^t \displaystyle\frac{  1 }
{ \left| s - t_e \right|^{1+ \lambda} }\, ds
- \displaystyle\frac{  1 }
{ \left| t - t_e \right| }
\right) 
\right| \notag \\
&\qquad 
\leq 
\displaystyle\frac{ F_{K,1} \left( t \right) \left| \tilde{W}_t \left( \omega \right) \right| \left| t- t_e \right|^{\lambda - 1} }{\rho^{\lambda}} 
\notag \\ 
&\left|
\lambda \left| t - t_e \right|^{\lambda - 1} 
 \displaystyle\int_{t_e - \rho }^t \displaystyle\frac{ \tilde{W}_s \left( \omega \right)  \left[  F_{K,1} \left( s  \right)  - F_{K,1} \left( t  \right)  \right] }
{ \left| s - t_e \right|^{1+ \lambda} }\, ds
\right| \notag \\
&\qquad 
\leq 
\left( \displaystyle\frac{ \lambda R_1 R_3 \left( \omega \right)}{1 - \lambda} \right) 
\left( \rho^{1 - \lambda} \left| t- t_e \right|^{\lambda - 1} - 1 \right) . 
\end{split}
\end{equation*}


When $\lambda > 1$, it follows that we see that (\ref{lamb pos theta expl key term 1}) has the 
same behavior as 
\begin{equation}\label{lamb > 1 reduced limt}
\displaystyle\lim_{t \uparrow t_e} \left[
 \left| t - t_e \right|^{\lambda - 1}  
 \displaystyle\int_{t_e - \rho }^t \displaystyle\frac{ \tilde{W}_s \left( \omega \right)  - \tilde{W}_t \left( \omega \right) }
{ \left| s - t_e \right|^{1+ \lambda} }\, ds 
\right]
\end{equation}
(all other terms tend to 0 $\tilde{P}$-a.s.). 
Using integration by parts,
\begin{align}\label{int by parts lamb > 1}
&\left| t - t_e \right|^{\lambda - 1}    \displaystyle\int_{t_e - \rho }^t \displaystyle\frac{ \tilde{W}_s \left( \omega \right)  - \tilde{W}_t \left( \omega \right) }
{ \left| s - t_e \right|^{1+ \lambda} }\, ds  \notag \\
&\qquad \sim
\mathcal{N} \left( 0 ,  \left| t - t_e \right|^{2\lambda - 2}  \displaystyle\int_{t_e - \rho }^t  \left( 
\displaystyle\frac{\left| s - t_e \right|^{-\lambda}}{\lambda} - \displaystyle\frac{1}{ \lambda \rho^{\lambda}}
\right)^2 ds \right) . 
\end{align}
Asymptotically, the variance in (\ref{int by parts lamb > 1}) grows like $\left| t - t_e \right|^{-1}$ as $t \uparrow t_e$, completing
the proof of Theorem~\ref{no Inv expl lem SS FIN}.  
\qed

{\small

\bibliographystyle{siam}  
\bibliography{munkThesisBib0502} }

\end{document}